\documentclass[12pt]{article}
\usepackage[dvipsnames,svgnames,x11names,hyperref]{xcolor}
\usepackage{amsmath}
\usepackage{amsthm}
\usepackage[T1]{fontenc}

\usepackage{fullpage}
\usepackage{thmtools}
\usepackage[colorlinks]{hyperref}
\hypersetup{
	linkcolor=[rgb]{0,0,0.4},
	citecolor=[rgb]{0, 0.4, 0},
	urlcolor=[rgb]{0.6, 0, 0}
}
\usepackage{amssymb,amsfonts,amsmath,amsthm,amscd,dsfont,mathrsfs}
\usepackage{complexity}
\usepackage{tikz}
\usetikzlibrary{decorations.pathreplacing,backgrounds}
\usepackage{multirow}
\usepackage{mathpazo}
\usepackage{braket}

\usepackage{algorithmicx}
\usepackage{algorithm} 
\usepackage[noend]{algpseudocode}

\algrenewcommand\algorithmicindent{1.0em}%

\usepackage{amsthm}
\usepackage{thmtools,thm-restate}

\numberwithin{equation}{section}
\declaretheoremstyle[bodyfont=\it,qed=\qedsymbol]{noproofstyle}

\declaretheorem[numberlike=equation]{observation}

\declaretheorem[name=Observation,numbered=no]{observation*}

\declaretheorem[numberlike=equation]{theorem}

\declaretheorem[name=Theorem,numbered=no]{theorem*}

\declaretheorem[numberlike=equation]{lemma}
\declaretheorem[name=Lemma,numbered=no]{lemma*}

\declaretheorem[numberlike=equation]{corollary}
\declaretheorem[name=Corollary,numbered=no]{corollary*}

\declaretheorem[name=Proposition,numbered=no]{proposition*}

\declaretheorem[numberlike=equation]{claim}
\declaretheorem[name=Claim,numbered=no]{claim*}

\declaretheorem[name=Conjecture,numbered=no]{conjecture*}

\declaretheorem[name=Question,numbered=no]{question*}

\declaretheoremstyle[bodyfont=\it,qed=$\lozenge$]{defstyle} 

\declaretheorem[numberlike=equation,style=defstyle]{definition}
\declaretheorem[unnumbered,name=Definition,style=defstyle]{definition*}

\declaretheorem[unnumbered,name=Example,style=defstyle]{example*}

\declaretheorem[unnumbered,name=Notation=defstyle]{notation*}

\declaretheorem[unnumbered,name=Construction,style=defstyle]{construction*}

\declaretheorem[numberlike=equation,style=defstyle]{remark}
\declaretheorem[unnumbered,name=Remark,style=defstyle]{remark*}

\usepackage{nth}
\usepackage{intcalc}
\usepackage{etoolbox}
\usepackage{xstring}

\usepackage{ifpdf}
\ifpdf
\else
\usepackage[quadpoints=false]{hypdvips}
\fi

\newcommand{\ECCCurl}[2]{http://eccc.hpi-web.de/report/\ifnumcomp{#1}{>}{93}{19}{20}#1/#2/}

\newcommand{\shortECCC}[2]{\texttt{\href{http://eccc.hpi-web.de/report/\ifnumcomp{#1}{>}{93}{19}{20}#1/#2/}{eccc:TR#1-#2}}}

\newcommand{\parseECCC}[1]{
\StrSubstitute{#1}{TR}{}[\tmpstring]%
\IfSubStr{\tmpstring}{/}{ 
\StrBefore{\tmpstring}{/}[\ecccyear]%
\StrBehind{\tmpstring}{/}[\ecccreport]%
}{
\StrBefore{\tmpstring}{-}[\ecccyear]%
\StrBehind{\tmpstring}{-}[\ecccreport]%
}%
\shortECCC{\ecccyear}{\ecccreport}}

\renewcommand{\set}[1]{\left\{ #1 \right\}}

\newcommand{\card}[1]{\left| #1 \right|}

\DeclareMathOperator{\Char}{\mathrm{char}}
\DeclareMathOperator{\Sys}{\mathrm{Sys}}

\renewcommand{\C}{\mathbb{C}}
\renewcommand{\R}{\mathbb{R}}

\newcommand{\N}{\mathbb{N}}
\newcommand{\F}{\mathbb{F}}
\newcommand{\Fn}{\F[x_1,\ldots,x_n]}

\newcommand{\SumPowSumk}{\Sigma^k \!\wedge\! \Sigma}

\newcommand{\SumProdSum}{\Sigma \Pi \Sigma}
\newcommand{\SumProdSumk}{\Sigma^k \Pi \Sigma}
\newcommand{\spsk}{\SumProdSumk}

\newcommand{\calC}{\mathcal{C}}
\newcommand{\calH}{\mathcal{H}}

\renewcommand{\vec}[1]{{\mathbf{#1}}}

\makeatletter
\newcommand{\va}{{\vec{a}}\@ifnextchar{^}{\!\:}{}}
\newcommand{\vb}{{\vec{b}}\@ifnextchar{^}{\!\:}{}}
\newcommand{\vc}{{\vec{c}}\@ifnextchar{^}{\!\:}{}}
\newcommand{\vd}{{\vec{d}}\@ifnextchar{^}{\!\:}{}}
\newcommand{\ve}{{\vec{e}}\@ifnextchar{^}{\!\:}{}}
\newcommand{\vy}{{\vec{y}}\@ifnextchar{^}{\!\:}{}}
\newcommand{\vs}{{\vec{s}}\@ifnextchar{^}{\!\:}{}}
\newcommand{\vt}{{\vec{t}}\@ifnextchar{^}{\!\:}{}}
\newcommand{\vx}{{\vec{x}}\@ifnextchar{^}{}{}}		
\newcommand{\vz}{{\vec{z}}\@ifnextchar{^}{\!\:}{}}
\newcommand{\vu}{{\vec{u}}\@ifnextchar{^}{\!\:}{}}	
\makeatother

\newcommand{\Sps}[1]{\Sigma^{#1} \Pi\Sigma}

\newcommand{\gb}[3]{{#1|_{#2,#3}}}

\newcommand{\fb}[2]{{\gb{f}{#1}{#2}}}

\newcommand{\rank}{\text{rank}}
\newcommand{\ranksyn}{\Delta_{\text{syn}}}
\newcommand{\dist}{\text{dist}}
\newcommand{\Lin}{\text{Lin}}

\DeclareMathOperator{\Span}{\text{span}}
\DeclareMathOperator{\simp}{\text{sim}}
\DeclareMathOperator{\codim}{\text{codim}}

\newcommand{\ranksem}{\Delta_{\text{sem}}}

\newcommand{\eqdef}{:=}

\newcommand{\cI}{\mathcal{I}}
\newcommand{\cJ}{\mathcal{J}}

\newcommand{\Res}{\mathrm{Res}}

\newcommand{\ignore}[1]{}

\title{Tensor Reconstruction Beyond Constant Rank\footnote{An extended abstract of this paper appeared in \cite{PelegSV24}.}}
\author{Shir Peleg\thanks{Blavatnik School of Computer Science, Tel Aviv University, Tel Aviv, Israel. Email: \texttt{shirpele@tauex.tau.ac.il}. The research leading to these results has received funding from the  Israel Science Foundation (grant number 514/20) and from the Len Blavatnik and the Blavatnik Family foundation. }
	\and
	Amir Shpilka\thanks{Blavatnik School of Computer Science, Tel Aviv University, Tel Aviv, Israel. Email: \texttt{shpilka@tauex.tau.ac.il}. The research leading to these results has received funding from the  Israel Science Foundation (grant number 514/20), the Len Blavatnik and the Blavatnik Family foundation and from the European Union (ERC, EACTP, 101142020). Views and opinions expressed are however those of the author(s) only and do not necessarily reflect those of the European Union or the European Research Council Executive Agency. Neither the European Union nor the granting authority can be held responsible for them.}
	\and 
	Ben Lee Volk\thanks{Efi Arazi School of Computer Science, Reichman University, Israel. Email: \texttt{benleevolk@gmail.com}. The research leading to these results has received funding from the  Israel Science Foundation (grant number 843/23). }}

\date{}

\begin{document}

	\maketitle
	
	\abstract{
		We give reconstruction algorithms for subclasses of depth-$3$ arithmetic
		circuits. In particular, we obtain the first efficient algorithm for 
		finding tensor rank, and an optimal tensor decomposition as a sum of rank-one tensors, when given black-box access to a tensor of super-constant rank.  Specifically, we obtain the following results:
		\begin{enumerate}
			\item A randomized algorithm that reconstructs polynomials computed by multilinear $\Sigma^{[k]}\prod^{[d]}\Sigma$ circuits in time $\poly(n,d,c) \cdot k^{k^{k^{k^{O(k)}}}}$, 
			\item \sloppy A randomized algorithm that reconstructs polynomials computed by  set-multilinear $\Sigma^{[k]}\prod^{[d]}\Sigma$ circuits in time $\poly(n,d,c) \cdot k^{k^{k^{k^{O(k)}}}}$,
		\end{enumerate}
		where $c=\log q$ if $\F=\F_q$ is a finite field, and $c$ equals the maximum bit complexity of any coefficient of $f$ if $\F$ is infinite.
		
		Prior to our work, polynomial time algorithms for the case when the rank, $k$, is constant, were given by Bhargava, Saraf and Volkovich \cite{BSV21}. 
		
		Another contribution of this work is correcting an error from a paper of Karnin and Shpilka \cite{KarninShpilka09} (with some loss in parameters) that also affected Theorem 1.6 of \cite{BSV21}. Consequently, the  results of \cite{KarninShpilka09,BSV21} continue to hold, with a slightly worse setting of parameters. For fixing the error we systematically study the relation between syntactic  and semantic notions of rank of $\SumProdSum$ circuits, and the corresponding partitions of such circuits. 
		
		We obtain our improved running time by introducing a technique for learning rank preserving coordinate-subspaces. Both \cite{KarninShpilka09} and \cite{BSV21} tried all choices of finding the ``correct'' coordinates, which, due to the size of the set, led to having a fast growing function of $k$ at the exponent of $n$. We manage to find these spaces in time that is still growing fast with $k$, yet it is only a fixed polynomial in $n$. 
		
	}
	
	\newpage
	
	\tableofcontents
	\newpage
	
	\section{Introduction}
	
	Reconstruction of algebraic circuits is a natural algorithmic problem that asks, given a black box access to a polynomial $f$ from some circuit class $\mathcal{C}$, to efficiently output an algebraic circuit computing $f$. Algebraic circuits are computational devices that compute multivariate polynomials using basic arithmetic operations, much like boolean circuits compute boolean functions using boolean bit operations. Thus, the reconstruction problem is a natural algebraic analog for well studied boolean learning problems \cite{Bshouty13}.
	
	It is often desired that the output of the algorithm will also be a circuit from the class $\mathcal{C}$ (which is called \emph{proper learning}). Requiring the learning algorithm to be efficient imposes an obvious upper bound on the size of the output, but it is also desirable to output a circuit as small as possible, ideally the smallest possible circuit from the class $\mathcal{C}$ that computes $f$.
	
	Reconstruction, however, is also a hard algorithmic problem. Results such as the $\NP$-hardness of computing or even approximating tensor rank \cite{Hastad90, Shitov16, BIJL18, Swernofsky18} force us to carefully manage our expectations regarding what's possible to compute efficiently, since it turns out that even for weak classes $\mathcal{C}$ (such as depth-$3$ set-multilinear circuits) it's unlikely to find an efficient algorithm that outputs the smallest possible circuit. Furthermore, reconstruction appears to be an even harder problem than black box Polynomial Identity Testing (PIT), the problem of determining whether the black box $f$ computes the identically zero polynomial. While PIT can be efficiently solved using randomness, efficient deterministic algorithms are known only for a handful of restricted circuit classes (we note, though, that in the reconstruction problem even giving a randomized algorithm is a non-trivial task). For a survey on algebraic circuits, PIT and reconstruction, see \cite{SY10}.
	
	Nevertheless, for some restricted classes, or when the constraints are sufficiently relaxed, it is possible to give many non-trivial efficient reconstruction algorithms. For example, many works have dealt with \emph{random} algebraic circuits (see, e.g., \cite{GKL11, KNS19, GKS20}, among others). In this setting, we think of the black box as being chosen randomly from the class $\mathcal{C}$ under some natural distribution on circuits from $\mathcal{C}$, and we require the algorithm to reconstruct $f$ with high probability over the chosen circuit (and perhaps over the random coins of the algorithm as well). Random circuits often avoid the degeneracies and pathologies that are associated with the clever cancellations that facilitate sophisticated algebraic algorithms, and are thus easier to handle and argue about.
	
	Another line of study, more relevant to our work, has to do with reconstruction of small depth algebraic circuits. The simplest non-trivial case is depth-$2$ circuits, for which the reconstruction problem is pretty well understood and can be done efficiently \cite{BOT88, KS01}. However, even slightly larger depths, like depth-$3$ and depth-$4$ circuits, already pose a much greater challenge. This is perhaps explained by a sequence of depth reduction results \cite{AV08, K12b, T15, GKKS16} that show that low depth circuits are expressive enough to  non-trivially simulate any algebraic circuit of polynomial size (and arbitrary depth). Thus, most attention has focused on restricted classes of depth-$3$ and depth-$4$ circuits \cite{KarninShpilka09, GKL12, Sinha16, Sinha22, BSV20, BSV21}.
	
	\subsection{Circuit Classes}
	
	A depth-$3$ circuit with top fan-in (that is, the in-degree of the top sum gate) $k$ computes a polynomial of the form $\sum_{i=1}^k \prod_{j=1}^{d_i} \ell_{i,j} (\vx)$, where each $\ell_{i,j}$ is a linear function in the input variables $\vx$. We denote this class $\SumProdSumk$. When $k$ is constant, this is a subclass of general depth-$3$ circuits that has been extensively studied (see Section 4.6 of \cite{SY10}).
	
	The circuit is called \emph{multilinear} if every gate in the circuit computes a multilinear polynomial. An even stronger restriction is \emph{set-multilinearity}. A polynomial $f$ is set-multilinear if the set of variables $\vx$ can be partitioned into disjoint sets $\vx_1, \ldots, \vx_d$ such that every monomial appearing in $f$ is a product of variables $x_{1,i_1} x_{2,i_2} \cdots x_{d,i_d}$ such that $x_{j,i_j}$ is in $\vx_j$.
	That is, a degree-$d$ set-multilinear polynomial is simply a $d$-dimensional tensor. Depth-$3$ set-multilinear circuits, which are circuits in which every gate computes a multilinear polynomial, are a natural model for computing tensors. Each product of linear functions $\prod_{j=1}^d \ell_j (\vx_j)$ corresponds to a rank one tensor, and thus we see that $f$ can be computed by a set-multilinear circuit of top fan-in $k$ if and only if its rank is at most $k$.
	
	An even more restricted model is \emph{depth-$3$ powering circuits}. In this model,  multiplication gates are replaced by powering gates. Such gates get as input a single linear function and their output is that function raised to some power. We denote the class of depth-$3$ powering circuit by $\SumPowSumk$. This is a natural computational model for computing \emph{symmetric tensors}, where again the top fan-in corresponds to the rank.
	
	Karnin and Shpilka \cite{KarninShpilka09} presented polynomial time reconstruction algorithms for $\SumPowSumk$ multilinear circuits for $k=O(1)$ over fields of size at most polynomial in $n$.
	Over fields of large characteristic or characteristic 0, Kayal \cite{Kayal12} has designed a randomized, proper, polynomial time reconstruction algorithm for $\SumPowSumk$ circuits assuming $2k< d$.\footnote{An earlier version of this paper contained a reconstruction algorithm for $\SumPowSumk$ circuits whose running time is polynomial in $n$ as long as $k$ is very slightly super-constant in $n$. However, after the initial version of this paper was posted, we became aware of Kayal's results which supersede ours, and thus we chose to remove this algorithm from this paper.}
	
	More recently, Bhargava, Saraf, and Volkovich \cite{BSV21} presented proper reconstruction algorithms for the circuit models discussed above. The running times of their algorithms are polynomial in $n$, the number of variables, and the degree $d$, assuming $k$ is constant, but not when $k$ is any growing function of $n$ or $d$. Unlike \cite{KarninShpilka09}, their algorithms work even over infinite fields. The exact running time is a polynomial whose exponent is a somewhat complicated expression that involves some quickly growing function of $k$. We describe their results more precisely vis-\`{a}-vis our results in \autoref{sec:results}.
	
	In particular, given a constant upper bound on the rank, they obtain efficient algorithms that given a tensor (or a symmetric tensor) can exactly compute its rank, and also obtain a decomposition as a sum of rank-one tensors. Since for large enough ranks the problem of computing the tensor rank becomes $\NP$-hard, it's natural to wonder at which point the intractability kicks in. That is, is there an efficient polynomial time algorithm that can compute the rank and obtain a decomposition even when the upper bound $k$ is super-constant?
	
	In this paper we obtain faster algorithms that remain polynomial time algorithms (in $n$ and $d$) even when $k$ is slightly super-constant. Our running times are of the form $\poly(n,d,T(k))$ where $T$ is some quickly growing function of $k$. Like the algorithms of Karnin and Shpilka \cite{KarninShpilka09} and Bhargava, Saraf, and Volkovich \cite{BSV21}, our learning algorithms are proper and return the smallest possible representation of $f$ in the relevant circuit model. In particular, they imply efficient randomized  algorithms for computing tensor rank even when the rank is slightly super-constant.
	
	Another contribution of this work is correcting an error that appeared in previous work. This error originated in \cite{KarninShpilka09} and affected Theorem 1.6 of \cite{BSV21} as well. Explaining the nature of the error requires some technical details that we present in \autoref{sec:errors}. Our correction recovers the affected results of \cite{KarninShpilka09,BSV21}, albeit with a slight change in the parameters that implies a somewhat worse dependence on the parameter $k$.

	Our algorithms require the field $\F$ to be large enough. The precise meaning of what ``large enough'' means depends on each case. The largeness assumption can always be guaranteed without loss of generality by considering field extensions, if necessary (in which case the output will also be a circuit over the extension field). In certain cases, we also assume that the characteristic of the field is large enough.
	
	\subsection{Our Results}
	\label{sec:results}
	
%
%
	
	We  provide reconstruction algorithms for multilinear depth-$3$ circuits with top fan-in $k$.
	
	\begin{theorem}
		\label{thm:intro:sumprodsumk}
		Let $\F$ be a field of characteristic zero, or a finite field of size at least $n^{k^{k^{O(k)}}}$ and characteristic greater than $d$.
		There exists a randomized algorithm that, given a black box access to a polynomial $f$ with $n$ variables and degree $d$, which is computed by a $\SumProdSumk$ multilinear circuit over a field $\F$ reconstructs $f$ in time $\poly(n,d,c) \cdot k^{k^{k^{k^{O(k)}}}}$, where $c=\log q$ if $\F=\F_q$ is a finite field, and $c$ equals the maximum bit complexity of any coefficient of $f$ if $\F$ is infinite.
	\end{theorem}
	
	Note that the algorithm in \autoref{thm:intro:sumprodsumk} runs in polynomial time for small enough (but super-constant) $k$, whereas the corresponding algorithm of \cite{BSV21} had running time of roughly $n^{T(k)}$ for some quickly growing function $T(k)$.

	We also present a reconstruction algorithm for set-multilinear depth-$3$ circuits. 
	Note that even though this class is a subclass of the previous model of multilinear circuits, as long as we insist on proper learning, reconstruction algorithms for a more general class don't imply reconstruction algorithms for its subclasses.
	
	\begin{theorem}
		\label{thm:intro:set-multilinear}
		Let $\F$ be a field of characteristic zero, or a finite field of size  at least $n^{k^{k^{O(k)}}}$ and characteristic greater than $d$.
		There exists a randomized algorithm that, given a black box access to a polynomial $f(\vx_1, \ldots \vx_d)$ such that $|\vx_i| \le n$  for every $i \in [d]$, such that $f$ is computed by a depth-$3$ set-multilinear circuit with top fan-in $k$ over $\F$, reconstructs $f$ in time $\poly(n,d,c) \cdot k^{k^{k^{k^{O(k)}}}}$, where $c=\log q$ if $\F=\F_q$ is a finite field, and $c$ equals the maximum bit complexity of any coefficient of $f$ if $\F$ is infinite.
	\end{theorem}
	
	\paragraph*{Follow-up Work}
	
	Unlike the algorithms from \cite{BSV21}, the algorithms we give in \autoref{thm:intro:sumprodsumk} and \autoref{thm:intro:set-multilinear} are randomized, even over $\R$ or $\C$. Following our work, Bhargava and Shringi \cite{BS24} obtained a deterministic analogue of \autoref{thm:intro:set-multilinear}, that also greatly improves the dependence on $k$ to singly-exponential. Derandomizing the algorithm from \autoref{thm:intro:sumprodsumk} remains an interesting open problem.

	\subsection{Proof Technique}
	\label{sec:technique}
	
	While our proof follows the general outline of the proofs in \cite{KarninShpilka09,BSV21}, improving the running time and correcting the errors (as explained in \autoref{sec:errors}) requires significant changes in parts of the argument.
	
	There are two main factors contributing to the doubly or triply exponential dependence on $k$ in the time complexity of the algorithms in \cite{BSV21}.
	
	The first is the fact that their algorithms solve systems of polynomial equations. This is required in order to find brute force solutions for the reconstruction problem over various projections of $f$ to a few variables, making the number of variables in the polynomial system of equations rather small (that is, only a function of $k$, and not of $n$). They then calculate the running time using the best known algorithms for solving such systems of polynomial equations. The exact running time depends on the field, and it is typically singly or doubly exponential time in the number of variables.
	
	Our main observation is that in all of these cases, it is also possible to modify the algorithms so that the \emph{degree} of the polynomial system of equations and the \emph{number of equations} are also only functions of $k$ (and not of $n$ or $d$, the number of variables and degree of the original polynomial $f$).
	
	The second reason their algorithms run in time $n^{T(k)}$ is a construction of an object called ``rank preserving subspace'', introduced in \cite{KarninShpilka09}, which is a subset of the coordinates that preserves certain properties of the polynomial, as we explain in \autoref{sec:technique-ml}. The dimension of this subspace depends on $k$, but finding it involves enumerating over all possible subsets of coordinates of the relevant size. As we soon explain, overcoming this difficulty requires a substantial amount of work.

%
	
	\subsubsection{Multilinear and Set-Multilinear $\SumProdSumk$ Circuits}\label{sec:technique-ml}
	
	The proofs of \autoref{thm:intro:sumprodsumk} and \autoref{thm:intro:set-multilinear} can be broken down to two parts, the first handles low degree polynomials and the second high degree polynomials. The analysis of both parts in \cite{BSV21} incurs factors of the form $n^{T(k)}$, which we would like to eliminate. While the proof of the low degree case follows the general outline of \cite{BSV21}, the proof of the high degree case is significantly more challenging and requires new ideas. As we describe later, the proof of the high degree case in \cite{BSV21}, for multilinear $\Sps{k}$ circuits, contains an error originating in \cite{KarninShpilka09}. We are able to correct this error (see \autoref{thm:unique-syn-part}), but even this correction doesn't suffice for improving the running time and a new approach is needed. Their result for set-multilinear circuits was not affected by this error as they use a different proof technique in the high degree case.

	\paragraph*{The low degree case:} when the degree $d$ is small, the number of linear functions in the circuit, which  is bounded by $kd$, is also small so that we can allow ourselves to try and learn the circuit for $f$ in an almost brute-force manner by solving a system of polynomial equations. Following \cite{BSV21} we first find, in polynomial time, an invertible linear transformation $A$ so that $g:=f(A\vx)$ depends  on a few \emph{variables} (and not merely linear functions). We then obtain a low-degree polynomial in a small number of variables, so that we can allow ourselves to learn the new circuit by solving a set of polynomial equations whose variables are the coefficients of the purported small circuit.
	
	We then wish to output $g(A^{-1}\vx)$. The problem is that this circuit may not be multilinear. To solve this, \cite{BSV21} introduce an additional set of $\approx \poly(n)$ low degree polynomial equations to  guarantee that the output circuit is multilinear. This results in a running time of about $n^{T(k)}$ for this part alone. We observe however that this set is highly redundant in the sense that, by dimension arguments, many of these equations are linearly dependent. By finding a basis to the polynomial system of equations and solving that basis alone, we're able to reduce the running time to $n \cdot T'(k)$ for some different function $T'(k)$.
	
	Our algorithm for low-degree set-multilinear circuits is very similar but a bit simpler. By slightly tweaking the polynomial system of equations that describes the circuit, we can learn $f(A \vx)$ as a set-multilinear circuit. Further, in this case it's possible to find $A$ such that $g(A^{-1} \vx)$ will automatically be set-multilinear, so that the challenge described in the previous paragraph doesn't exist in this setting.
	
	\paragraph*{The high degree case:} this is the more complicated and tedious part of the argument. We start by explaining the high level approach of \cite{KarninShpilka09} and \cite{BSV21}.
	
	The \emph{(syntactic) rank} of a $\Sps{k}$ circuit is defined to be the dimension of the span of the linear functions appearing in its multiplication gates, after factoring out the greatest common divisors of these gates (that is, the linear functions appearing in all of them). For more details see \autoref{def:syn-rank-dist}. This is a well studied notion originating in the work of Dvir and Shpilka \cite{DS07} and used in many later works \cite{KarninShpilka08, KarninShpilka09, SaxenaS11, SaxenaS13, KS09a}. The rank function allows one to define the \emph{distance} between two circuits $C_1$ and $C_2$ as the rank of their sum.
	
	The algorithm of \cite{BSV21} for learning multilinear $\Sps{k}$ circuits relies on a structural property of such circuits claimed by Karnin and Shpilka \cite{KarninShpilka09}. Karnin and Shpilka \cite{KarninShpilka09} partition the $k$ multiplication gates in the circuit to \emph{clusters}, so that each cluster has a low rank and each two distinct clusters have a large distance. In \cite{KarninShpilka09}, it is claimed that for some choice of parameters, this partition is unique and depends only on the polynomial computed by the circuit and not on the circuit itself (this is where the error is, and this is what is being corrected in \autoref{thm:unique-syn-part}). Thus, the authors of \cite{BSV21} try to obtain black box access to each of the clusters. Then, factoring out their greatest common divisors they can reconstruct them as, by multilinearity, the remaining part is a low degree polynomial.
	
	Obtaining black box access to the clusters is most of the technical work in the proof of Theorem 1.6 of \cite{BSV21}. At a high level, using their uniqueness result, Karnin and Shpilka \cite{KarninShpilka09} claimed to prove the existence of a small ``rank preserving subset'' of the variables $B$, such that after randomly fixing the variables outside of $B$, the remaining circuit $C|_B$ has the property that its clusters are in one-to-one correspondence with the original clusters restricted to $B$. The circuit $C|_B$ can again be reconstructed using the low-degree case, as it only involves a small number of variables, and thus we can get direct access to its clusters. Using a clever algorithm, Bhargava, Saraf and Volkovich \cite{BSV21} are able to obtain evaluations of the original clusters using evaluations of the restricted clusters.
	
	Regardless of the correctness issue that we discuss soon, a big bottleneck of this argument is that one needs to iterate over all subsets $B$ of $[n]$ up to a certain size bound (that depends only on $k$). Clearly such a procedure requires running time of the form $n^{T(k)}$. 
		
	Thus, we would like to obtain an algorithm that explicitly constructs a set $B$. One natural approach is to start with the empty set and add one variable at a time. This can be done by reconstructing the polynomial $f$ restricted to the current set $B$, and its clusters, and checking whether adding a variable to $B$ changes one of the parameters. If so then we add the variable and repeat the process. We continue as long as either the number of clusters or the rank of a cluster increases. The challenge with this approach is that the uniqueness guarantee of \autoref{thm:unique-syn-part} does not suffice. Note that if $f$ has a $\Sps{k}$ circuit $C$, $f$ restricted to $B$ as a natural $\Sps{k}$ circuit obtained by restricting the circuit $C$ to $B$. However, our low-degree algorithm learns \emph{some}, and potentially different, $\Sps{k}$ circuit that computes the restriction of $f$ to $B$. While \autoref{thm:unique-syn-part} guarantees both circuits would have the same number of clusters, computing the same polynomials, we don't have the guarantee that the \emph{rank} of each cluster is  the same in the different circuits, as the rank may depend on the circuit. Thus, as we gradually increase $B$, it seems hard to compare rank of clusters between different representations.
	
	To circumvent that we introduce \emph{semantic} versions of ranks and distances, which are properties of a \emph{polynomial} and not of a circuit computing it. We then develop a theory that studies the semantic and syntactic notions of rank, and the relations between them. In fact, our version of the semantic rank was already introduced by Karnin and Shpilka in \cite{KarninShpilka08, KarninShpilka09}, in the context of learning so-called $\SumProdSum(k,d,\rho)$ circuits. These are a generalization of $\SumProdSumk$ circuits of degree $d$ that allow every multiplication gate to also multiply an arbitrary polynomial that depends on at most $\rho$ linear functions (see \autoref{def:k-d-rho}).  Their notion of ``rank'' for such circuits is a certain hybrid between syntactic and semantic rank. Since in our case the distinction is important, we try to mention explicitly whether we mean syntactic or semantic rank.
	
	\subsubsection{Semantic Notions of Rank}
	\label{sec:intro-semantic}
	
	The semantic rank of a polynomial $f$ is defined as follows: first write $f = \prod_i \ell_i \cdot h$, where the $\ell_i$'s are linear functions and $h$ has no linear factors. Then define the semantic rank of $f$ to be the minimal number $r$ such that $h$ depends on $r$ linear functions.
	
	This number is well defined and doesn't depend on any representation of $f$ as a $\Sps{k}$ circuit. Working with semantic rank has advantages and disadvantages. On the one hand, it is now possible to prove stronger uniqueness properties regarding the clusters, since if two clusters compute the same polynomial then they also have the same rank. Indeed we prove such a uniqueness statement for some parameters. On the other hand, analyzing the semantic rank and its behavior under various operations (such as restricting the circuit to a subset of the variables, or increasing the set $B$ using the approach mentioned above) is significantly more difficult. Thus, we also prove various connections between semantic and syntactic ranks and we are able to show that if $f$ is computed by an $\Sps{k}$ circuit $C$, then the semantic and syntactic ranks of $C$ are not too far apart. 
	
	Recall that our main challenge is to explicitly construct a cluster-preserving subset $B$ of the variables, whose existence for syntactic ranks was proved by \cite{KarninShpilka09} (see \autoref{sec:errors} for a discussion of this result). In the context of semantic rank, proving such an analogous statement is significantly more challenging. In fact, while the proof of \cite{KarninShpilka09} is existential (and then the algorithm of \cite{BSV21} essentially enumerates over all possible subspaces), our proof is algorithmic (see \autoref{algo:finding-B}).
	In essence, our algorithm follows the outline described above: it starts with the empty set and on each iteration adds a few variables to $B$ until the cluster structure ``stabilizes'', i.e., their number and their ranks stay the same. Proving that this algorithm works requires a significant amount of technical work.
	
	\subsubsection{The Errors in Previous Work and Our Corrections}
	\label{sec:errors}
	
	Explaining the nature of the erroneous statements appearing in \cite{BSV21, KarninShpilka09} requires giving some more technical details.
	
	As mentioned earlier, one of the main components in the reconstruction algorithm for multilinear $\Sps{k}$ circuits given in \cite{BSV21} is the uniqueness of clusters property for such circuits, which is claimed by Karnin and Shpilka \cite{KarninShpilka09}. Note that the rank and distance measures for $\Sps{k}$ circuits are \emph{syntactic} and inherently tied to a circuit. Karnin and Shpilka \cite{KarninShpilka09} define a \emph{clustering} algorithm that, given a circuit, partitions the $k$ multiplication gates into several sets such that the rank of the subcircuit corresponding to each set is small, and the distance between every pair of subcircuits is large.
	
	In Corollary 6.8 of \cite{BSV21} it is claimed, based on \cite{KarninShpilka09}, that these clusters are unique, even among different circuits that compute the same polynomial. That is, if $C$ and $C'$ are two circuits computing the same polynomial $f$, the clustering algorithm of \cite{KarninShpilka09} would return the same clusters (perhaps up to a permutation). Such a claim can indeed be read from Theorem 5.3 of \cite{KarninShpilka09}. However, in our judgment, the paper \cite{KarninShpilka09} does not contain a valid mathematical proof for such a statement.
	
	Karnin and Shpilka associate with each partition into clusters two parameters, $\kappa$ and $r$. The parameter $r$ upper bounds the rank of each cluster, and the parameter $\kappa$ controls the distance: their clustering algorithm guarantees that each pair of clusters has distance at least $\kappa r$. Consequently, their clustering algorithm receives $\kappa$ as an additional input, and outputs a clustering with parameters $\kappa$ and $r$ for some value of $r$ that can be upper bounded as a function of $\kappa$ and $k$.
	
	The proof of Theorem 5.3 of \cite{KarninShpilka09} assumes without justification that, given two different circuits $C$ and $C'$ computing the same polynomial, the clustering algorithm with parameters $\kappa$ would return partitions with the same value of the parameter $r$, which is crucially used in their proof.

	In this work we provide a corrected proof of Theorem 5.3 of \cite{KarninShpilka09} (\autoref{thm:unique-syn-part}). While the corrected version is not identical to the original statement word-for-word (as our parameter $\kappa$ is much larger than originally stated as a function of $k$), it suffices for fixing the arguments in \cite{KarninShpilka09} and \cite{BSV21}, with the straightforward corresponding changes in parameters throughout.
	
	We wish to stress again that Theorems 1.1 and 1.4 of \cite{BSV21}, that give algorithms for learning depth-$3$ set-multilinear and depth-$3$ powering circuits, respectively, are not affected by the error in \cite{KarninShpilka09}.
	
	\subsection{Open Problems}
	\label{sec:open}
	
	One natural problem our work raises is the question of how large the top fan-in $k$ needs to be before the reconstruction problem becomes intractable. The $\NP$-hardness results for tensor rank imply that clearly when $k=\poly(n)$ we shouldn't expect to find exact proper learning algorithm, whereas we show that the intractability barrier is not at the regime when $k$ is constant. It remains an interesting problem to bridge the gap.

	Another interesting problem is derandomizing our algorithm from \autoref{thm:intro:sumprodsumk}. In \autoref{rem:no-derand} we explain why we cannot derandomize our algorithm with the same improved running time.

	\section{Preliminaries}
	
	The following notation will be very useful throughout our paper.
	
	\begin{definition}\label{def:restriction}
		For $\va\in \F^n$, $B\subseteq [n]$, and a polynomial $f\in \F[x_1,\ldots,x_n]$ we define $\fb{B}{\va}$ the polynomial obtained by fixing $x_j =  \va_j$ for every $j\notin B$.
	\end{definition}
	
	\subsection{Black Box Access to Directional Derivatives}
	\label{sec:derivatives}
	
	\begin{lemma}
		\label{lem:black-box-derivative-variable}
		Let $\F$ be a field of size at least $d+1$ and let $f(x_1, \ldots, x_n) \in \F[x_1, \ldots, x_n]$ be a polynomial of degree $d$. Given a black box access to $f$, for every $e \le d$ and for every variable $x \in \set{x_1,\ldots,x_n}$, we can simulate a black box access to $g := \partial^e f / \partial x^e$ using at most $d+1$ queries to $f$.
	\end{lemma}
	
	\begin{proof}
		Without loss of generality assume $x=x_1$. 
		Write $f(x) = \sum_{i=0}^d f_i (x_2, \ldots, x_n) \cdot x_1^i$, so that 
		\begin{equation}
			\label{eq:derivative-formula}
			g = \frac{\partial^e f}{\partial x_1^e} = \sum_{i=e}^{d} \left( \prod_{t=0}^{e-1} (i-t) \right) f_i (x_2, \ldots, x_n) x_1^{i-e}.
		\end{equation}
		Pick arbitrary distinct $\alpha_1, \ldots, \alpha_{d+1} \in \F$. Since $P(X) := \sum_{i=0}^{d} f_i (x_2, \ldots, x_n) X^i$ is a univariate polynomial in $X$ of degree $d$, by standard polynomial interpolation there are (efficiently computable) elements $\beta_{i,j} \in \F$, where $i,j \in \set{0,\ldots, d}$, such that for every $i \in \set{0,\ldots, d}$,
		\begin{equation}
			\label{eq:derivative-interpolation}
			f_i (x_2, \ldots, x_n) = \sum_{j=0}^d \beta_{i,j} P(\alpha_j).
		\end{equation}
		
		Suppose now we would like to evaluate $g$ at the point $\vc=(c_1, \ldots, c_n)$. We first compute $f_i (c_2, \ldots, c_n)$ for every $i \in \set{0,\ldots,d}$ using the relation \eqref{eq:derivative-interpolation}. This requires a total of $d+1$ black box evaluations of $f$, for evaluating $f(\alpha_j, c_2, \ldots, c_n) = P(\alpha_j)$ for every $j \in \set{0,\ldots, d}$. Given $f_0, \ldots, f_d$, we can evaluate $g$ at $\vc$ directly using the relation \eqref{eq:derivative-formula} by plugging in $x_1 = c_1$.
	\end{proof}
	
	\autoref{lem:black-box-derivative-variable} can be generalized to directional derivatives.
	
	\begin{lemma}
		\label{lem:black-box-derivative-directional}
		Let $\F$ be a field of size at least $d+1$ and let $f(x_1, \ldots, x_n) \in \F[x_1, \ldots, x_n]$ be a polynomial of degree $d$. Given a black box access to $f$, for every $e \le d$ and for every $0 \neq \vu \in \F^n$, we can simulate a black box access to $g := \frac{\partial^e f}{\partial \vu^e}$ using at most $d+1$ queries to $f$.
	\end{lemma}
	
	\begin{proof}
		Let $A$ be an invertible matrix such that $A \ve_1 = \vu$ (that is, the first column of $A$ is $\vu$). Further, define $f_A (\vx) = f(A\vx)$. Clearly, we can simulate black box access to $f_A$ using black box access to $f$. Further, note that by the chain rule, for every $\vc \in \F^n$
		\[
		\frac{\partial f_A}{\partial x_1} (\vc) = \sum_{i=1}^n \frac{\partial f}{\partial x_i} (A\vc) \cdot A_{i,1} = \sum_{i=1}^n u_i \frac{\partial f}{\partial x_i} (A\vc).
		\]
		Thus, in order to evaluate
		\[
		\frac{\partial f}{\partial \vu} (\vx) = \sum_{i=1}^n u_i \frac{\partial f}{\partial x_i} (\vx)
		\]
		at any point $\vc'$, we can evaluate $\frac{\partial f_A}{\partial x_1}$ at the point $A^{-1} \vc'$ using \autoref{lem:black-box-derivative-variable}.
		
		For higher order derivative, we use the same method, that is, we simulate $\partial^e f / \partial \vu^e$ using black box access to $\partial^e f_A / \partial x_1^e$ (as guaranteed by \autoref{lem:black-box-derivative-variable}), which is in turn simulated by black box access to $f$ itself.
	\end{proof}

	\subsection{Essential Variables}
	\label{sec:essential}
	
	Let $f$ be an $n$-variate polynomial. We say that $f$ depends on $m$ essential variables if there exists an invertible linear transformation $A$ such that $f(A\vx)$ depends on $m$ variables. An interesting fact is that it's possible, given a black box access to $f$, to compute a linear transformation $A$ such that $g := f(A\vx)$ depends only on $x_1, \ldots, x_m$.
	
	\begin{lemma}[\cite{Kayal11, Carlini06}]
		\label{lem:essential-vars}
		Let $f \in \F[\vx]$ be an $n$-variate polynomial of degree $d$ with $m$ essential variables, where $\Char(\F)=0$ or $\Char(\F) > d$. Suppose $f$ is computed by a circuit of size $s$. Then, there's an efficient randomized algorithm that, given black box access to $f$, runs in time $\poly(n,d,s)$ and computes an invertible linear transformation $A$ such that $f(A\vx)$ depends on the first $m$ variables $x_1, \ldots, x_m$.
	\end{lemma}

	Bhargava, Saraf and Volkovich \cite{BSV21} derandomize this lemma when $f$ is computed by a $\SumPowSumk$ circuit, a depth-$3$ set-multilinear circuit of top fan-in $k$ or a depth-$3$ multilinear circuit of top fan-in $k$. However the time required for their derandomization involves factors of $n^{O(k)}$ and thus we want to obtain an improved running time.
	
	For a class of polynomials $\calC$ defined over a field $\F$, we denote by $\Sigma^t \calC$ the class of polynomials of the form $\alpha_1 f_1 + \alpha_2 f_2 + \cdots + \alpha_t f_t$ with $\alpha_i \in \F$ and $f_i \in \calC$ for every $i$.
	
	\begin{lemma}
		\label{lem:hitting-set-dependency}
		Let $\calC$ be a class of polynomials and let $f_1, \ldots, f_t \in \calC$. Let $\calH$ be a hitting set for $\Sigma^t \calC$. Denote by $f_i|_{\calH}$ the vector (of length $|\calH|$) $(f_i(\beta))_{\beta \in \calH}$. Then for any $\alpha_1, \ldots, \alpha_t \in \F$,
		\[
		\sum_{i=1}^t \alpha_i f_i = 0 \iff \sum_{i=1}^t \alpha_i f_i |_{\calH} = 0.
		\]
		In particular, the polynomials $f_1, \ldots, f_t$ are linearly independent if and only if the vectors $f_1|_{\calH}, \ldots, f_t|_{\calH}$ are linearly independent.
	\end{lemma}
	
	\begin{proof}
		The implication from left to right is clear. In the other direction, $\sum_{i=1}^t \alpha_i f_i |_{\calH} = (\sum_{i=1}^t \alpha_i f_i) |_{\calH}$. Since $\sum_{i=1}^t \alpha_i f_i \in \Sigma^t \calC$ and $\calH$ is a hitting set, it follows that $\sum_{i=1}^t \alpha_i f_i  = 0$.
	\end{proof}
	
	\autoref{lem:hitting-set-dependency} gives an efficient way to test for dependency of polynomials assuming the existence of a small and efficiently constructible hitting set for $\Sigma^t \calC$.
	
	A derandomized version of \autoref{lem:essential-vars} is given below.
	
	\begin{lemma}
		\label{lem:det-essential-vars}
		Let $\calC$ be a class of polynomials closed under taking first order partial derivatives. Denote by $\calH$ a hitting set for $\Sigma^{t+1} \calC$. Then, there's a deterministic algorithm that, given a black box access to a degree-$d$ polynomial $f(\vx) \in \F[\vx]$ that has $t$ essential variables such that $f \in \calC$, runs in time $\poly(n,d,|\calH|)$ and outputs an invertible matrix $A \in \F^{n \times n}$ such that $f(A\vx)$ depends only on the first $t$ variables.
	\end{lemma}
	
	\begin{proof}
		As noted by \cite{BSV21} (in their Lemma 3.24), the only place in which randomness is used in the proof of \autoref{lem:essential-vars} is in finding a basis for the vector space
		\[
		V = \set{ \va \in \F^n : \sum_{i=1}^n a_i \frac{\partial f}{\partial x_i} = 0}.
		\]
		It turns out that $\codim V = t$ where $t$ is the number of essential variables of $f$, which equals $\dim \Span \left\{ \frac{\partial f}{\partial x_i} \right\}$. Let $f_i = \frac{\partial f}{\partial x_i}$. By assumption $f_i \in \calC$, and further using \autoref{lem:black-box-derivative-variable} we can obtain black box access to each $f_i$. Using \autoref{lem:hitting-set-dependency}, we can greedily pick $t$ linearly independent polynomials among $f_1, \ldots, f_n$, as well as, for each element $f_j$, compute the coefficients that express it as a linear combination of basis vectors (note that this requires applying \autoref{lem:hitting-set-dependency} on at most $t+1$ polynomials). That is, as in \cite{BSV21}, we compute a matrix $M \in \F^{n \times t}$ and indices $i_1, \ldots, i_t$ such that
		\[
		M \begin{pmatrix}
			f_{i_1} \\
			\vdots \\
			f_{i_t}
		\end{pmatrix}
		= 
		\begin{pmatrix}
			f_{1} \\
			f_{2} \\
			\vdots \\
			f_{n-1} \\
			f_{n}
		\end{pmatrix}
		\]
		The left kernel of $M$ is $V$.
	\end{proof}
	
	We note that the models we consider in this work are all closed under first order partial derivatives.
	
	\subsection{Hitting Sets for Depth-$3$ Circuits}
	\label{sec:hitting-sets}
	
	While there exist quasi-polynomial size hitting sets for depth-$3$ set-multilinear circuits \cite{FS13, FSS14, AGKS15}, we insist on obtaining polynomial size hitting sets for these models when $k$ is slightly super-constant (when $k$ is constant there are hitting sets of size $n^{\poly(k)}$ for general depth-$3$ circuits of top fan-in $k$, see, e.g., Section 4.6.2 of \cite{SY10} and \cite{SS12}).	
	Guo and Gurjar constructed such explicit polynomial size hitting sets for read-once algebraic branching programs (roABPs) of super-constant width. 
	
	\begin{theorem}[\cite{GG20}]
		\label{thm:guo-gurjar-hitting-set}
		There's an explicit hitting set of size $\poly(n,d)$ for the class of $n$-variate, individual degree $d$ polynomials computed by any-order roABPs of width $w$, assuming there's a constant $\varepsilon > 0$ such that $w = 2^{O(\log^{1-\varepsilon}(nd))}$.
	\end{theorem}

	\begin{corollary}
		\label{cor:poly-size-hitting-set-set-ml}
		There's an explicit hitting set of size $\poly(n,d)$ for the class of set-multilinear polynomials computed by  depth-$3$ set-multilinear circuits of degree $d$ and top fan-in $k$, assuming $k=2^{O(\log^{1-\varepsilon}(nd))}$ for some $\varepsilon >0$.
	\end{corollary}
	
	\begin{proof}
		Write $\vx = \vx_1 \cup \cdots \cup \vx_d$ for the set of variables of $f$. Note that by substituting $x_{i,j}$ by $y_i^j$ we obtain a polynomial $\tilde{f}$ in $\vy=(y_1, \ldots, y_d)$, of individual degree $n$, which is non-zero if and only if $f$ is non-zero. Further, $\tilde{f}$ is naturally computed by an any-order roABP of width $k$: we convert every multiplication gate to a path of width 1, and connect them in parallel. Thus we get a hitting set for $f$ of the required size.
	\end{proof}
	
%
%
%
%
	For general depth-$3$ circuits with top fan-in $k$, the known results are slightly weaker.
	
	\begin{lemma}[\cite{SS12}]
		\label{lem:PIT}
		There exists an explicit hitting set for the class of $n$-variate polynomials computed by multilinear $\SumProdSumk$ circuits of degree $d$ of size $n^{O(k^2 \log k)}$.
	\end{lemma}

	We remark that the hitting set presented in \cite{SY10} is of size $n^{O(R(k,r))}$ where $R(k,r)$ is the so-called rank bound for $\SumProdSumk$ circuits, which (for some fields, as explained in \cite{SY10}) depends on $d$. However for \emph{multilinear} circuits the above result is a corollary of Corollary 6.9 of \cite{DS07} and the rank bounds of \cite{SaxenaS13}.
	
	We further note that had we used \autoref{lem:PIT}, our algorithm wouldn't run in polynomial time for super-constant $k$, which is one of the reasons we use a randomized PIT algorithm for this class in our reconstruction algorithm. However, this is not the major obstacle for derandomization: derandomizing our algorithm in polynomial time would require a deterministic PIT for much larger classes than multilinear $\Sps{k}$ circuits. It's an interesting open problem to obtain a derandomization for our algorithm even modulo \autoref{lem:PIT}.

	\subsection{Solving a System of Polynomial Equations}
	\label{sec:sys}
	
	Let $\F$ be a field and let $\Sys_{\F}(n,m,d)$ denote the randomized time complexity of finding a solution to a polynomial system of $m$ equations in $n$ variables of degree $d$. A detailed analysis of this function for various fields $\F$ appears in Section 3.8 of the arXiv version of \cite{BSV21}. For our purposes, it is enough to note that for every field $\F$, $\Sys_{\F}(n,m,d) = \poly(nmd)^{n^{n}}$, if we allow solutions from an algebraic extension of $\F$. Further, for $\F=\R$, $\C$ or $\F_q$, extensions are not needed, and if $\F = \R$ or $\C$ then the algorithm is in fact deterministic.

	\subsection{Resultants}\label{sec:resultant}

	Let $f(x), g(x)$ be two polynomials of degrees $m$ and $\ell$ in the variable $x$, respectively. Suppose $m,\ell > 0$, and write
	
	\begin{align*}
		f(x) &= c_m x^m + c_{m-1}x^{m-1} + \cdots + c_0 \\
		g(x) &= d_\ell x^\ell + d_{\ell-1} x^{\ell-1} + \cdots + d_0.
	\end{align*}
	
	The \emph{Sylvester matrix} of the polynomials $f$ and $g$ with respect to the variable $x$ is the following $(m + \ell) \times (m+ \ell)$ matrix:
	
	\[
	\begin{pmatrix}
		c_m & & & &  d_\ell & & & \\
		c_{m-1} & c_m & &  & d_{\ell-1} & & \\
		c_{m-2} & c_{m-1} & \ddots & & d_{\ell-2} & d_{\ell-1} & \ddots & \\
		\vdots & & \ddots & c_{m} & \vdots & & \ddots & d_\ell \\
		& \vdots & & c_{m-1} & & \vdots & & d_{\ell-1} \\
		c_0 & & & & d_0 & & & \\
		& c_0 & & \vdots & & d_0 & & \vdots \\
		& & \ddots & & & & \ddots & \\
		& & & c_0 & & & & d_0
	\end{pmatrix}
	\]
	
	The determinant of this matrix is called the \emph{resultant} of $f$ and $g$ with respect to the variable $x$ and is denoted $\Res_x (f,g)$. 
	
	In our case we often think of $f, g \in \F[x_1, \ldots, x_n]$ interchangeably as $n$-variate polynomials or as univariate polynomials in some variable, say $x_1$, over the ring $\F[x_2, \ldots, x_n]$, in which case the resultant is a polynomial in $x_2, \ldots, x_n$. The main property of resultant we use is that, assuming the degree in $x_1$ of both $f$ and $g$ is positive, $f$ and $g$ have a common factor in $\F[x_2, \ldots, x_n]$ if and only if $\Res_{x_1} (f,g) = 0$ (see, e.g., Proposition 3 in Chapter 3, Section 6 of \cite{CLO07}).

	\section{Syntactic Rank of Depth-$3$ Circuits}\label{sec:syn}
	
	In the following two sections, we define syntactic and semantic notions of ranks of polynomials computed by $\Sps{k}$ circuits. Note that \emph{syntactic} ranks are inherently tied to \emph{circuits} computing the polynomials, whereas semantic ranks are independent of the representation or computation of the polynomials.

	For a circuit $C$ we denote by $[C]$ the polynomial computed by $C$.
	For two $\SumProdSumk$ circuits $C, C'$, we define their \emph{syntactic sum}, $C + C'$, to be the depth-$3$ circuit whose top gate sums all multiplication gates in $C$ and $C'$. Observe that $C + C'$ is a $\Sps{2k}$ circuit.
	
	We start by defining \emph{syntactic} notions of rank and distance for $\Sps{k}$ circuits.
	
	\begin{definition}[Syntactic Rank and Distance]
		\label{def:syn-rank-dist}
		Let $C = \sum_{i= 1}^k M_i = \sum_{i=1}^{k}\prod_{j=1}^{d_i}\ell_{i,j}$ be a $\Sps{k}$ circuit. Define the following notions:
		\begin{enumerate}
			\item $\deg(C) = \max \{ \deg [M_i] : 1 \le i \le k \}$.
			\item $\gcd(C)$ is the set of linear functions appearing in all of $M_1, \ldots, M_k$ (up to multiplication by a constant). I.e., $\gcd(C)=\gcd(M_1,\ldots,M_k)$.
			\item $\simp(C) \eqdef \frac{C}{\gcd(C)} = \sum_{i=1}^k \frac{M_i}{\gcd(C)}\in \Sps{k}$ is called the \emph{simplification} of $C$. $C$ is called \emph{simple} if $\gcd(C)=1$.
			\item We say that $C$ is \emph{minimal} if for every $\emptyset\neq S \subsetneq [k]$, $\sum_{i\in S}M_i\neq 0$.
			\item Let $\mathcal{L}_i$ be the collection of linear polynomials appearing in $\frac{M_i}{\gcd(C)}$, we define $\ranksyn(C) \eqdef \dim(\Span\set{\mathcal{L}_1,\ldots,\mathcal{L}_k})$. 
			\item Let $C'$ be a $\SumProdSumk$ circuit.  We define $\dist(C,C') = \ranksyn(C + C')$. $\hfill \qedhere$
		\end{enumerate}
	\end{definition}
	
	The usefulness of syntactic rank is expressed in the following well known \emph{rank bound} for multilinear depth-$3$ circuits.
	
	\begin{theorem}[\cite{DS07, KS09a, SaxenaS11, SaxenaS13}]
		\label{thm:rank bound ml}
		There's a monotone function $R(k,d)$ such that 
		any simple and minimal $\Sps{k}$ circuit $C$ that computes the zero polynomial and such that $\deg(C)\leq d$,  satisfies $\ranksyn(C) \le R(k,d)$. Further, $R(k,d) \le  4k^2 \log (2d)$.
		
		If $C$ is multilinear there's a similar function $R_M(k)$ depending only on $k$: any simple and minimal, multilinear $\Sps{k}$ circuit $C$, computing the zero polynomial satisfies $\ranksyn(C) \leq R_M(k)$. One can take $R_M(k) \le 10 k^2 \log k$.
	\end{theorem}

	The next lemma will be useful when studying different representations of the same polynomial.
	
	\begin{lemma}
		\label{lem:small-syn-rank-t}
		For $j \in [t]$, let
		\begin{align*}
			M_j = \sum_i M_{j,i}, & \quad T_j = \sum_i T_{j,i}
		\end{align*}
		be $\SumProdSumk$ circuits (with $M_{j,i}, T_{j,i}$ denoting multiplication gates).
		
		Suppose that for every $j \in [t]$, $M_j - T_j$ are minimal circuits with
		$\ranksyn(M_j - T_j) \le s$. Further, assume that
		$\ranksyn(\sum_j T_j) \le r$.
		Then, $\ranksyn(\sum_j M_j) \le t(r+2s)$.
	\end{lemma}
	
	\begin{proof}
		By factoring out the gcd of the circuits we obtain
		\[
		M_j - T_j = \left( \prod_i a^j_i \right) (M'_j + T'_j).
		\]
		where $\ranksyn(M'_j + T'_j) \le s$.	
		Further, by the assumption that $\ranksyn(\sum_j T_j) \le r$, we have
		\[
		\sum_j T_j = \left( \prod_{i} \ell_i\right) (\sum_j \tilde{T}_j).
		\]
		where $\ranksyn(\sum_j \tilde{T}_j) \le r$.
		Thus, for every $j \in [t]$,
		\begin{align*}
			T_j &= \left( \prod a^j_i \right) T'_j = \left(\prod \ell_i \right) \tilde{T}_j
		\end{align*}
		with $\ranksyn(T'_j) \le s$ and $\ranksyn(\tilde{T}_j) \le r$. Multilinearity implies $\deg(T'_j)\le s$ and $\deg(\tilde{T}_j)\le r$. 
		Hence, $|\set{\ell_i} \setminus \set{a^j_i}| \le s$.
		Thus, $\set{\ell_i} \cap \bigcap_{j=1}^t \set{a^j_i} = \set{\ell'_1, \ldots, \ell'_q}$, where $q \geq |\{\ell_i\}| - ts$. Denote $\{\tilde{\ell}_1,\ldots,\tilde{\ell}_u\} = \set{\ell_i}\setminus\set{\ell'_1, \ldots, \ell'_q}$, for $u= |\{\ell_i\}| -q \leq ts$.
		Hence,
		\[
		\sum_{j=1}^t M_j = \left( \prod_i \ell'_i \right) \cdot
		\sum_{j=1}^t \left( \frac{\prod a^j_i}{\prod \ell'_i} \cdot \left(M'_j \right) \right).
		\]
		To bound the syntactic rank of $\sum_{j=1}^t \left( \frac{\prod a^j_i}{\prod \ell'_i} \cdot \left(M'_j \right) \right)$ we observe that for every $j$, it holds that $\set{a^j_i} \setminus \set{\ell'_i} \subseteq \left(\set{a^j_i} \setminus \set{\ell_i}\right) \cup \{\tilde{\ell}_1,\ldots,\tilde{\ell}_u\}$. 
		Further, since $\ranksyn(M'_j)$ is at most $s$, and  $|\set{a^j_i} \setminus \set{\ell_i}| \le r$, we get that
		\[
		\ranksyn \left( \sum_{j=1}^t M_j \right) =\ranksyn\left(\sum_{j=1}^t \left( \frac{\prod a^j_i}{\prod \ell'_i} \cdot \left(M'_j \right) \right)\right)\le t r + ts + u \leq t(r+2s). \qedhere
		\]
	\end{proof}

	\subsection{Syntactic Partitions of $\Sps{k}$ Circuits}\label{sec:syn-part}

	In this section we study syntactic partitions of $\Sps{k}$ circuits. In \autoref{sec:sem-part} we shall discuss \emph{semantic} partitions and compare the two notions.
	
	\begin{definition}[Syntactic Partition, Definition 3.3 of \cite{KarninShpilka09}]
		\label{def:tau-r-syntactic}
		Let $C = \sum_{i=1}^k \prod_{j=1}^{d_i} \ell_{i,j} = \sum_{i=1}^k M_i$ be a $\SumProdSumk$ circuit.
		Let $I = \set{A_1, \ldots, A_s}$ be a partition of $[k]$. For each $i \in [s]$ let $C_i = \sum_{j \in A_i} M_j$. We say that $\set{C_i}_{i \in [s]}$ is a $(\tau, r)$-syntactic partition of $C$ if:
		\begin{itemize}
			\item For every $i \in [s]$, $\ranksyn(C_i) \le r$.
			\item For every $i \neq j \in [s]$, $\dist(C_i, C_j) \ge \tau r$. $\hfill \qedhere$
		\end{itemize}
	\end{definition}

	The following lemma captures an important property of the definition.

	\begin{lemma}
		\label{lem:large-rank-syn}
		Let $C$ be a $\SumProdSumk$ circuit. Let $(C_1, \ldots, C_s)$ be a $(\tau,r)$-syntactic partition of $C$ with $\tau \ge 10$.
		Let $M, T$ be two multiplication gates in $C$ that belong to different clusters.
		Then,
		\[
		\dist(M,T) > \tau r / 10.
		\]
	\end{lemma}

	\begin{proof}
		Assume without loss of generality that $M$ belongs to $C_1$ and $T$ belongs to $C_2$, i.e., $C_1 = M + \tilde{M}$ and $C_2 = T + \tilde{T}$.
		By pulling out the linear factors from each cluster, we write:
		\[
		C_1 = \left( \prod_{i=1}^{m_a} a_i \right) (M' + \tilde{M}'), \quad C_2 = \left( \prod_{i=1}^{m_b} b_i\right)  (T' + \tilde{T}').
		\]
		
		Further, assume towards contradiction that $\dist(M, T) \le \tau r / 10$, and write
		\[
		M + T = \left( \prod_{i=1}^{m_\ell} \ell_i\right) (\widehat{M} + \widehat{T}).
		\]
		Then, we have
		\begin{align*}
			M &= \left( \prod a_i \right) M' = \left(\prod \ell_i \right) \widehat{M}\\
			T &= \left( \prod b_i \right) T' = \left(\prod \ell_i \right) \widehat{T}
		\end{align*}
		
		Since $\ranksyn(C_1) \le r$, $\deg(M') \le r$ and similarly $\deg(T') \le r$. Further, by the assumption that $\ranksyn(M+T) \le \tau r / 10$, we get that $\deg(\widehat{M}) \le \tau r / 10$ and similarly $\deg(\widehat{T}) \le \tau r / 10$.
		
		Recall that $m_a = |\set{a_i}|$, $m_b = |\set{b_i}|$ and $m_\ell = |\set{\ell_i}|$. The above inequalities imply that
		$|\set{\ell_i} \cap \set{a_i}| \ge m_a- \tau r / 10$ and $|\set{\ell_i} \setminus \set{a_i}| \le r$. Similar inequalities holds for $\set{b_i}$.

		Denote $\set{a_i} \cap \set{b_i} = \set{\ell'_1, \ldots, \ell'_q}$. We have that
		\[
		q \ge |\set{a_i} \cap \set{b_i} \cap \set{\ell_i}| \ge |\set{a_i} \cap \set{\ell_i}| - |\set{\ell_i} \setminus \set{b_i}| \ge m_a - \tau r / 10 - r.
		\]
		Similarly, $q \ge m_b - \tau r / 10 - r$.
		Hence,
		\[
		C_1 + C_2 = \left( \prod_i \ell'_i \right) \cdot \left( \frac{\prod a_i}{\prod \ell'_i} \cdot \left(M' + \tilde{M}' \right) +
		\frac{\prod b_i}{\prod \ell'_i} \left(T' + \tilde{T}' \right)
		\right).
		\]
		The number of linearly independent linear forms in $ \frac{\prod a_i}{\prod \ell'_i}$ is at most $\tau r / 10 + r$, and similarly for $\frac{\prod b_i}{\prod \ell'_i}$.
		Further, since $M', \tilde{M}'$ are in the simplification of a cluster, their total rank is at most $r$, and similarly for $T', \tilde{T}'$. All of which goes to show that
		\[
		\tau r\leq \ranksyn(C_1, C_2) \le 2 \tau r / 10 + 4 r \leq 6 \tau r / 10,
		\]
		where in the first inequality we used the assumption that $C_1$ and $C_2$ are a part of a $(\tau, r)$ partition, and in the second inequality we used the assumption that $\tau\geq 10$. This is a contradiction.
	\end{proof}


	We next study different syntactic partitions of a circuit $C$.
	
	\begin{claim}[Lower rank implies finer partition]
		\label{cla:low-rank-finer-part-syn}
		Let $\tau \ge 10$.
		Let $C$ be a minimal multilinear $\Sps{k}$ circuit. Let $(C_1, \ldots, C_s)$ be a $(\tau, r_C)$-syntactic partition of the multiplication gates in $C$. Let $(D_1, \ldots, D_{s'})$ be another $(\tau, r_D)$ partition of the gates of $C$.
		
		Assume $r_C \ge r_D$. Then, for every $i \in [s]$ there is a subset $S_i \subseteq [s']$ such that
		\[
		C_i = \sum_{j \in S_i} D_j
		\]
		and the subsets $S_1, \ldots, S_s$ form a partition of $[s']$.
	\end{claim}
	
	\begin{proof}
		The claim would follow if we prove that if $D_i$ contains a multiplication gate from $C_j$ then it contains no gate from $C_{j'}$ for $j\neq j'$. This will show that $D_i$ is ``contained'' in $C_j$. Indeed, if this was the case then \autoref{lem:large-rank-syn} would imply that
		\[
		r_C\leq \tau r_C/10 < \ranksyn(M_j+M_{j'})\leq \ranksyn(D_i)\leq r_D
		\]
		in contradiction.
	\end{proof}

	\begin{corollary}\label{cor:max-clus-min-rank-syn}
		Let $\tau \ge 10$. Let $C$ be a minimal multilinear $\Sps{k}$ circuit. Let $(C_1, \ldots, C_s)$ be a $(\tau, r_C)$-syntactic partition of the multiplication gates in $C$,
		that has the largest number of clusters among all $\tau$-syntactic partition of $C$. Then, for every other $\tau$ partition of $C$, $(D_1, \ldots, D_{s'})$, we have that $r_C\leq r_{D}$.
	\end{corollary}
	
	\begin{proof}
		If it was the other case then \autoref{cla:low-rank-finer-part-syn} would give that there must be more clusters in $\{D_i\}$. 
	\end{proof}

	\begin{corollary}[Uniqueness of syntactic partitions with the same number of clusters]\label{cor:unique-syntactic}
		Let $\tau \ge 10$. Let $C$ be a minimal multilinear $\Sps{k}$ circuit.
		Let $(C_1, \ldots, C_s)$ and  $(D_1, \ldots, D_s)$ be $(\tau, r_C)$ and $(\tau, r_D)$-syntactic partitions of the multiplication gates in $C$, respectively.
		Then, there is a permutation $\pi$ on $[s]$ such that for every $i\in[s]$, $C_i=D_{\pi(i)}$.
	\end{corollary}
	
	\begin{proof}
		Without loss of generality suppose that $r_C \ge r_D$. By \autoref{cla:low-rank-finer-part-syn}, $(D_1, \ldots, D_s)$ is a refinement of $(C_1, \ldots, C_s)$. However, they have the same number of clusters, so they must be the same partition.
	\end{proof}

	\subsubsection{Algorithms for Computing Partitions}

	An algorithm for computing $(\tau,r)$-syntactic partitions was provided by Karnin and Shpilka \cite{KarninShpilka09}.
	
	\sloppy
	\begin{lemma}[Syntactic Clustering Algorithm; See Algorithm 1 and Lemma 5.1 of \cite{KarninShpilka09}]
		\label{lem:syntactic-clustering}
		Let  $n,k, r_{init}, \tau \in \N$.
		There exists an algorithm that given $\tau$ and an $n$-variate multilinear $\Sps{k}$ circuit $C$ as input, outputs $r \in \N$ such that 
		\[R_M(2k) \leq r \leq k^{(k-2) \cdot \lceil \log_k(\tau)\rceil} \cdot R_M(2k)\leq  (k\tau)^{k-2} \cdot R_M(2k)
		\]
		and a $(\tau, r)$-syntactic partition of $[k]$, in time $O(\log(\tau) \cdot n^3k^4).$ Further, with an additional running time of $2^{O(k^2)} \cdot \poly(n)$, we can guarantee that this syntactic partition has the lowest value of $r$ among all $\tau$ syntactic partitions of $C$.
	\end{lemma}
	We remark that the ``further'' part isn't explicitly stated in \cite{KarninShpilka09}. However, it is easy to modify their algorithm in order to guarantee this property. For example, after running their algorithm one can run a brute force search over all partitions of $[k]$ and search for a $\tau$-partition with a lower value of $r$. In the applications of \autoref{lem:syntactic-clustering}, the additional running time incurred by this step is either irrelevant or anyway subsumed by larger factors of $k$ originating from other elements in the proof.

	\subsection{Existence of a Unique Syntactic Partition}
	
	In this section we prove that for every multilinear polynomial $f\in\Sps{k}$ there is a parameter $\tau$, which is bounded by some function of $k$, such that any two $\tau$ partitions of any two $\Sps{k}$ circuits computing $f$ define, up to a permutation, the same clusters.
	
	We start with the following claim that shows the existence of a $(\tau_1,r)$ partition with the special property that its rank is bounded as function of $\tau_0$ that is much smaller than $\tau_1$.

	\begin{claim}
		\label{cl:partition-two-taus-syn}
		For every function $\varphi: \mathbb{N} \to \mathbb{N}$, parameter $\tau_\text{min}$ and  multilinear $f \in \Sps{k}$ there is $\tau_\text{min}\leq \tau_0 \le {\tau_\text{min}}^{ \varphi(k)^k}$ and a $({\tau_1}, r)$-syntactic partition of $f$ with:
		\begin{itemize}
			\item ${\tau_1} = \tau_0^{\varphi(k)}$.
			\item $r \le  R_M(2k) \cdot  (k\tau_0)^{k-2}$.
		\end{itemize}
	\end{claim}
	
	\begin{proof}
		Let $C$ be any $\Sps{k}$ circuit computing $f$.
		Using the algorithm promised in \autoref{lem:syntactic-clustering}, find a $\kappa_0=\tau_\text{min}$-syntactic partition of $C$.
		\autoref{lem:syntactic-clustering} guarantees that  this $(\kappa_0, r_0)$ partition satisfies $r_0 \le R_M(2k) \cdot  (k\kappa_0)^{k-2}$. If this partition is also a $(\kappa_1, r_0)$ partition for $\kappa_1 = \kappa_0^{\varphi(k)}$, then we are done by setting $\tau_0=\kappa_0$ and $\tau_1=\kappa_1$.
		
		Otherwise, we continue in the same manner with $\kappa_1$ instead of $\kappa_0$ (i.e., we consider a syntactic partition with parameter $\kappa_1$) etc. 
		
		We claim that this process terminates after at most $k$ iterations and finds the desired partition. It suffices to show that at every step the number of  clusters in the partition decreases.
		At the $i$-th iteration of this algorithm we have a $(\kappa_i, r_i)$ partition and similarly at the $(i+1)$-th step, a $(\kappa_{i+1}, r_{i+1})$ partition.
		Note that $r_{i+1} \ge r_i$, as otherwise, the $(i+1)$-th partition would also be a $\kappa_i$ partition with a lower rank than the $i$-th partition, and would have been found by the algorithm in \autoref{lem:syntactic-clustering}.
		
		Both the $i$-th and the $(i+1)$-th partitions are $\tau_\text{min}$ partitions. By \autoref{cla:low-rank-finer-part-syn}, the $i$-th partition is a refinement of the $(i+1)$-th partition, and in particular, since they are not the same partition (as otherwise the algorithm terminates), the number of clusters decreases.
	\end{proof}
	
	The next claim proves that for every multilinear polynomial $f \in \Sps{k}$ there is $\tau =O(k^{k+2})^{k^{2k+1}}$ such that all $\tau$ partitions of $f$ are equivalent. This result fixes the aforementioned error in \cite{KarninShpilka09} and makes the argument in \cite{BSV21} work.
	
	\begin{theorem}\label{thm:unique-syn-part}
		For every multilinear polynomial $f \in \Sps{k}$ there is $\tau =O(k^{k+2})^{k^{2k+1}}$ such that the following holds: Let $C,D$ be any two $\Sps{k}$ circuits computing $f$. Let $C=\sum_{i=1}^{s}C_i$ and $D=\sum_{i=1}^{s'}D_i$ be the $\tau$-partitions of $C$ and $D$, respectively, that \autoref{lem:syntactic-clustering} guarantees. Then $s=s'$ and there is a permutation $\pi:[s]\to[s]$ such that $[C_i]=[D_{\pi(i)}]$. Furthermore, for every $i$,  $\ranksyn(C_i)/k -2 R_M(2k)\le \ranksyn(D_{\pi(i)})\le  k \cdot \ranksyn(C_i)+2k R_M(2k)$. 
	\end{theorem}
	
	\begin{proof}
		For every $\tau$ we denote by
		\[
		r(\tau):= R_M(2k) \cdot  (k \tau)^{k-2}
		\] 
		the rank bound given in \autoref{lem:syntactic-clustering} for the rank of clusters in a $\tau$ partition. 
		
		Let $C$ be any $\Sps{k}$ circuit computing $f$.	
		Apply \autoref{cl:partition-two-taus-syn} on $C$ with $\varphi(k)=k^2$ and $\tau_\text{min}= 10 R_M(2k)k^{k}$ and  let $\tau_0,\tau_1$ be the parameters of the claimed partition. Denote this partition as $C=\sum_{i=1}^{s}C_i$.  Set $\tau=\tau_0^k$. 
		Calculating we see that 
		\[
		\tau=\tau_0^k \leq \left(\tau_\text{min}^{\varphi(k)^k}\right)^k=\tau_\text{min}^{k^{2k+1}}=( 10 R_M(2k)k^{k-2})^{k^{2k+1}}= O(k^{k+2})^{k^{2k+1}}.
		\]
		
		Let $D$ be any $\Sps{k}$ computing $f$ such that $D=\sum_{i=1}^{s'}D_i$ is a $(\tau,r_D)$ partition of $D$, as promised by \autoref{lem:syntactic-clustering}. 
		
		Consider the circuit $C-D$ and a minimal subcircuit in it $E=\sum M_i -\sum T_j$ where the $M_i$'s are multiplication gates in $C$ and the $T_j$'s are from $D$. We claim that there cannot be $T_1$ and $T_2$ from different $D_i$'s. Assume for a contradiction  that $T_1\in D_1$ and that $T_2\in D_2$.  In this case \autoref{lem:large-rank-syn} gives
		\[
		{\tau_\text{min}}^k/10\leq \tau_0^k/10=\tau/10\leq  \tau r_D/10 < \ranksyn(T_1+T_2)\leq \ranksyn(E) \leq R_M(2k)
		\]
		in contradiction to the choice of $\tau_\text{min}$. A similar argument would show that all the $M_i$'s belong to the same cluster. Thus, every minimal circuit $E$ contains gates from a single $C_i$ and a single $D_j$. 
		We partition each cluster $C_i$ and each $D_i$ according to the subsets of multiplication gates appearing in each minimal circuit. That is,
		we write $C_i=\sum_j C_{i,j}$ and $D_i=\sum_j D_{i,j}$ so that each such minimal circuit $E$ is of the form $C_{i,j}-D_{i',j'}$.
		
		As each such minimal $E$ contain multiplication gates from a single $D_i$, we can represent every $D_i$ as $D_i=\sum_j\sum_p C_{j,p}$. We wish to show that all $C_{j,p}$ come from a single $C_j$. Indeed, assume towards a contradiction that $C_{1,1}\in C_1$ and $C_{2,1}\in C_2$ (this is without loss of generality) and that $E_1=D_{i,1}-C_{1,1}$ and $E_2=D_{i,2}-C_{2,1}$ are minimal circuits computing the zero polynomial. As $\ranksyn(D_{i,1}+D_{i,2})\leq\ranksyn(D_i)\leq r_D$ and $\ranksyn(E_1),\ranksyn(E_2)\leq R_M(2k)$, we get from \autoref{lem:small-syn-rank-t} (for $t=2$) and  \autoref{thm:rank bound ml} that
		\begin{align}\label{eq:oneC}
			\tau_1/10\leq  \tau_1 r_C/10 &< \ranksyn(C_{1,1}+C_{2,1})\leq 2\ranksyn(D_i) + 4R_M(2k) \nonumber \\ &\leq  2 R_M(2k) \cdot  (k\tau)^{k-2}+ 4R_M(2k)
		\end{align}
		(where the lower bound follows by \autoref{lem:large-rank-syn}). This again contradicts the choice of $\tau_\text{min}$ and the fact that $\tau_1=\tau^k$. 
		
		This implies that for every $i \in [s]$, we can sum over the different minimal circuits in which  the different ``parts'' of $C_i$ appear, and get  a subset $S_i \subseteq [s']$ such that
		$C_i = \sum_{j \in S_i} D_j$. This implies that $s' \ge s$, and observe that if $s=s'$ then this guarantees the existence of the claimed permutation $\pi$. So assume that $s'>s$ and that $|S_i|>1$. Consider the $\Sps{2k}$ circuit $C_i-\sum_{j \in S_i} D_j$. Using the same notation as before we write $\sum_j C_{i,j}=C_i=\sum_{j\in S_i}\sum_p D_{j,p}$, such that there is a one to one and onto map, that maps, for every $j$, a pair $j',p'$ such that $E_j=C_{i,j}-D_{j',p'}$ is a minimal circuit computing the zero polynomial. Applying \autoref{lem:small-syn-rank-t} again we get that  
		\begin{equation}\label{eq:rcrd}
			\ranksyn(\sum_{j\in S_i}\sum_p D_{j,p})=\ranksyn(\sum_j C_{i,j})\leq tr_C+2t R_M(2k) \leq kr_C+2k R_M(2k).
		\end{equation}
		On the other hand, our assumption that $|S_i|>1$ implies that there are $j\neq j'\in S_i$. 
		We now get that
		\begin{equation}\label{eq:rctau'}
			\tau/10 \leq \tau r_D/10< \ranksyn(D_{j,1}+D_{j',1})\leq 	\ranksyn(\sum_{j\in S_i}\sum_p D_{j,p}) .
		\end{equation}
		Observe that 
		\begin{align*}
			kr_C+2k R_M(2k)& \leq k r(\tau_0) +2k R_M(2k) = k  R_M(2k) \cdot  (k \tau_0)^{k-2} +2k R_M(2k) \nonumber\\
			&\leq R_M(2k) k^{k-1} \tau_0^{k-2} +2k R_M(2k) < 3R_M(2k)k^k \tau_0^{k-2} \nonumber\\
			& < \tau_0^k/10 =\tau/10, 
		\end{align*}
		where the last inequality follows from definitions of $\tau_\text{min}$ and the fact that $\tau_\text{min}\leq \tau_0$.
		Combining the last calculation with \eqref{eq:rcrd} and \eqref{eq:rctau'} we get a contradiction. Hence $s=s'$ and there is a matching $\pi$ such that $C_i=D_{\pi(i)}$.

		The claim regarding the relation between $\ranksyn(C_i)$ and $\ranksyn(D_{\pi(i)})$ follows from the same argument as the one leading to \eqref{eq:rcrd}.
	\end{proof}
	

	\section{Semantic Rank of Depth-$3$ Circuits}
	
	In the following section, we define the semantic rank of polynomials computed by $\Sps{k}$ circuits. Note that while the \emph{syntactic} rank is inherently tied to a \emph{circuit} $C$ computing the polynomial, the \emph{semantic} rank is independent of the representation or computation of the polynomial.

	We say that a polynomial $g \in \F[x_1, \ldots, x_n]$ depends on $r$ linear functions if there exist $r$ linear functions $\ell_1, \ldots, \ell_r$ and a polynomial $h \in \F[y_1, \ldots, y_r]$ such that $g(\vx) = h(\ell_1(\vx), \ldots, \ell_r(\vx))$.
	
	\begin{definition}[Semantic Rank]
		\label{def:sem-rank}
		Let $f \in \F[x_1,\ldots,x_n]$ be a polynomial.   Define $\Lin(f)$ to be the product of the
		linear factors of $f$. Let $r \in \N$ be the minimal integer such that $f /Lin(f)$ is a polynomial of exactly $r$ linear functions. We define $\ranksem(f) = r$.
	\end{definition}
	
	Recall that the number of linear functions that a polynomial depends on equals the rank of its \emph{partial derivative matrix}.
	
	\begin{definition}
		\label{def:partial-derivative-matrix}
		Let $f \in \F[x_1,\ldots,x_n]$ be a polynomial.  Define $M_f$ to be a matrix whose $i$-th row contains the coefficients of $\partial f / \partial x_i$.
	\end{definition}
	
	Note that if $f$ depends on exactly $r$ linear functions and  $\Char(\F)=0$ or $\Char(\F)> \deg(f)$, then $\rank(M_f) = r$.
	
	\begin{remark}
		\label{rem:semantic-rank-0}
		Note that under the definition above, it may be the case that $f$ is non-zero and yet $\ranksem(f)$ equals $0$. This happens when $f$ is a product of linear functions.
		In what follows we will often implicitly assume that $\ranksem(f) \ge 1$. This doesn't affect our results but somewhat simplifies the presentation. One may also arbitrarily define the semantic rank of $f$ to be 1 when $f$ is a non-zero product of linear functions.
	\end{remark}

	\subsection{Semantic vs.\ Syntactic Rank}
	
	We now prove several claims that relate the syntactic and semantic notions of rank for polynomials computed by multilinear $\Sps{k}$ circuits. We start by observing that the semantic rank is at most the syntactic rank.

	\begin{observation}
		\label{obs:semantic-vs-syntactic}
		Suppose $f$ is a polynomial in multilinear $\Sps{k}$.
		Then, every multilinear $\Sps{k}$ circuit $C$  computing $f$ satisfies $\ranksyn(C) \ge \ranksem(f)$.
	\end{observation}
	
	We will also need the following definition from \cite{KarninShpilka09}.
	
	\begin{definition}[$\SumProdSum(k,d,\rho)$ circuits, \cite{KarninShpilka09}]\label{def:k-d-rho}
		A generalized depth-3 circuit with parameters $(k,d,\rho)$ is a depth-$3$ circuit with top fan-in $k$ such that every multiplication gate can also multiply a polynomial that depends on at most $\rho$ linear functions. That is, a polynomial $f(\vx)$ that is computed by a $\SumProdSum(k,d,\rho)$ circuit has the form
		\[
		f(\vx) = \sum_{i=1}^k \left( \prod_{j=1}^{d_i} \ell_{i,j} (\vx) \right) \cdot h_i (\ell'_{i,1} (\vx), \ldots, \ell'_{i,\rho_i}(\vx))
		\]
		where $d_i \le d$ and $\rho_i \le \rho$ for all $i \in [k]$.
	\end{definition}
	
	The measure $\ranksyn$ is extended to $\SumProdSum(k,d,\rho)$ circuits in Definition 2.1 of \cite{KarninShpilka08} (where it is simply called ``rank'').
	
	The following observation follows immediately from the definitions.
	\begin{observation}
		\label{obs:1-d-r}
		Suppose $f$ is a polynomial computed by a multilinear $\Sps{k}$ circuit $C$. Let $r = \ranksem(f)$.
		Then, there is a multilinear $\SumProdSum(1,d,r)$ circuit $C_f$ computing $f$ such that $\ranksyn(C_f) = \ranksem(f)$.
	\end{observation}

	An analog of \autoref{obs:semantic-vs-syntactic} also holds for $\SumProdSum(k,d,\rho)$ circuits.
	
	\begin{lemma}[Lemma 2.20 in \cite{KarninShpilka09}]
		\label{lem:syntatic-vs-semantic-k-d-rho}
		Suppose $f$ is a polynomial computed by a multilinear $\SumProdSum(k,d,\rho)$ circuit $C$.
		Then, every multilinear  $\SumProdSum(k,d,\rho)$ circuit $C'$  computing $f$ satisfies $\ranksyn(C') \ge \ranksem(f)$.
	\end{lemma}

	We now want to upper bound the syntactic rank as a function of the semantic rank (naturally, this only makes sense for \emph{minimal} circuits, as other circuits can have artificially large syntactic rank). Our argument is essentially identical to Lemma 2.20 of \cite{KarninShpilka09}. However since our notation is different and since we observe that we can slightly improve the bound in the multilinear case (the bound in \cite{KarninShpilka09} depends on the degree of the polynomial being computed), we repeat their short argument.
	
	We start with the rank bound for $\SumProdSum(k,d,\rho)$ circuit proved in \cite{KarninShpilka08}.
	
	\begin{lemma}[Lemma 4.2 of \cite{KarninShpilka08}]
		\label{lem:rank-bound-k-d-pho}
		Let $C$ be a simple and minimal $\SumProdSum(k,d,\rho)$ circuit computing the zero polynomial. Suppose
		\[
		C = \sum_{i=1}^k \left( \prod_{j=1}^{d_i} \ell_{i,j} \right) \cdot h_i (\tilde{\ell}_{i_1}, \ldots, \tilde{\ell}_{i,\rho_i})
		\]
		and let $\tilde{R} = \sum_{i=1}^k \rho_i$.
		Then $\ranksyn(C) \le R(k,d) + \tilde{R}$.
	\end{lemma}
	
	Here $R(k,d) = 4k^2 \log (2d)$ is the rank bound for (not necessarily multilinear) $\Sps{k}$ circuits (recall \autoref{thm:rank bound ml}). Note that trivially $\tilde{R} \le k \rho$, but in \autoref{lem:small-semantic-implies-small-syntactic-rank} we shall use the stricter upper bound stated in the lemma.

	The proof of \autoref{lem:rank-bound-k-d-pho} in \cite{KarninShpilka08} is also not very complicated. Given a $\SumProdSum(k,d,\rho)$  circuit $C$ as in the statement of the lemma, one fixes randomly the linear functions $\tilde{\ell}_{i_1}, \ldots, \tilde{\ell}_{i,\rho_i}$, for $i \in [k]$, to obtain a simple and minimal $\Sps{k}$ circuit of degree at most $d$, and applies the rank bound for such circuits. Note that fixing those linear functions might make the circuit non-multilinear even if the original circuit was multilinear, which means we have to use the rank bound $R(k,d)$ for non-multilinear $\Sps{k}$ circuits. This incurs a dependence on the degree $d$. 	However, it is also convenient to have a form of \autoref{lem:rank-bound-k-d-pho} with no dependence on $d$. This is possible since $C$ is multilinear. A similar observation was made by Dvir and Shpilka \cite{DS07} for $\Sps{k}$ circuits. Since $C$ is multilinear,  all linear functions appearing in each multiplication gate are variable disjoint, and hence linearly independent, which implies that $\ranksyn(C) \ge d$. Together with the upper bound in \autoref{lem:rank-bound-k-d-pho}, this implies the following corollary.
	
	\begin{corollary}[Rank bound for multilinear $\SumProdSum(k,d,\rho)$ circuits with no dependence on $d$]
		\label{cor:rank-bound-multilinear-k-d-pho}
		Let $C$ be a simple and minimal $\SumProdSum(k,d,\rho)$ circuit computing the zero polynomial. Then,
		\[
		\ranksyn(C) \le 40 \cdot (k^2 \log k + k^2 \rho).
		\]
	\end{corollary}
	
	\begin{proof}
		Since $C$ is multilinear and by \autoref{lem:rank-bound-k-d-pho}, we have that
		\[
		d \le \ranksyn(C) \le R(k,d) + k \rho \le 4k^2 \log(2d)  + k \rho.
		\]
		We claim that inequality holds only if $d \le 40 \cdot (k^2 \log k + k^2 \rho)$. Indeed, if $d \le 40 k^2\rho$ then the bound clearly holds. Otherwise we get
		\begin{align*}
			d &\leq 4 k^2 \log (2d) +k \rho <  4 k^2 \log(2d) + d/40k \\
			d&\leq \left((160k^3)/(40k-1)\right)\log (2d)
		\end{align*}
		and a simple calculation shows that this does not hold if 
		$d > 40 \cdot k^2 \log k$. This implies the claimed upper bound on $\ranksyn(C)$.
	\end{proof}

	The following lemma uses the notation of \autoref{thm:rank bound ml}.
	
	\begin{lemma}[Small semantic-rank implies small syntactic-rank, similar to Lemma 2.20 in \cite{KarninShpilka09}]
		\label{lem:small-semantic-implies-small-syntactic-rank}
		Let $C$ be a minimal multilinear $\Sps{k}$ circuit computing a polynomial $f$. Suppose that $\ranksem(f) \le r$. Then $\ranksyn(C) \le r  + R(k+1,\ranksyn(C))$. In particular, $\ranksyn(C)  \le 2^7 r k^2 \log k$. 
	\end{lemma}
	
	\begin{proof}
		Denote $C = \sum_{i=1}^k M_i$ where each $M_i$ is a multiplication gate, and denote $r_C = \ranksyn(C)$.
		
		Let $C_f$ be a $\SumProdSum(1,d,r)$ computing $f$, so that $C_f = (\prod_i \ell_i) \cdot h(\tilde{\ell}_1, \ldots, \tilde{\ell}_r)$ where $h$ has no linear factors.
		
		Consider the circuit $C - C_f$, which computes the zero polynomial. We factor out the $\gcd$ of this circuit, which is the common linear factor of $M_1, \ldots, M_k$ and $ (\prod_i \ell_i) $. Note, however, that if a linear function divides all of $M_1, \ldots, M_k$, then it divides $f$ (since $C$ computes $f$), and thus it is one of the $\ell_i$'s. Hence, the $\gcd$ of the circuit $C-C_f$ equals $\gcd(M_1, \ldots, M_k)$.
		
		Consequently, we can write $C-C_f$ as
		\[
		C-C_f = \gcd(M_1, \ldots, M_k) \cdot \left( \sum_{i=1}^k \tilde{M}_i - \left( \prod_{i \in A} \ell_i \right) \cdot h(\tilde{\ell}_1, \ldots, \tilde{\ell}_r) \right),
		\]
		where $A$ is some subset of the (indices of the) original linear functions appearing in $C_f$.
		Let $D=\left( \sum_{i=1}^k \tilde{M}_i - \left( \prod_{i \in A} \ell_i \right) \cdot h(\tilde{\ell}_1, \ldots, \tilde{\ell}_r) \right)$. Since $C-C_f$ computes the zero polynomial, $D$ is also a circuit computing the zero polynomial
		
		By construction, $D$ is simple, and by minimality of $C$ it is also minimal: indeed, no subcircuit of $C$ (nor $C$ itself) computes the zero polynomial, and no subset $S \subsetneq [k]$ of the $\tilde{M_i}$'s can equal $\left( \prod_{i \in A} \ell_i \right) \cdot h$, as this would imply that $[\sum_{i \in S} M_i]=f$, which would again contradict the minimality of $C$.
		
		Finally, note that the degree of $D$ is at most $r_C$. Indeed, by multilinearity, each $\tilde{M_i}$ is a product of variable-disjoint (and hence linearly independent) linear functions, which means that their number can be at most $r_C$. As $D$ computes the zero polynomial, the degree of $ \left( \prod_{i \in A} \ell_i \right) \cdot h(\tilde{\ell}_1, \ldots, \tilde{\ell}_r)$ is also at most $r_C$.
		
		Thus, we have concluded that $D$ is a simple and minimal $\SumProdSum(k+1, r_C, r)$ circuit computing the zero polynomial. By \autoref{lem:rank-bound-k-d-pho},
		\[
		\ranksyn(D) \le R(k+1, r_C) + r.
		\]
		Finally note that $r_C$, which is the dimension of the span of the linear functions in the $\tilde{M}_i$'s, is at most $\ranksyn(D)$. Thus, we get the inequality
		\[
		r_C\le\ranksyn(D) \le R(k+1, r_C) + r \le  4 (k+1)^2 \log (2r_C)+r.
		\]
		For the ``in particular'' part, as before, we consider two cases. If $r_C  \le ((k+1)^2+1)r$,  then the bound clearly holds. Otherwise,
		$r_C \le 4 ((k+1)^2+1) \log (2r_C)$ and a simple calculation shows that this implies  $r_C \le 2^7 k^2 \log k$. The statement follows from a combination of the two cases.
	\end{proof}

	\subsection{Semantic Partitions of $\Sps{k}$ Circuits}\label{sec:sem-part}
	
	We next define semantic partitions of $\Sps{k}$ circuits, that correspond to semantic rank in the same manner that syntactic partitions correspond to syntactic rank (recall \autoref{def:tau-r-syntactic}).

	\begin{definition}[Semantic Partition]
		\label{def:tau-r-semantic}
		Let $f$ be a multilinear polynomial.
		We say that $(f_1, \ldots, f_s)$ is a $(\tau,r)$ \emph{semantic} partition of $f$ if $f = \sum_{i=1}^s f_i$, and
		\begin{itemize}
			\item For every $i \in [s]$, $\ranksem(f_i) \le r$.
			\item For every $i \neq j \in [s]$, $\ranksem(f_i + f_j) \ge \tau r$.
		\end{itemize}
		We further say that the partition is \emph{realizable} if there exists a $\Sps{k}$ circuit $C$ computing $f$ and a partition of its multiplication gates $(C_1, \ldots, C_s)$ such that $[C_i] = f_i$. From now on, we only consider realizable partitions.
	\end{definition}
	
	We also often use the term ``$\tau$-partition'' (either syntactic or semantic) where  it is implied that the partition is a $(\tau,r)$-partition for some value of $r$.
	
	\begin{corollary}\label{cor:sem-is-syn}
		Let $C$ be a minimal multilinear $\Sps{k}$ circuit.
		Every $(\tau,r)$-semantic partition of $[C]$ is also a $(\tau',r')$-syntactic partition of $C$ with $r'=2^7 k^2 \log k \cdot r $ and $\tau'=\tau/(2^7 k^2 \log k)$.
	\end{corollary}
	\begin{proof}
		Let $[C]=\sum_i [C_i]$ be the assumed semantic partition (recall that we always assume that the partition is realizable). From \autoref{lem:small-semantic-implies-small-syntactic-rank} we get $\ranksyn(C_i)\leq 2^7 \ranksem([C_i]) k^2 \log k \leq  2^7 r k^2 \log k=r'$.  Furthermore, \autoref{obs:semantic-vs-syntactic} implies that $\ranksyn(C'_i + C'_j)\geq \ranksem([C'_i] + [C'_j])\geq  \tau r \ge \tau'\cdot r'$, as claimed. 
	\end{proof}
	
	\begin{corollary}\label{cor:syn-is-sem}
		Let $C$ be a minimal multilinear $\Sps{k}$ circuit.
		Let $(C_1, \ldots, C_s)$ be a $(\tau,r)$-syntactic partition of $C$. Then $([C_1], \ldots, [C_s])$ is a $(\tau',r)$-semantic partition of $[C]$ with $\tau'=\tau/ (2^7 k^2 \log k)$.
	\end{corollary}
	\begin{proof}
		By \autoref{obs:semantic-vs-syntactic}, for every $i \in [s]$ we have that $\ranksem([C_i]) \le \ranksyn(C_i) \le r$.
		Furthermore, \autoref{lem:small-semantic-implies-small-syntactic-rank} implies that for every $i \neq j$,
		\[
		\ranksem([C_i] + [C_j]) \ge \ranksyn(C_i + C_j) / (2^7 k^2 \log k) \ge \tau r / (2^7 k^2 \log k) = \tau' r.  \qedhere
		\]
	\end{proof}

	\subsubsection{Uniqueness Properties of Semantic Partitions}
	\label{sec:unique-semantic-partition}
	
	We next state and prove an analogous claim to \autoref{cla:low-rank-finer-part-syn}.
	
	\begin{claim}[Lower rank implies finer partition]
		\label{cla:low-rank-finer-part}
		Let $\tau >2^{10} (k^2 \log k)$.
		Let $C, D$ be two minimal multilinear $\Sps{k}$ circuits computing the same polynomial $f$. Let $C=\sum_{i=1}^{s}C_i$ be a partition of the multiplication gates in $C$ and similarly $D=\sum_{i=1}^{s'}D_i$ a partition of the gates in $D$.
		Let $f_i = [C_i]$ and suppose that $(f_1, \ldots, f_s)$ is a $(\tau,r_1)$-semantic partition of $f$. Similarly, let $g_i = [D_i]$ and suppose that $(g_1, \ldots, g_{s'})$ is a $(\tau, r_2)$-semantic partition of $f$. Assume $r_1 \ge r_2$. 
		
		Then, for every $i \in [s]$ there is a subset $S_i \subseteq [s']$ such that
		\[
		f_i = \sum_{j \in S_i} g_j
		\]
		and the subsets $S_1, \ldots, S_s$ form a partition of $[s']$.
	\end{claim}
	
	\begin{proof}
		The proof is similar to the proof of  \autoref{cla:low-rank-finer-part-syn}.	
		Consider the following $\SumProdSum(2k,d,r_1)$ circuit $E$ computing the zero polynomial:
		\[
		\sum_{i=1}^s C_{f_i} - \sum_{i=1}^{s'} C_{g_i}
		\]
		(where $C_{f_i}$ is a $\SumProdSum(1,\deg(f_i),r_1)$ gate computing $f_i$, and similarly $C_{g_i}$ computes $g_i$).
		Consider a minimal subcircuit $E'$ of $E$,
		\[
		\sum_{i \in \cI} C_{f_i} - \sum_{j \in \cJ} C_{g_j}
		\]
		(note that $\cI, \cJ \neq \emptyset$ since both $C,D$ are minimal circuits).
		
		Assume towards contradiction that there exist $i,i' \in \cI$ with $i \neq i'$. Then
		\[
		\ranksem(f_i + f_{i'}) \le \ranksyn(C_{i} + C_{{i'}}) \le \ranksyn(E') \le 40 ((2k)^2 \log (2k) + (2k)^2 r_1).
		\]
		The first inequality follows from \autoref{lem:syntatic-vs-semantic-k-d-rho}. The second inequality is immediate from the definition of syntactic rank, and the last inequality is the rank bound of \autoref{cor:rank-bound-multilinear-k-d-pho}.  
		
		From the assumption, $\ranksem(f_i + f_{i'}) \ge \tau r_1$, which contradicts the choice of $\tau$.
	\end{proof}

	\begin{corollary}\label{cor:same-num-semantic}
		Let $C$ and $D$ be as in \autoref{cla:low-rank-finer-part}. If $s=s'$ then there is a permutation $\pi$ of $[s]$ such that $f_i=g_{\pi(i)}$.
	\end{corollary}
	
	\begin{proof}
		Indeed, in the notation of \autoref{cla:low-rank-finer-part}, we will get that the sets $S_i$ must be of size $1$, as otherwise we would have $s'>s$, in contradiction to the assumption. As $|S_i|=1$ it follows that $f_i=g_j$ where $S_i=\{j\}$.
	\end{proof}

	The next lemma considers two different semantic partitions. In the case of syntactic partitions, \autoref{cor:unique-syntactic} argues about partitions of the same circuit, in contrast, in the case of semantic partitions, we can compare partitions of different circuits computing $f$. This is one of the advantages of semantic rank over syntactic rank.
	
	\begin{lemma}[Relation between different partitions]\label{lem:lower-rank-finer-partition}
		Let $\tau > R_M(2k) + 2^{10} k^2 \log k$.
		Let $C, D$ be two minimal multilinear $\Sps{k}$ circuits computing the same polynomial $f$. Let $C=\sum_{i=1}^{s}C_i$ be a partition of the multiplication gates in $C$ and similarly $D=\sum_{i=1}^{s'}D_i$ a partition of the gates in $D$.
		
		Let $f_i = [C_i]$ and suppose that $(f_1, \ldots, f_s)$ is a $(\tau,r_1)$-semantic partition of $f$. Similarly, let $g_i = [D_i]$ and suppose that $(g_1, \ldots, g_{s'})$ is a $(\tau, r_2)$-semantic partition of $f$. Assume $r_1 \ge r_2$.
		
		Then, there's a refinement of $(C_1, \ldots, C_s)$, denoted $(C'_1, \ldots, C'_{s'})$, such that up to reordering, $[C'_i] = g_i$. 
	\end{lemma}
	
	\begin{proof}
		By \autoref{cla:low-rank-finer-part}, for every $i \in [s]$ there's a subset $S_i \subseteq [s']$ such that
		\begin{equation}
			\label{eq:fi-sum-of-g_js}
			f_i = \sum_{j \in S_i} g_j
		\end{equation}
		and the subsets $S_1, \ldots, S_s$ form a partition of $[s']$.
		Equation \eqref{eq:fi-sum-of-g_js} implies $[C_i] = [\sum_{j \in S_i} g_j]$. Thus,
		\[
		C_i - \sum_{j \in S_i} D_j
		\]
		is a $\Sps{2k}$ circuit $E$ computing the zero polynomial.
		
		Let $E'$ be a minimal subcircuit of $E$. Write
		\[
		E' = \sum_t M_t - \sum_{j} \sum_p T_{j,p}
		\]
		where $M_t \in C_i$ and $T_{j,p} \in D_j$. In particular,
		\[
		\ranksyn(\sum_j \sum_p T_{j,p}) \le \ranksyn(E') \le R_M(2k).
		\]
		Suppose $j_1\neq j_2 \in [s']$ appear in the sum $\sum_j \sum_p T_{j,p}$. Let $T_{j_1,p_1} \in D_{j_1}$ and $T_{j_2,p_2} \in D_{j_2}$. As
		\[
		\dist(T_{j_1,p_1}, T_{j_2,p_2}) \le R_M(2k),
		\]
		we get from Lemma 2.17 in \cite{KarninShpilka09} (the triangle inequality for syntactic rank) and \autoref{lem:small-semantic-implies-small-syntactic-rank}  that 
		\begin{align*}
			\dist(D_{j_1}, D_{j_2}) &\le \dist(D_{j_1}, T_{j_1,p_1}) +\dist(T_{j_1,p_1}, D_{j_2})\\ &\leq \dist(D_{j_1}, T_{j_1,p_1}) +\dist(T_{j_1,p_1}, T_{j_2,p_2}) + \dist(T_{j_2,p_2}, D_{j_2})\\ &\le \ranksyn(D_{j_1}) +R_M(2k) + \ranksyn(D_{j_2})\\ &\le 2^8 k^2 \log k \cdot r_2 + R_M(2k)\;.
		\end{align*}
		Hence, using \autoref{obs:semantic-vs-syntactic},
		\[
		\tau r_2 \le \ranksem(g_{j_1} + g_{j_2}) \le \dist(D_{j_1}, D_{j_2})\le 2^8 k^2 \log k \cdot r_2 + R_M(2k)
		\]
		 in contradiction. Thus, the sum $\sum_j \sum_p T_{j,p}$ contains multiplication gates from a single $D_j$. This implies the claim.
	\end{proof}

	\begin{corollary}\label{cor:max-clus-min-rank}
		Let  $\tau > R_M(2k) + 2^{10} k^2 \log k$.
		If $f=\sum f_i$ is a  $(\tau,r_1)$-semantic partition that has the largest number of clusters among all $\tau$-semantic partitions of $f$, then for every other $(\tau,r_2)$-semantic partition $f=\sum g_i$ we have that $r_1\leq r_{2}$.
	\end{corollary}
	
	\begin{proof}
		If it was the other case then \autoref{lem:lower-rank-finer-partition} would give that there must be more clusters in $\{g_i\}$. 
	\end{proof}

	\begin{corollary}[Uniqueness of maximal partition regardless of representation]\label{cor:unique-semantic}
		Let $\tau > R_M(2k) + 2^{10} k^2 \log k$.
		Let $C, D$ be two minimal multilinear $\Sps{k}$ circuits computing the same polynomial $f$. Let $(C_1, \ldots, C_s)$ be a $\tau$-semantic partition of the multiplication gates in $C$ of minimal semantic rank. Similarly let $(D_1, \ldots, D_{s'})$ be a $\tau$-semantic partition of  $D$ of minimal semantic rank.
		Then $s=s'$ and there is a permutation $\pi$ such that for every $i\in[s]$, $[C_i]=[D_{\pi(i)}]$.	
	\end{corollary}
	
	\begin{proof}
		The claim follows immediately from \autoref{cor:max-clus-min-rank} and  \autoref{lem:lower-rank-finer-partition}.
	\end{proof}


	\subsubsection{An Algorithm for Computing Partitions}
	
	In this section we provide an algorithm for constructing $(\tau,r)$ partitions (either syntactic or semantic).
	
	In \autoref{lem:syntactic-clustering} we saw an algorithm,  by Karnin and Shpilka \cite{KarninShpilka09}, for computing $(\tau,r)$-syntactic partitions.
	We now provide an algorithm for constructing $(\tau,r)$-semantic partitions. Our algorithm is quite straightforward, it  simply checks all possible partitions and picks the best one. The guarantee on the output of the algorithm follow from \autoref{lem:syntactic-clustering} and the relations between syntactic and semantic rank.

	We first note that computing the semantic rank of a polynomial $f$ in $\Sps{k}$ can be done in randomized polynomial time given black box access to $f$.
	
	\begin{lemma}
		\label{lem:compute-semantic-rank}
		There exists a randomized polynomial time algorithm that, given black box access to a polynomial $f \in \Sps{k}$ computes $\ranksem(f)$.
	\end{lemma}
	
	\begin{proof}
		We start by using the Kaltofen-Trager black box factorization algorithm \cite{KT90} in order to factor out the linear factor of $f$. Write $f=P \cdot g$ where $P$ is the product of the linear factors of $f$. Then $\ranksem(f)=\rank(M_g)$ (recall \autoref{def:partial-derivative-matrix}). Thus it remains to compute the rank of the matrix $M_g$, which is the dimension of the set of first order partial derivatives of $g$. This can be done in randomized polynomial time as in (for example) Lemma 4.1 of \cite{Kayal11}.
	\end{proof}
		
	We now describe the semantic clustering algorithm.
	
	\begin{algorithm}[H]
		\caption{: Semantic clustering algorithm}
		\begin{algorithmic}[1]
			\Require{White box access to a $\Sps{k}$ circuit, $C$, and a parameter $\tau$.}
			\Ensure{A $(\tau,r)$ partition of $C$.}
			\State{$r=\infty$}
			\For{every partition $(C_1, \ldots, C_s)$ of the $k$ multiplication gates of $C$} \label{line:part-mainloop}
			\State{Compute the semantic ranks $r_i$ of $[C_i]$ and the distances $d_{i,j} = \ranksem([C_i], [C_j])$}
			\State{Set $r'=\max_{i \in [s]} r_i$}
			\If{for all $i,j$, $d_{i,j} \ge \tau r'$ and $r' < r$}
			\State{Set $r:=r'$ and save the partition}
			\EndIf
			\EndFor
			\State{Return $r$ and the saved partition}
		\end{algorithmic}
		\label{algo:semantic-clustering-algorithm}
	\end{algorithm}
	
	Recall that \autoref{lem:syntactic-clustering} implies that the Karnin-Shpilka \emph{syntactic} clustering algorithm returns \emph{syntactic} clusters. This allows us to obtain some guarantees on the output of \autoref{algo:semantic-clustering-algorithm}.
	
	\begin{claim}
		\label{cl:semantic-clusters}
		For every $\tau$,
		\autoref{algo:semantic-clustering-algorithm} runs in time at most $2^{k^2} \cdot \poly(n)$ and outputs a $(\tau,r)$ partition where 
		\[
		r \le R_M(2k) \cdot k^{{\lceil \log_k (\tau \cdot 2^7 k^2 \log k) \rceil} \cdot (k-2)}\leq R_M(2k) \cdot 2^{7k} k^{4k} \tau^{k-2}.
		\]
	\end{claim}
	
	\begin{proof}
		By \autoref{lem:syntactic-clustering} with $\kappa=\tau \cdot 2^7 k^2 \log k$, 
		there exist a $(\kappa,r)$ \emph{syntactic} partition $(C_1, \ldots, C_s)$ where $r \le R_M(2k) \cdot k^{{\lceil \log_k (\kappa) \rceil} \cdot (k-2)}\leq  (k\kappa)^{k-2} \cdot R_M(2k)$. \autoref{cor:syn-is-sem} implies that it is also a $(\tau,r)$ \emph{semantic}-partition. Thus, the algorithm will find at least one $\tau$-semantic partition and it clearly outputs the one with the minimal rank.
		The statement on the running time follows from \autoref{lem:compute-semantic-rank}.
	\end{proof}

	\begin{remark}
		\label{rem:alg-semantic-unique}
		Note that \autoref{cor:unique-semantic} shows that the semantic partition with the minimal semantic rank is unique regardless of the representation. Hence the output of \autoref{algo:semantic-clustering-algorithm} does not depend on the circuit $C$ but only on the polynomial $f$ it computes.
	\end{remark}
	
	
	\subsubsection{Semantic Partitions under Restrictions}
	
	\autoref{cor:unique-semantic} proves that any maximal semantic partition is unique. However, in our reconstruction algorithm we shall consider restrictions of the unknown polynomial to subsets of the variables. Hence, we will need a stronger property that is analogous to the one proved in \autoref{thm:unique-syn-part}.
	We start with the obvious fact that restrictions can't increase the semantic rank (recall \autoref{def:restriction}).
	
	\begin{claim}
		\label{cl:restriction-rank}
		If $\ranksem(f|_{B,\va})\geq t$ then $\ranksem(f)\geq t$.
	\end{claim}
	
	The next claim shows that for every multilinear polynomial, $f\in\Sps{k}$, there exists a $({\tau_1},r)$-semantic partition with the special property that its rank bound, $r$, is upper bounded as a function of $\tau_0$, which is much smaller than ${\tau_1}$.
	The claim and its proof are completely analogous to that of \autoref{cl:partition-two-taus-syn} and so we omit the proof.
	
	\begin{claim}
		\label{cl:partition-two-taus}
		For every function $\varphi: \mathbb{N} \to \mathbb{N}$, every $\tau_\text{min} \in\mathbb{N}$ and for every multilinear $f \in \Sps{k}$ there is $\tau_\text{min}\leq \tau_0 \le \tau(k) = R_M(2k) ^{ \varphi(k)^k}$ such that there is a $({\tau_1}, r)$-semantic partition of $f$ with:
		\begin{itemize}
			\item ${\tau_1} = \tau_0^{\varphi(k)}$.
			\item $r \le  R_M(2k) 2^{7k} k^{4k} \tau_0^{k-2}$. 
		\end{itemize}
	\end{claim}
	
	\autoref{cl:partition-two-taus} implies a useful inclusion property of clusters in any partition of a restriction of $f$.
	
	\begin{claim}
		\label{cl:cluster-refinement-in-B}
		Let $f\in\Sps{k}$  be a multilinear polynomial. Let $f = \sum_{i=1}^s f_i$ be the semantic partition guaranteed in \autoref{cl:partition-two-taus} for  $\varphi(k)\geq 1$ and $\tau_\text{min}=R_M(2k) \cdot 2^{7k+12} \cdot k^{4k+3}$, with parameters $(\tau_0, {\tau_1})$ and rank $r$. Set $\tau=\tau_0^k$ as in \autoref{thm:unique-syn-part}.  
		
		Let $B \subseteq [n]$ and $\va \in \F^n$. Let $f|_{B,\va}=\sum_{i=1}^{s'} g_i $ be a maximal ${\tau}$-semantic partition of $f|_{B,\va}$, with rank $r_B$ and let $D$ be a minimal circuit computing $f|_{B,\va}$ with $(D_1, \ldots, D_{s'})$ its corresponding partition.
		
		Then, $s' \le s$ and for every $i \in [s']$ there's a subset $S_i \subseteq [s]$ such that $D_i = \sum_{j \in S_i} (C_j)|_{B,\va}$.
	\end{claim}
	
	\begin{proof}
		Let $C=\sum_{i=1}^{s}C_i$ be a $\Sps{k}$ realization of the partition  $f = \sum_{i=1}^s f_i$, where $[C_i]=f_i$. \autoref{rem:alg-semantic-unique} and  \autoref{cl:partition-two-taus} imply that $ r \leq  R_M(2k) 2^{7k} k^{4k} \tau_0^{k-2}$.
		
		\sloppy As $C|_{B,\va}$ computes $f|_{B,\va}$, \autoref{cor:max-clus-min-rank} and \autoref{lem:lower-rank-finer-partition}  imply that there is a partition $C|_{B,\va}=\sum_{i=1}^{s'}D_i$  such that $[D_i]=g_i$. \autoref{cor:sem-is-syn} implies that this is also a $(\tau_\text{syn}, r_{B,\text{syn}})$-syntactic partition for $\tau_{\text{syn}} := \tau/(2^7 k^2 \log k)$ and $r_{B,\text{syn}}:=r_B\cdot 2^7 k^2 \log k$.
		
		Let $M_1, M_2$ be two multiplication gates appearing in the same cluster $C_1$. Assume towards contradiction that their restrictions to $(B,\va)$ appear in different clusters $D_1, D_2$ (without loss of generality).
		\autoref{lem:large-rank-syn} and \autoref{lem:small-semantic-implies-small-syntactic-rank} imply that 
		\begin{align*}
			\tau /10 &\leq  \tau r_B/10  = \tau_\text{syn} r_{B,\text{syn}}/10\leq  \ranksyn((M_1)|_{B,\va}+(M_2)|_{B,\va}) \leq \ranksyn((C_1)|_{B,\va}) \\ &\leq \ranksyn(C_1) \leq 2^7 k^2 \log k \ranksem(C_1) \leq 2^7 k^2 \log k \cdot  r  \\ & \leq 2^7 k^2 \log k \cdot  R_M(2k) \cdot 2^{7k} k^{4k} \tau_0^{k-2} \leq  R_M(2k) \cdot 2^{7k+7} k^{4k+3} \tau_0^{k-2}.
		\end{align*}
		This contradicts the choice of $\tau_\text{min}$ and $\tau$.
	\end{proof}


	
	\section{Learning Low Degree Polynomials}
	\label{sec:lowdegree}

	As in \cite{BSV21}, we start by providing an algorithm, which is efficient when the degree $d$ is very small. 
	
	\subsection{Reconstruction Algorithm for Low Degree Multilinear $\SumProdSumk$ circuits}
	\label{sec:sumprodsumk-lowdeg}
	
	Bhargava, Saraf and Volkovich proved the following lemma.
	
	\begin{lemma}[Lemma 6.4 in \cite{BSV21}]
		\label{lem:BSV-low degree}
		Let $f \in \Fn$ be a polynomial computed by a degree $d$, multilinear $\Sps{k}$ circuit $C_f$ of the form 
		\begin{equation*}
			\sum\limits_{i = 1}^k T_i(X)= \sum_{i= 1}^k
			\prod_{j=1}^{d_i}\ell_{i,j}(X).
		\end{equation*} 
		Then, there is a randomized algorithm that given $k,d$ and black-box access to $f$ outputs a multilinear $\Sps{k}$ circuit computing $f$, in time $\poly(c) \cdot (dkn)^{{O(d^2k^3)}^{d^2k^2}}$, 	where $c=\log q$ if $\F=\F_q$ is a finite field, and $c$ equals the maximum bit complexity of any coefficient of $f$ if $\F$ is infinite.
	\end{lemma}
	
	As before, we'd like to make the dependence on $n$ polynomial, regardless of $d$ and $k$. To explain the required changes, we start by sketching the proof of \autoref{lem:BSV-low degree}.
	
	\begin{proof}[Proof sketch of \autoref{lem:BSV-low degree}]
		The proof follows the following steps.
		\begin{enumerate}
			\item \label{item:var-reduction} Variable reduction: Denote by $m$ the number of essential variables of $f$. As there are at most $kd$ linear functions appearing in the circuit, it holds that $m \leq kd$.
			
			By Lemma~\ref{lem:essential-vars}, there is a polynomial-time randomized algorithm that, given black-box access to $f$, computes an invertible linear transformation $A \in \F^{n \times n}$ such that $f(A  \vx)$ depends only on the first $m$ variables. Apply this algorithm to obtain $A$, and denote $g:=f(A\vx)$.
			
			\item \label{item:sparse-interpolate} Learn $g$. As $g$ has $m\leq kd$ variables and degree at most $d$, we can interpolate in order to find a representation $g(x) = \sum_{ \ve : |\ve| \leq d} c_{\ve}\cdot  \vx^{\ve}$. This is done in time $\poly(\binom{m+d}{d})\leq \poly(\binom{kd+d}{d})$.
			
			\item \label{item:require-sumprodsum} Add the requirement that $g$ has a $\SumProdSumk$ representation. To do so, consider the representation 
			\[\sum_{i=1}^k \prod_{j=1}^d (a^{(i)}_{j,1}x_1 + a^{(i)}_{j,2}x_2 + \ldots + a^{(i)}_{j,m}x_m + a^{(i)}_{j, m+1}) = g = \sum_{ \ve } c_{\ve}\cdot  \vx^{\ve}.\] 
			We view this as a set of at most $\binom{kd+d}{d}$ polynomial equations in  $kd(m+1)\leq 2k^2d^2$ variables, $\{a^{(i)}_{j,t}\mid i\in [k], j\in[d], t\in[m+1]\}$.
			
			\item\label{item:lifting-multilinear} Make sure that $f=g(A^{-1}x)$ is a multilinear polynomial (``lifting''). Recall that we have $A$, thus we can compute $A^{-1}$. Let $L_t = \langle R_t, \vx \rangle$ where $R_t$ is the $t$-th row of $A^{-1}$. $L_t$ is a linear function in $\vx$, i.e., $L_t = \sum_{p\in [n]}\alpha^t_p x_p$ for some coefficients $\alpha^t_p$. 
			
			For every $p\in [n], j\in [k]$, make sure that the degree of $x_p$ in the $j$-th product gate is at most one. We can do this by adding a set of polynomial equations that guarantee that the product of coefficients of $x_p$ in any two linear forms appearing in the $j$-th product gate is $0$. This adds $k\cdot n \cdot d^2$ equations in the variables $\{a^{(i)}_{j,t}\mid i\in [k], j\in[d], t\in[m+1]\}$.
			
			\item \label{item:solve-poly-eqns} Solve the polynomial equations. In total we have $k\cdot n \cdot d^2+ \binom{kd+d}{d}$ many equations in $2k^2d^2$ variables. As mentioned in \autoref{sec:sys}, this is solvable in time
			\[
			\poly{\left(n,\Sys(2d^2k^2, kd^2n+ \binom{dk+d}{d},d)\right)} \leq  (dkn)^{{O(d^2k^3)}^{2d^2k^2}}.\qedhere
			\]
		\end{enumerate}
	\end{proof}
	
	The following lemma improves the time complexity of \autoref{lem:BSV-low degree}.
	
	\begin{lemma}
		\label{lem:low degree}
		Let $f \in \Fn$ be a polynomial computed by a degree $d$, multilinear $\Sps{k}$ circuit $C$ of the form 
		\begin{equation*}
			\sum\limits_{i = 1}^k T_i(X)= \sum_{i= 1}^k
			\prod_{j=1}^{d_i}\ell_{i,j}(X).
		\end{equation*} 
		Then, there is a randomized algorithm that given $k,d$ and black-box access to $f$ outputs a multilinear $\Sps{k}$ circuit computing $f$, in time $ 	\poly\left(n, c, (dk)^{{O(d^3k^2)}^{2d^2k^2}}\right)$,
		where $c=\log q$ if $\F=\F_q$ is a finite field and $c$ is the maximal bit complexity of the coefficients of $f$ if $\F$ is infinite.		
	\end{lemma}
	
	As we only change one step in the algorithm of \cite{BSV21}, we will describe the change and its effect on the time complexity.

	\begin{proof}
		
		We only change the behavior  of Step~\ref{item:lifting-multilinear} in \autoref{lem:BSV-low degree}. In this step we make sure that the ''lifting" of the polynomial $g$ is indeed a multilinear polynomial. This is done by adding a set of $k\cdot d^2\cdot n$ polynomial equations that make sure that the degree of each variable in each product gate is at most one.    
		To ease notation let us assume that we want to enforce that the individual degree of $x_1$ is 1. For $t\in [m]$, denote the coefficient of $x_1$ in $L_t$ by $\alpha^t_1$ (recall that $L_t$ is the linear function corresponding to the $t$-th row of $A^{-1}$). Note that the algorithm knows $\alpha^1_1,\ldots,\alpha^m_1$. Therefore, when substituting $A^{-1} \vx$ in the circuit, the coefficient of $x_1$ in the $j$-th linear function of the $i$-th multiplication gate, $\ell_{i,j}$, is $\sum_t a^{(i)}_{j,t} \alpha^t_1$. As we need $x_1$ to appear in at most one of the linear functions $\ell_{i,1},\ldots, \ell_{i,d_i}$, it is enough to require that for every $j_1\neq j_2\in [d_i]$, 
		\[(\sum_{t=1}^m a^{(i)}_{j_1,t} \alpha^t_1) \cdot (\sum_{t=1}^m a^{(i)}_{j_2,t} \alpha^t_1) =0.\]
		
		This is a quadratic equation in the coefficients $\set{a^{(i)}_{j,t}}$. 
		As the set of quadratic polynomials in $mkd = k^2d^2$ variables has dimension at most $k^4d^4$, we can find a basis to the set of equations (in time $\poly(n,k,d)$) and add only the equations in the basis to the set of our polynomial constrains, thus adding at most $k^4d^4$ polynomial equations.
		
		\sloppy
		Observe that any solution to the new system will have the property that the lift will be multilinear. Moreover, this system is solvable: as noted by \cite{BSV21}, this follows since the ``natural'' circuit $C_g$ for $g$, which is obtained from the original multilinear circuit for $f$ by applying $A \vx$ to the inputs, has the property that if we replace (in $C_g$) the input $\vx$ by $A^{-1} \vx$, then we get a multilinear depth-$3$ circuit for $f$. 
		
		Hence, the system is solvable in time $\Sys(mdk, k^4d^4+ \binom{dk+d}{d},d)$, and the overall time complexity of the algorithm is bounded by $\poly{\left(n, c, \Sys(d^2k^2, k^4d^4+ \binom{dk+d}{d},d)\right)} \leq  \poly\left(n,c, (dk)^{{O(d^3k^2)}^{d^2k^2}}\right)$.
	\end{proof}
	
	Similarly to Lemma 6.5 in \cite{BSV21}, we can use our improved running time in order to learn circuits with low \emph{syntactic} rank.
	
	\begin{lemma}[Similar to Lemma 6.5 in \cite{BSV21}]
		\label{lem:low rank exact}
		Let $f \in \Fn$ be a polynomial computed by multilinear $\SumProdSumk$ circuit $C$ with $\ranksyn(C) \leq r$.
		Then, there is a randomized algorithm that given $k,r$ and black-box access to $f$, outputs a multilinear $\Sps{k'}$ circuit computing $f$, where $k' \leq k$ is the smallest possible fan-in,
		in time \sloppy$ \poly(n, c, (rk)^{{O(r^3k^2)}^{r^2k^2}})$.
	\end{lemma}
	
	This lemma also enables us to learn circuits computing polynomials of low \emph{semantic} rank.
	
	\begin{lemma}
		\label{lem:low semantic rank exact}
		Let $f \in \Fn$ be a polynomial computed by multilinear $\SumProdSumk$ circuit $C$ with $\ranksem(f) \leq r$.
		Then there is a randomized algorithm that given $k,r$ and black-box access to $f$ outputs a multilinear $\Sps{k'}$ circuit computing $f$, where $k' \leq k$ is the smallest possible fan-in,
		in time \sloppy$ \poly(n, c, (rk)^{{(rk)}^{\poly(r,k)}})$.
	\end{lemma}
	
	\begin{proof}
		By \autoref{lem:small-semantic-implies-small-syntactic-rank}, there is a circuit $C$ computing $f$ with $\ranksyn(C) \le 2^7 k^2 \log k \cdot r$. We now apply \autoref{lem:low rank exact}. Note that since the algorithm is black box we don't need to actually know $C$.
	\end{proof}

	\subsection{Reconstruction of Low-Degree Depth-$3$ Set-Multilinear Circuits}
	Since set-multilinear depth-$3$ circuits are a special case of depth-$3$ multilinear circuits, it's clear that
	given a black box access to a set-multilinear polynomials $f(\vx_1, \ldots, \vx_d)$ computed by a depth-$3$ set-multilinear circuit of top fan-in $k$, the algorithm in \autoref{sec:sumprodsumk-lowdeg} is able to reconstruct $f$. However, as stated it's possible that the algorithm outputs a multilinear circuit rather than a set-multilinear circuit, and thus doesn't give a proper learning algorithm.
	In this section we explain how to modify the proof of \autoref{lem:BSV-low degree} to output a set-multilinear circuit.
	
	Let $\vx=\vx_1 \bigcup \vx_2 \cdots \vx_d$ denote the set of variables and let $n_i = |\vx_i|$. 
	
	We first observe that in \autoref{item:var-reduction} we can find a matrix $A$ that ``respects'' the partition $\vx_1, \ldots, \vx_d$: that is, $A$ is a direct sum of matrices $A_1, \ldots, A_d$ such that $A_i$ operates on $\vx_i$ and $f(A\vx)$ depends on at most $kd$ variables.
	As noted in \autoref{lem:det-essential-vars}, finding, e.g., $A_1$ amounts to finding a basis to the vector space
	\[
	\set{ \va \in \F^{n_1} : \sum_{i=1}^{n_1} \frac{\partial f}{\partial x_{1,i}} (\vx) = 0}.
	\]
	By \autoref{cor:poly-size-hitting-set-set-ml} we can find such $A_1$ deterministically in polynomial time (for small, super-constant $k$), so that now the polynomial depends on at most $m_1 \le k$ variables from $\vx_1$. We then find $A_2, A_3, \ldots, A_d$.
	
	The next step, \autoref{item:sparse-interpolate} is done as in \autoref{lem:BSV-low degree}.
	
	We change \autoref{item:require-sumprodsum} to require that $f$ has a depth-$3$ set-multilinear representation of top fan-in $k$. This is simply done by changing the system of polynomial equations to
	\[
	\sum_{i=1}^k \prod_{j=1}^d (a^{(i)}_{j,1} x_{j,1} + \cdots + a^{(i)}_{j,n_j}x_{j,n_j}+ a^{(i)}_{j,m_j+1}) = g = \sum_{ \ve } c_{\ve}\cdot  \vx^{\ve}.
	\]
	For $m=\max_{j} m_j \le k$, this is a set of at most $\binom{kd+d}{d}$ polynomial equations in  $kd(m+1)\leq 2k^2d^2$ variables. Solving the system can be done in time 
	$\Sys(d^2k^2, \binom{dk+d}{d},d)$.

	Note that \autoref{item:lifting-multilinear} and \autoref{item:solve-poly-eqns} are now redundant. By the structure of $A$, we can simply apply $A^{-1}$ (which replaces every variable $x_{i,j} \in \vx_i$ by a linear function in $\vx_i$). This is because the set multilinear circuits only multiplies linear functions in distinct sets in the partition $\vx_1, \ldots, \vx_d$, so that $f(A^{-1} \vx)$ is set multilinear.
	
	As a result of this discussion we obtain the corollary.
	
	\begin{corollary}
		\label{cor:set-ml-low-deg}
		Let $f(\vx) \in \F[\vx]$ be a set-multilinear polynomial computed by a degree $d$, set-multilinear depth-$3$ circuit. Suppose $\vx=\vx_1 \cup \cdots \cup \vx_d$ and $|\vx_i| \le n$ for all $i$.
		Then, there is a randomized algorithm that given $n,k,d$ and black-box access to $f$ outputs a set-multilinear depth-$3$ circuit with top fan-in $k$ that computes $f$, in time $\poly\left(n, c, (dk)^{{O(d^2k^3)}^{d^2k^2}}\right)$.
	\end{corollary}

	The following is analogous to \autoref{lem:low rank exact}.
	
	\begin{lemma}
		\label{lem:low rank set-ml exact}
		Let $f \in \Fn$ be a polynomial computed by set-multilinear $\SumProdSumk$ circuit $C$ with $\ranksem(C) \leq r$.
		Then, there is a randomized algorithm that given $k,r$ and black-box access to $f$ outputs a set-multilinear $\SumProdSumk$ circuit computing $f$, where $k' \leq k$ is the smallest possible fan-in,
		in time \sloppy$ \poly\left(n,c,(rk)^{{O(r^3k^2)}^{r^2k^2}}\right)$.
	\end{lemma}

	\section{Efficient Construction of Cluster Preserving Sets}
	\label{sec:cluster-preserving}
	
	In order to reconstruct general multilinear $\Sps{k}$ circuits,  we would again like to follow the steps of \cite{BSV21}. However, as some of our definitions are different, and we replace some brute force steps with algorithmically efficient steps, we're required to make substantial changes in the algorithm. In particular we replace their use of the notion of ``rank preserving subspaces'' with an explicit construction of a subset $B$ of the variables that, in some sense, preserves the structure of semantic clusters of $f$.
	
	The following algorithm attempts to construct a set $B$ together with a vector $\va$ such that the clusters of $f|_{B,\va}$ (with respect to a certain semantic $\tau$-partition), found by \autoref{algo:semantic-clustering-algorithm}, are in one-to-one correspondence with the clusters that the same algorithm would have outputted on $f$. Our algorithm receives $\tau$ as a parameter ($\tau$ will be related to the parameters from \autoref{cl:partition-two-taus}).

	\begin{algorithm}[H]
		\caption{: Randomized construction of a cluster-preserving set}
		\begin{algorithmic}[1]
			\Require{Black box access to a degree $d$, $n$-variate polynomial $f$, computed by a $\Sps{k}$ circuit, $C$, and a parameter $\tau$.}
			\Ensure{A subset of the variables $B$ and an assignment $\va \in \F^n$ (whose properties are specified in \autoref{cl:tau-semantic-restriction})}
			\State{Let $S \subseteq \F$ an arbitrary set of size $n^{k^{k^{O(k)}}}$}
			\State{Set $B = \emptyset$, $s=0$, $f_{curr}=\sum_{j=1}^s C_j = 0$}
			\State{Pick at random $\va \in S^n$.}
			\For{every $I  \subseteq \set{x_1, \ldots, x_n} \setminus B$ of size at most 4} \label{line:spsk-mainloop}
			\For{$k'\in[k]$}\label{line:spsk-minimalloop}
			\State{Using \autoref{lem:low semantic rank exact} with $k=k'$, learn $f_I = \fb{B\cup I}{\va}$. Using randomized polynomial identity testing, make sure that indeed $f_I \equiv \fb{B \cup I}{\va}$, and if so, go to the next line}\label{step:learn-C}
			\EndFor
			\State{Run \autoref{algo:semantic-clustering-algorithm} on $f_I$ with parameter $\tau$ to get $f_I = \sum_{j=1}^{s_I}C'_j$}\label{step:cluster}
			\If{$s_I \neq s$} \label{line:cluster-increase}
			\State{Set $B=B \cup I$, $s=s_I$, $f_{curr}=f_I$, and save $(C'_1,\ldots, C'_{s_I})$. Restart the main loop in line \ref{line:spsk-mainloop}}
			\EndIf
			\For{$j\in [s]$}
			\State{Find $\sigma(j)\in [s]$ s.t. $C'_j|_{x_I=\va_I} = C_{\sigma(j)}$ using a randomized PIT algorithm}
			\State{Calculate $r'_j = \ranksem(C'_j), r_j = \ranksem(C_{\sigma(j)})$}
			\If{$r'_j > r_j$} \label{line:rank-increase}
			\State{Set $B=B \cup I$, $s=s_I$, $f_{curr}=f_I$, and save $(C'_1,\ldots, C'_{s_I})$. Restart the main loop in line \ref{line:spsk-mainloop}}
			\State{Abort in case any subprocedure failed during the execution of the algorithm}
			\EndIf
			\EndFor
			\EndFor
			\State{Return $B$ and $\va$.}
		\end{algorithmic}
		\label{algo:finding-B}
	\end{algorithm}
	
	We now explain what guarantees we get on the outputs $(B,\va)$ of \autoref{algo:finding-B}.
	We first bound the running time of the algorithm.
	
	\begin{claim}
		\label{cla:size-of-B}
		\autoref{algo:finding-B}, when given the parameter $\tau$ guaranteed in \autoref{cl:tau-semantic-restriction} as input, runs in time $\poly(n) \cdot k^{k^{k^{k^{\poly(k)}}}}$ and returns a set $B$ of size at most $k^{k^{O(k)}}$.
	\end{claim}

	The following important claim shows that if $(B,\va)$ is the output of \autoref{algo:finding-B}, then for the ``correct'' choice of $\tau$, $f|_{B,\va}$ preserves the clusters (with respect to a $\tau$-semantic partition) of $f$.

	\begin{claim}
		\label{cl:tau-semantic-restriction}
		Let $f \in \Sps{k}$ be a  multilinear polynomial and let $C$ be a minimal multilinear $\Sps{k}$ computing $f$. 
		There exists a non-zero polynomial $\Gamma_C$ of degree at most $n^{k^{k^{O(k)}}}$ such that if $B,\va$ are the outputs of \autoref{algo:finding-B} on $f$, and $\Gamma_C(\va) \neq 0$, then the following holds:		
		Consider the semantic partition of $f$, $f = \sum_{i=1}^s f_i$,  given by \autoref{cl:partition-two-taus} with $\varphi(k) = k^2$ and $\tau_\text{min}=R_M(2k) \cdot 2^{7k+20} \cdot k^{4k+4}$. Let $\tau_0, \tau_1, r$ be its parameters as promised by the claim and let $\tau=\tau_0^{k}$.
		Let $D$ be a minimal multilinear $\Sps{k}$ circuit computing $f|_{B,\va}$. Then, the output of \autoref{algo:semantic-clustering-algorithm} on $D$ with parameter $\tau$, denoted by $[D] = \sum_{i=1}^{s'} g_i$, satisfies:
		\begin{enumerate}
			\item $s'=s$.
			\item There is a permutation $\pi$ on $[s]$ such that $g_{\pi(i)} = (f_i)|_{B,\va}$.
			\item $\ranksem(g_{\pi(i)}) = \ranksem(f_i)$.
		\end{enumerate}
		In particular, the $g_i$'s also form a $(\tau,r)$ partition.
	\end{claim}

	The rest of this section is devoted to proving \autoref{cla:size-of-B} and \autoref{cl:tau-semantic-restriction}.

	\begin{remark}\label{rem:no-derand}
		The fact that $\deg(\Gamma_C)=n^{k^{k^{O(k)}}}$ is the bottleneck in derandomizing our algorithms. For derandomization we shall have to find $\va$ such that $\Gamma_C(\va)\neq 0$, and it is not clear how to achieve this in time $\poly(n,F(k))$ (i.e. not have $k$ in the exponent of $n$).
	\end{remark}

	\subsection{Proof of \autoref{cla:size-of-B}}
	
	Let $f$ be the multilinear polynomial in question.
	Let  $C$ be a multilinear $\Sps{k}$ circuit computing $f$. 
	Consider a partition of $f$, $f = \sum_{i=1}^s f_i$, as given by \autoref{cl:partition-two-taus} with $\varphi(k) = k^2$. Let $f = \sum_{i=1}^s f_i$ the $\tau=\tau_0^k$-semantic partition of minimal rank and let $C = \sum_{i=1}^s C_i$ the corresponding partition of the multiplication gates in $C$.

	\begin{claim}
		\label{cl:rank-of-subcluster}
		Let $B \subseteq [n]$ and $\va \in \F^n$.  Consider a ${\tau}$-semantic partition of $g=f|_{B,\va}$, $g=\sum_{i=1}^{s'} g_i$. Observe that by \autoref{cl:cluster-refinement-in-B}, there is $\cI \subseteq [s]$ such that $g_1 = \sum_{i \in \cI} (f_i)|_{B,\va}$. Then, for every $\cJ \subseteq \cI$ it holds that
		\[
		\ranksem(\sum_{j \in \cJ} (f_j)|_{B,\va}) \le 2^7 k^2 \log k \cdot \ranksem(g_1).
		\]
	\end{claim}

	\begin{proof}
		Let $D=C|_{B,\va}$ be the restriction of $C$ to $(B,\va)$, which is a multilinear $\Sps{k}$ circuit computing $g$. By \autoref{rem:alg-semantic-unique}, there is a subcircuit $D_1$ such that $[D_1]=g_1$.
		Note that
		\begin{align*}
			\ranksem \left( \sum_{j \in \cJ} (f_j)|_{B,\va} \right) & \le \ranksyn \left(\sum_{j \in \cJ} (C_j)|_{B,\va} \right) \le \ranksyn \left(\sum_{i \in \cI} (C_i)|_{B,\va} \right) \\
			&= \ranksyn(D_1) \le 2^7 k^2 \log k \cdot \ranksem(g_1).
		\end{align*}
		where the first inequality follows from \autoref{obs:semantic-vs-syntactic}, the second from monotonicity of the syntactic rank, and the third inequality follows from \autoref{lem:small-semantic-implies-small-syntactic-rank}.
	\end{proof}
	
	Denote by $B_i$ the set obtained in the $i$-th iteration of the main loop in \autoref{algo:finding-B}. Associate with $B_i$ a partition $\pi(B_i)$ of $[k]$, $[k]=S_1\sqcup S_2 \ldots\sqcup S_{s_i}$  such that if we denote $g_{i,j}=\sum_{e\in S_j} f_e|_{B_i,\va}$, then $g_i=\sum_{j=1}^{s_i}g_{i,j}$ is a $\tau$-semantic partition of $g_i=f|_{B_i,\va}$. We denote by $r_i$ the rank of this partition.

	The following claim provides a ``progress measure'' that bounds the running time of \autoref{algo:finding-B}.
	
	\begin{claim}
		\label{cl:algo-progress-measure}
		
		Let $t < t'$ such that $\pi(B_t) = \pi(B_{t'})$. 
		Then, there exists $i \in [s_{t}]$ and $\cJ \subseteq S_i$ such that
		\[
		\ranksem \left( \sum_{j \in \cJ} (f_j)|_{B_t,\va} \right) < \ranksem \left( \sum_{j \in \cJ} (f_j)|_{B_{t'},\va} \right).
		\]
	\end{claim}

	\begin{proof}
		
		Observe that since 	 $\pi(B_t) = \pi(B_{t'})$ then for every $i\in [s_t]=[s_{t'}]$ it holds that $g_{t,i}= g_{t',i}|_{B_t,\va}$. For the rest of the proof we shall denote $g_i=g_{t',i}$ and $s'=s_t=s_{t'}$.
		
		We split the proof into three cases according to the relation between the rank parameters of the partitions obtained in various stages of the algorithm.
		
		We first observe that if 
		$(g_1)|_{B_{t+1},\va}, \ldots, (g_{s'})|_{B_{t+1},\va}$ is not a $(\tau, r_t)$ partition then the claim holds. Indeed,  since for every $i \neq j$
		\[
		\ranksem((g_i)|_{B_{t+1},\va}, (g_j)|_{B_{t+1},\va}) \ge \ranksem((g_i)|_{B_{t},\va}, (g_j)|_{B_{t},\va}) \ge \tau \cdot r_t,
		\]
		for it not to be a $(\tau, r_t)$ partition it must be the case that there's some $i \in [s']$ such that $\ranksem((g_i)|_{B_{t+1},\va}) > r_t$ and thus
		\begin{equation}\label{eq:rankt+1}
			\ranksem ((g_i)|_{B_{t'},\va}) \ge \ranksem ((g_i)|_{B_{t+1},\va})  > r_t,
		\end{equation}
		which implies the claim. Hence,  we assume from now on that $(g_1)|_{B_{t+1},\va}, \ldots, (g_{s'})|_{B_{t+1},\va}$  is  a $(\tau, r_t)$ partition and that $r_{t+1} \le r_t$.
		

		If $r_{t+1} = r_t$, then \autoref{cor:unique-semantic} implies that the polynomials $(g_i)|_{B_{t+1}, \va}$ are the unique partition. As the number of clusters did not change and since we added variables to $B_t$, the description of the algorithm implies that the rank of one of the clusters increased, which proves the claim.
		
		Finally, suppose $r_{t+1} < r_t$. Consider the output of \autoref{algo:semantic-clustering-algorithm} at the $(t+1)$-th step. This corresponds to a partition $\pi(B_{t+1})$ of the clusters of $f$. 
		As we showed that $(g_1)|_{B_{t+1},\va}, \ldots, (g_{s'})|_{B_{t+1},\va}$ are a $(\tau, r_t)$ partition, the assumption that $r_{t+1} < r_t$ and 
		\autoref{lem:lower-rank-finer-partition} imply that  $\pi(B_{t+1})$ is a refinement of $\pi(B_t)=\pi(B_{t'})$.

		To ease notation let us denote with $h_1,\ldots,h_{s_{t+1}}$ the polynomials corresponding to the clusters of the partition $B_{t+1}$. I.e., $h_i=\sum_{j\in R_i}f_j|_{B_{t+1},\va}$, where $R_1,\ldots,R_{s_{t+1}}$ are the sets in the partition $\pi(B_{t+1})$.
		In the $t'$-th step, the polynomials $\{h_i|_{B_{t'},\va}:=\sum_{j\in R_i}f_j|_{B_{t'},\va}\}$  no longer form a $(\tau, r_{t+1})$ partition (here we use the fact that $\pi(B_{t'}) = \pi(B_t) \neq \pi(B_{t+1})$), as by \autoref{cla:low-rank-finer-part} a finer partition with the same $\tau$ implies lower rank, which means that the output of \autoref{algo:semantic-clustering-algorithm} at the $t'$-th step would not have been $g_1,\ldots,g_{s'}$.
		Similarly to before, this implies that for some $i \in [s_{t+1}]$,
		\[
		\ranksem(h_i|_{B_{t'}}) > \ranksem(h_i) \ge  \ranksem(h_i|_{B_{t}}).
		\]
		As  $\pi(B_{t+1})$ is a refinement of $\pi(B_t)$, there is some $j$ such that $R_i\subseteq S_j$, and the claim holds. This concludes the proof.
	\end{proof}

	\begin{proof}[Proof of \autoref{cla:size-of-B}]
		Suppose that the algorithm makes $T$ iterations of additions of variables to $B$ (each addition adds at most 4 variables). Recall that at each such iteration the algorithm holds a partition $\pi$ of the clusters of $f$. As there are at most $2^{k^2}$ possible partitions of the clusters (since there are at most $k$ clusters), there exists a partition $\pi$ that is obtained at least $T / 2^{k^2}$ times.
		
		By \autoref{cl:algo-progress-measure} and by another application of the pigeonhole principle, there exists a cluster $i$ and a set $\cJ \subseteq S_i$, such that for at least $T / (2^{k^2} \cdot 2^{k})$ values of $t$, $\ranksem(\sum_{j \in \cJ} (f_j)|_{B_t,\va})$ increases. In particular after that many steps,
		\[
		\ranksem(\sum_{j \in \cJ} (f_j)|_{B_t,\va}) > \frac{T}{2^{k^2+k}}.
		\]
		On the other hand, \autoref{cl:rank-of-subcluster} and \autoref{cl:partition-two-taus} promise that
		\[
		\ranksem(\sum_{j \in \cJ} (f_j)|_{B_t,\va}) \leq  2^7 k^2 \log k \cdot  R_M(2k) \cdot 2^{7k} \cdot k^{4k} \tau^{k-2} \le R_M(2k) \cdot 2^{7k+7}k^{4k+3}\tau^{k-2}. \]
		Hence,
		\begin{equation}
			\label{eq:upper-bound-on-B}
			|B|=4T\leq 4 \cdot 2^{k^2+k} \cdot\ranksem(\sum_{j \in \cJ} (f_j)|_{B_t,\va}) < R_M(2k) \cdot 2^{(k+4)^2}k^{4k+3}\tau^{k-2}.
		\end{equation}
		To complete the proof we recall that by \autoref{cl:partition-two-taus}, 
		\[
		\tau=\tau_0^k \leq \left(R_M(2k)^{{(k^2)}^{k}}\right)^k =k^{ k^{O(k)}}.
		\] 
		Together with \eqref{eq:upper-bound-on-B}, this implies $|B| \le k^{k^{O(k)}}$.
		
		We also note that by \autoref{cl:semantic-clusters} at each step
		\[
		r_i\leq R_M(2k) \cdot 2^{7k} k^{4k} \tau^{k-2} \leq  k^{ k^{O(k)}}.
		\] 
		Plugging in the upper bounds on $r$ and $|B|$ to \autoref{lem:low semantic rank exact} shows that every iteration of Step~\ref{step:learn-C} of \autoref{algo:finding-B} takes at most 
		\[\poly\left(|B|,(rk)^{{(rk)}^{\poly(r,k)}}\right)= k^{k^{ k^{ k^{O(k)}}}}.
		\] 
		The other steps in the algorithm run in time polynomial in $n$ and in smaller factors of $k$, and thus the claim follows.
	\end{proof}
	
	\begin{remark}
		\label{rem:running-time-of-finding-B}
		We have bounded the size of $B$ and the running time only for a certain ``correct'' choice of the parameter $\tau$ in \autoref{algo:finding-B}. In our reconstruction algorithm we run \autoref{algo:finding-B} with many possible choices of $\tau$, since we don't know a-priori which is the correct one. However, it is easy to modify the algorithm so that the upper bounds on the size of $B$ and on the running time always hold, by simply terminating the algorithm if $B$ gets too large or if it runs for too many steps.
	\end{remark}

	\subsection{Proof of \autoref{cl:tau-semantic-restriction} }
	
	\sloppy
	
	The following lemma shows that for a ``good'' choice of $\va$, if the semantic rank of $f|_{B,\va}$ is smaller than the semantic rank of $f$, then one can add a small number of variables to $B$ so that the semantic rank increases. 
	Note that such a statement is fairly easy to prove for \emph{syntactic} rank. Proving it for semantic rank, however, takes a considerable amount of work and is deferred to \autoref{subsec:adding-vars-to-B}. In \autoref{sec:whysemantic} we explain why we had to work with semantic rank rather than with syntactic rank.

	
	\begin{claim}\label{cla:add-variables-to-B}
		Let $f\in \F[x_1,\ldots,x_n]$ and let $B\subseteq\{x_1,\ldots,x_n\}$. There exists a non-zero polynomial $\Psi_B$ of degree at most $6n^7$ with the following property: For every $\va\in \F^n$ such that $\Psi_B(\va) \neq 0$,  if $\ranksem(f|_{B,\va})< \ranksem(f)$ then there exist variables $x,y,z,w \notin B$ such that  $\ranksem(f|_{B,\va})< \ranksem(f|_{B\cup \{x,y,z,w\}, \va})$.
	\end{claim}
	
	We stress that \autoref{cla:add-variables-to-B} holds for \emph{every} set $B$ and the polynomial $\Psi_B$ doesn't depend on $\va$. In \autoref{algo:finding-B}, the set $B$ that the algorithm constructs \emph{does} potentially depend on $\va$. However, we will later assert that since $\va$ is chosen randomly, it is (among other things) not a zero of $\Psi_B$ for \emph{any} set $B$.

	We continue with the proof of \autoref{cl:tau-semantic-restriction}.
	It turns out that it is easier to argue about how \emph{syntactic} rank behaves with respect to restrictions. We next state some lemmas showing that there are simple polynomial conditions such that if $\va$ is not a zero of any of them then fixing some variables according to $\va$ preserves the structure of the circuit.
	
	We start with a well known lemma (see, e.g., Observation 2.1 in \cite{KarninMSV13}) that constructs a polynomial that preserves pairwise linear independence between linear functions. 
	
	\begin{lemma}\label{lem:D-lin-ind}
		Let $L(\vx)=\sum_{i=1}^{n}a_ix_i+a_0$ and $R(\vx)=\sum_{i=1}^{n}b_ix_i+b_0$ be two linearly independent linear functions. Denote $S=\{i\mid a_i\neq 0 \; \text{or}\; b_i\neq 0\}$. Let
		\[D(L,R):=\prod_{i\in S}\left(a_i R(\vx) -b_i L(\vx)\right).
		\]
		Assume that $\va$ satisfies  $D(L(\va),R(\va))\neq 0$. Then, for every $I\subsetneq [n]$, such that $S\not\subseteq I$, it holds that $L|_{\vx_I=\va_I}$ and $R|_{\vx_I=\va_I}$ are linearly independent.
	\end{lemma}
	
	We next define another polynomial condition that allows us to claim that a certain fixing does not hurt the rank too much.

	\begin{claim}
		\label{cl:rank-of-restriction}\label{cl:restriction-partition}
		Let $C$ be a multilinear $\spsk$ circuit, and $\set{C_1, \ldots, C_s}$ be a $(\tau, r)$-syntactic partition of $C$.
		There exists a non-zero $n$-variate polynomial $\Phi_{C}$ of degree at most $n^3k^3$, such that for every $\va \in \F^n$, if $\Phi_{C}(\va) \neq 0$, then:
		\begin{enumerate}
			\item For every $I\subseteq [n]$ and every subcircuit $C'$ of $C$, it holds that $\ranksyn(C'|_{x_I=\va_I}) \ge \ranksyn(C')-|I|$.
			\item For every set $I \subseteq \set{x_1, \ldots, x_n}$
			it holds that $(C_1)|_{\vx_I=\va_I}, \ldots, (C_r)|_{\vx_I=\va_I}$ is a $(\tau-|I|, r)$-syntactic partition of $C|_{\vx_I=\va_I}$.
		\end{enumerate}
	\end{claim}
	
	\begin{proof}
		Let $C=\sum_{i=1}^{k}T_i$, where $T_i=\prod_j\ell_{i,j}(\vx)$ are multilinear multiplication gates, and let $C'=\sum_{i\in S}T_i$ be a subcircuit of $C$, where $S\subseteq [k]$. Assume $\ranksyn(C')=r'$.
		As in the proof of Lemma 6.14 of \cite{BSV21}
		we define the polynomial 
		\[\Phi_C(\vx)=\prod_{i=1}^{k}T_i \cdot \prod_{(i,j)\neq(i',j')}D(\ell_{i,j},\ell_{i',j'}).\]
		Observe that $\deg(\Phi_C)\leq kn\cdot {kn \choose 2} < n^3k^3$.
		
		We first claim that if $\Phi_C(\va) \neq 0$, then for every $I\subseteq [n]$, $\gcd(C'|_{\vx_I=\va_I}) \sim \gcd(C')|_{\vx_I=\va_I}$.
		Indeed, it's clear that for every $\ell \in \gcd(C')$, $\ell|_{\vx_I=\va_I}$ is also in $\gcd(C')|_{\vx_I=\va_I}$. In the other direction, suppose that
		$\ell|_{\vx_I=\va_I}$ is in $\gcd(C')|_{\vx_I=\va_I}$ 
		but $\ell$ isn't in $\gcd(C')$. This implies that there is a multiplication gate $T$ in $C'$ such that $\ell$ doesn't divide $T$.  If $\ell$ is not supported on $I$, then \autoref{lem:D-lin-ind} implies that $\ell|_{\vx_I=\va_I}$ does \emph{not} divide $T|_{\vx_j=a_j}$, which contradicts the assumption that $\ell|_{\vx_I=\va_I}$ is in $\gcd(C')|_{\vx_I=\va_I}$. 
		If $\ell$ is supported on $I$ then, as $\prod_j T_j(\va)\neq 0$, $\ell_{\vx_I=\va_I}\neq 0$. Thus, $\ell$ is restricted to a nonzero constant and does not affect the gcd.

		Having shown that, let $\tilde{C'}$ be the simplification of $C'$. The argument above shows that the only linear functions that ``disappear'' from $\tilde{C'}$ when restricted to $\vx_I=\va_I$, are those that are supported on the variables in $I$.
		Let $\ell_1, \ldots, \ell_{r'}$ be linear functions that span the linear functions in $\tilde{C'}$. Thus, $(\ell_1)|_{\vx_I=\va_I}, \ldots, (\ell_r)_{\vx_I=\va_I}$ span the simplification of $C'|_{\vx_I=\va_I}$. Finally, note that the dimension of $\set{(\ell_1)|_{\vx_I=\va_I}, \ldots, (\ell_{r'})_{\vx_I=\va_I}}$ is at least $r'-|I|$, as it is the intersection of an $r$ dimensional space with a subspace of codimension $|I|$.
		
		To prove the second item, we note that  for every $i \in [s]$, $\ranksyn((C_i)_{\vx_I=\va_I}) \le \ranksyn(C_i) \le r$.
		Further, for every $i \neq i'$, the argument above applied  to $C'=C_i+C_{i'}$, implies that
		\begin{align*}
			\dist(C_i|_{\vx_I=\va_I}, C_{i'}|_{\vx_I=\va_I}) &= \ranksyn(C_i|_{\vx_I=\va_I} + C_{i'}|_{\vx_I=\va_I}) = \ranksyn(C'|_{\vx_I=\va_I})\\ 
			&\ge \ranksyn(C') - |I|  \ge \tau r - |I| \ge (\tau-|I|) r. \qedhere
		\end{align*}
	\end{proof}
	
	\begin{claim}
		\label{cl:anti-restriction-partition}
		Let $C=\sum_{i=1}^{s}C_i$ be a  multilinear $\spsk$ circuit, where each $C_i$ is a sum of one or more multiplication gates. Let $\Phi_C$ be as in \autoref{cl:restriction-partition} and  
		$\va \in \F^n$ be such that $\Phi_C(\va) \neq 0$. Then,
		for every subset $I \subseteq \set{x_1, \ldots, x_n}$ of at most $4$ variables, if  $C|_{\vx_I=\va_I}=\sum_{i=1}^{s}C_i|_{\vx_I=\va_I}$ is a $(\tau,r)$-syntactic partition, then it holds that $\set{C_1, \ldots, C_s}$ is a $(\tau/5, r+4)$-syntactic partition of $C$.
	\end{claim}
	
	\begin{proof}
		\autoref{cl:restriction-partition} gives
		$\ranksyn(C_i) \le \ranksyn((C_i)|_{\vx_I=\va_I})+4 \le r+4$.
		Further, for $i \neq i'$,
		\begin{align*}
			\dist(C_i, C_{i'}) &= \ranksyn(C_i+ C_{i'}) \\
			&\ge \ranksyn((C_i)|_{\vx_I=\va_I} + (C_{i'})|_{\vx_I=\va_I}) \ge \tau r \ge (\tau/5)(r+4). \qedhere
		\end{align*}
	\end{proof}
	
	Finally, we need a simple lemma regarding restrictions that preserve minimality of circuits.
	
	\begin{lemma}
		\label{lem:restriction-minimal}
		Let $C=\sum_{i=1}^{k}T_i$ be a minimal multilinear $\Sps{k}$ circuit, where the $T_i$s are multiplication gates. Assume that $k$ is the minimal integer such that $[C]\in\Sps{k}$. Let $B$  a subset of the variables. Then, there's a polynomial $\Upsilon_{B,C}$ of degree at most $n \cdot 2^{3k}$ with the following property: for every $\va \in \F^n$, if $\Upsilon_{B,C} (\va) \neq 0$ then $f|_{B,\va}$ is minimal. Furthermore, for every $S\neq S'$, $\sum_{i\in S}T_i|_{B,\va} -\sum_{i\in S'}^{k}T_i|_{B,\va}\neq 0$.
	\end{lemma}
	
	\begin{proof}
		Consider the symbolic restriction $C|_{B,\vz}$ over $\F(\vz)$. This is still a minimal circuit since this restriction just amounts to renaming some of the $\vx$ variables to $\vz$. For every non-empty $S \subseteq [k]$, let $C_S$ be the subcircuit of $C|_{B,\vz}$ consisting of the multiplication gates in $S$. $C_S$ computes a non-zero multilinear polynomial in $\vx,\vz$. Similarly, by the assumption on $k$, we have that $C_S-C_{S'}\neq 0$. Thus, for every such $S$ and $S'$ there are polynomials $\Upsilon_S$ and $\Upsilon_{S,S'}$, in the $\vz$ variables, of degree at most $n$, such that if $\Upsilon_S (\va) \cdot \Upsilon_{S,S'}(\va)\neq 0$ then the restriction of $\vz$ to $\va$ preserves the non-zeroness of both $C_S$ and $C_S-C_{S'}$.  Finally take $\Upsilon_{B,C} = \prod_S \Upsilon_S\cdot \prod_{S\neq S'}\Upsilon_{S,S'}$.
	\end{proof}

	We are now ready to prove \autoref{cl:tau-semantic-restriction}.
	
	\begin{proof}
		Let $C\in \spsk$ be the minimal circuit in the statement of the claim such that $[C]=f$ and let $C_i$ be the subcircuit of $C$ computing $f_i$. Thus, $[C|_{B,\va}] = f|_{B,\va}$. 
		Let
		\[
		\Gamma_C = \Phi_C \cdot \left( \prod_B \Upsilon_{B,C} \cdot \left( \prod_{S \subseteq [s]} \Psi_{B,S} \right) \right) \cdot \Xi (\vx)
		\]
		where the product is over all sets $B$ of size at most $k^{k^{O(k)}}$, $\Phi_C$ is as defined in \autoref{cl:restriction-partition}, $\Upsilon_{B,C}$ is as defined in \autoref{lem:restriction-minimal}, $\Psi_{B,S}$ is the polynomial $\Psi_B$ from \autoref{cla:add-variables-to-B} applied to the polynomial $\sum_{i \in S} f_i$, and $\Xi(\vx) := \prod_{i=1}^s f_i \cdot \prod_{i\neq j \in [s]} (f_i + f_j)$ (the fact that $\va$ is not a root of $\Xi(\vx)$ will not be used in this proof but only later, in the proof of \autoref{cl:beta-clusters}).
				
		We start by showing the degree upper bound on $\Gamma_C$. The major contribution to the degree comes from multiplying over all possible subsets $B$, but
		note that by \autoref{cla:size-of-B} 
		we may assume that $|B| \le k^{k^{O(k)}}$. Now observe that the polynomials $\Phi_C$ and $\Xi$ are of degree $\poly(n)$; that for each fixed $B$, the polynomial $\Upsilon_{B,C}$ is of degree $n \cdot 2^{3k}$; and similarly, for each fixed $B$ and $S$, the degree of of $\Psi_{B,S}$ is at most $\poly(n)$.
		 Therefore, since $S$ ranges over all subsets of $[s]$ whose number is at most $2^k$, and since $B$ ranges over all possible subsets of $[n]$ of size $k^{k^{O(k)}}$, the degree of $\Gamma_C$ is at most $n^{k^{k^{O(k)}}}$.
		
		Let $\va$ be such that $\Gamma_C(\va) \neq 0$. Then in particular $C|_{B,\va}$ is also minimal.
		
		Let $[D] = \sum_{i=1}^{s'} g_i$ be the output of \autoref{algo:semantic-clustering-algorithm} on $D$ with parameter $\tau$, where $D$ is a minimal multilinear $\Sps{k}$ circuit computing $f|_{B,\va}$, as in the statement of the claim.
		
		By \autoref{cl:cluster-refinement-in-B}, $s' \le s$ and for every $i \in [s']$ there's a subset $S_i \subseteq [s]$ such that $g_i = \sum_{j \in S_i} (f_j)|_{B,\va}$. Note that if $s'=s$ then the second item of \autoref{cl:tau-semantic-restriction} follows immediately.
		
		Suppose towards contradiction that $s' < s$. Without loss of generality, suppose that $g_1 = \sum_{j \in S_1} (f_j)|_{B,\va}$ where $|S_1|>1$. To reach a contradiction we will show that if this was the case then the algorithm would have continued running and in particular would have added more variables to $B$. 
		
		We start by proving that
		\begin{equation}\label{eq:drankfvsdrankg}
			\ranksem \left( \sum_{j \in S_1} f_j \right) >  \ranksem \left( \sum_{j \in S_1} (f_j)|_{B,\va} \right) = \ranksem(g_1).
		\end{equation}
		This is a consequence of the following set of inequalities. First, 
		\begin{equation}
			\label{eq:lower-bound-on-rank}
			\ranksem \left( \sum_{j \in S_1} f_j \right) \ge \ranksyn \left( \sum_{j \in S_1} C_j \right) / (2^7 k^2 \log k) \ge \tau_1 r / (2^{11} k^2 \log k),
		\end{equation}
		where the first inequality follows from \autoref{lem:small-semantic-implies-small-syntactic-rank}, and the second from the fact that $|S_1| > 1$, the assumption in the statement of \autoref{cl:tau-semantic-restriction}  and \autoref{lem:large-rank-syn}. On the other hand,  \autoref{cl:semantic-clusters} promises that the output of \autoref{algo:semantic-clustering-algorithm} on $D$, with parameter $\tau$, satisfies $\ranksem(g_1) \le R_M(2k) 2^{7k} k^{4k} \tau^{k-2}$,
		which,  by \eqref{eq:lower-bound-on-rank} and our choice of parameters, is indeed smaller than $\ranksem \left( \sum_{j \in S_1} f_j \right)$. Thus by \autoref{cla:add-variables-to-B} there's a set $I$ of at most $4$ variables such that
		\begin{equation}\label{eq:drankfincreased}
			\ranksem \left( \sum_{j \in S_1} (f_j)|_{B \cup I,\va} \right)  > \ranksem \left( \sum_{j \in S_1} (f_j)|_{B,\va} \right).
		\end{equation}
		
		We wish to show that when considering the set $I$,  \autoref{algo:finding-B} will add it to $B$ in contradiction to the fact that it returned $B$ when run on $f$.
		
		Denote with $h = \sum_{i=1}^{s_I} h_i$, the output of \autoref{algo:semantic-clustering-algorithm} on $f|_{B \cup I, \va}$. If $s_I\neq s'$ then \autoref{algo:finding-B} would have added $I$ to $B$. 
		So assume that $s_I=s'$. 
		
		The $h_i$'s form a $\tau$-semantic partition of $f|_{B \cup I, \va}$, and thus, from \autoref{cor:sem-is-syn}, they are also $\kappa$-syntactic partition for $\kappa=\tau/(2^7 k^2 \log k)$. \autoref{cl:restriction-partition} shows that the $h_i|_{\vx_I=\va_I}$'s are a $(\kappa-4)$-syntactic partition of $f|_{B,\va}$ as well. Finally, \autoref{cor:syn-is-sem} shows that they are a $(\kappa-4)/(2^7 k^2 \log k)$-semantic partition of $f|_{B,\va}$. As the $g_i$'s are also a $(\kappa-4)/(2^7 k^2 \log k)$-semantic partition of $f|_{B,\va}$, and $s'=s_I$,  \autoref{cor:same-num-semantic} shows the existence of the a matching between them. That is, there is a permutation $\pi$ of $[s']$ such that $h_i|_{\vx_I=\va_I}=g_{\pi(i)}$.
		
		For simplicity assume that $\pi$ is the identity permutation. I.e., $h_i|_{\vx_I=\va_I}=g_i= \sum_{j \in S_i} (f_j)|_{B,\va}$. Going back to \autoref{algo:finding-B} we see that the algorithm can find the required map $\sigma$ (this is the $\pi$ that we found). 
		
		We now wish to show that $h_1= \sum_{j \in S_1} (f_j)|_{B \cup I,\va} $. If this is not the case then, as before, there is a set $S'_1\neq S_1$ such that  $h_1= \sum_{j \in S'_1} (f_j)|_{B \cup I,\va} $. Hence  
		\[
		\sum_{j \in S'_1} (f_j)|_{B,\va} =h_1|_{B,\va} = \sum_{j \in S_1} (f_j)|_{B,\va}. 
		\]
		Therefore, 
		\[
		\sum_{j \in S'_1} (C_j)|_{B,\va} - \sum_{j \in S_1} (C_j)|_{B,\va}=
		\sum_{j \in S'_1} (f_j)|_{B,\va} - \sum_{j \in S_1} (f_j)|_{B,\va}=0
		\] 
		As $S'_1\neq S_1$, this contradicts the fact that $\Upsilon_{S_1,S'_1}(\va)\neq 0$ (recall \autoref{lem:restriction-minimal}).
		
		Concluding, the algorithm found the partition $h=\sum_{i=1}^{s_I} h_i$ such that $s_I=s'$, a matching between the $h_i$'s and $g_i$'s, and it holds that $h_1= \sum_{j \in S_1} (f_j)|_{B \cup I,\va} $. This means that, when considering the set $I$, the algorithm would add it to $B$, due to \eqref{eq:drankfincreased}, in contradiction.
		
		We have therefore shown that it must be the case that $s=s'$. \autoref{cor:same-num-semantic} implies the second item in the claim, which also implies the third item.
	\end{proof}

		\subsubsection{Proof of \autoref{cla:add-variables-to-B}}
		\label{subsec:adding-vars-to-B}

		We first prove some preliminary lemmas. We start with an obvious observation.
		
		\begin{lemma}
			\label{lem:rank-under-shifts}
			Let $f \in \F[\vx]$ be a polynomial and $\va \in \F^n$. Let $g(\vx)=f(\vx+\va)$. Then, $\rank(M_f) = \rank(M_g)$.
		\end{lemma}

		The next claim shows that we can add two variables to $B$ and increase the rank of the partial derivative matrix (note that this doesn't imply \autoref{cla:add-variables-to-B}, as the rank of a polynomial is the rank of this matrix after pulling out the linear factors).
		
		\begin{claim}
			\label{cl:add-var-to-increase-rank}
			Let $f$ be a multilinear polynomial and $B$ a subset of the variables. 
			Suppose that $\rank(M_f)=r$. There exists a polynomial $\Lambda_B$ of degree at most $n^2$ with the following property: for every $\va \in \F^n$ such that $\Lambda_B(\va)\neq 0$,  if
			$\rank(M_{f|_{B,\va}}) = t <r$, then there are two variables $x_i , w\not \in B$ such that
			\[
			\rank(M_{f|_{B \cup \set{x_i,w},\va}}) >t.
			\]
		\end{claim}
		
		Some preliminary work is required before proving \autoref{cl:add-var-to-increase-rank}.
		A polynomial $P(x) = \sum_{m} c_m \cdot m\in\F^t[\vx]$ is called a vector polynomial. I.e., it is a polynomial whose coefficients $c_m \in \F^t$. We denote by $V(P)$ the vector space spanned by its coefficients.
		
		For a multilinear polynomial $f$ we define $P_f$ the polynomial whose coefficients are the corresponding columns of the matrix $M_f$. In particular, as $f$ is multilinear, the coefficient $c_M\in\F^n$ of a monomial $M$ satisfies: $(c_M)_i$ is the coefficient of $M\cdot x_i$ in $f$. 
		
		The following remark can be verified by direct inspection:
		
		\begin{remark}
			\label{lem:shifting-vector}
			Let $h(\vx)=f(\vx+\vz)$. Then $P_h(\vx)=P_f(\vx+\vz)$.
		\end{remark}
		
		The next key claim follows from a work of Forbes-Ghosh-Saxena \cite{FGS18}.
		
		\begin{claim}
			\label{cl:good-basis-for-shift}
			Let $f$ be a multilinear polynomial and $B$ a subset of the variables.
			Suppose $V(P_f)$ has a basis $\mathcal{B}$ such that the monomials corresponding to all but at most one of the basis elements are supported only on variables of $B$.
			
			Then, for some $x_i \not \in B$, $V(P_{f(\vx+\vz')})$ (which is a vector space over $\F(\vz')$) has a  basis\footnote{The proof of \cite{FGS18} gives a stronger property -- that the basis $\mathcal{A}$ is \emph{cone closed}. I.e. that it is closed under taking submonomials. In other words, if the coefficient of a monomial $M$ is in the basis and $N$ divides $M$ then so is the coefficient of $N$. We do not need this stronger property.} $\mathcal{A}$ such that the monomials corresponding to \emph{all} the basis elements are supported only on variables of $B \cup \set{x_i}$.
		\end{claim}
		
		\begin{proof}
			Order the variables so that the variables in $B$ appear first, and then the variables outside of $B$.
			
			The claim follows from closely looking at  Algorithm 1 and Theorem 2 of \cite{FGS18}. We use $\mathcal{B}$ and $\mathcal{A}$, respectively, to denote the sets called $B$ and $A$ in their algorithm. The fact the $\mathcal{A}$ is a basis (and is in fact cone-closed) follows from Theorem 2 of \cite{FGS18}.
			The fact that $\mathcal{A}$ is supported on $B \cup \set{x_i}$ for some variable $x_i$ follows by inspecting their Algorithm 1: Indeed, note that, starting from $\mathcal{B}$, as long as the projection map $\pi$ (defined in their algorithm) is one-to-one (i.e., in the notation of that algorithm $\ell=1$) their algorithm ``erases'' the last variable and continues recursively. In the first time that $\ell=2$, since we started with a basis $\mathcal{B}$ such that the monomials corresponding to all but at most one of the basis elements are supported on only  variables of $B$, it must be the case that the current set held by the algorithm is supported on only variables of $B \cup \set{x_i}$. This fact doesn't change until the completion of the algorithm.
		\end{proof}

		\begin{claim}
			\label{cl:good-shift} Let $f$ be as in \autoref{cl:add-var-to-increase-rank}.
			There exists a non-zero $n$-variate polynomial $\Lambda_B(\vy)$, of degree at most $n^2$, such that if $\Lambda_B(\va) \neq 0$, then,   $g(\vx)=f(\vx+\va)$ satisfies
			$\rank(M_{g|_{B,0}}) = t$, and there exist $x_i, w \not \in B$ such that for $B'=B \cup \set{x_i,w}$, $\rank(M_{g|_{B',0}}) > t$.
		\end{claim}
				
		\begin{proof}
			Consider the symbolic shift $g(\vx)=f(\vx+\vz)$ where $\vz$ is a new set of variables and $g \in \F(\vz)[\vx]$. By \autoref{lem:rank-under-shifts}, $\rank(M_g) = \rank(M_f) = r$ (note that $M_g$ is defined over $\F(\vz)$).
			For simplicity of notation, suppose that $B=\set{x_1, \ldots, x_b}$ for some $1 \le b \le n$.
			Observe that
			\begin{align*}
				f|_{B,\vz} (\vx) &= f(x_1, \ldots, x_b, z_{b+1}, \ldots, z_n) \\
				& = f( (x_1 + z_1) - z_1 , \ldots, (x_b + z_b) - z_b, z_{b+1}, \ldots, z_n) \\
				& = g( x_1 -z_1, \ldots x_b - z_b, 0, \ldots, 0) = g|_{B,0} (\vx-\vz).
			\end{align*}
			Therefore,
			\[
			t = \rank(M_{f|_{B,\vz}}) = \rank(M_{g|_{B,0} (\vx-\vz)}).
			\]
			By \autoref{lem:rank-under-shifts}, it now follows that
			\[
			t = \rank(M_{g|_{B,0} (\vx)}).
			\]
			Suppose the partial derivatives corresponding to $x_1, \ldots, x_t$ are a row basis for $M_{g|_{B,0}}$. There exists a corresponding set of $t$ columns, i.e., monomials in the variables of $B$, such that the submatrix has rank $t$. Since $M_{g|_{B,0}}$ is a submatrix of $M_g$, these rows and columns are also linearly independent in $M_g$. Since $t<r$, we can pick a variable $w \not\in B$ such that $\partial g / \partial w$ is linearly independent of $\partial g / \partial x_1, \ldots, \partial g / \partial x_t$. We can similarly add another monomial (which would necessarily contain a variable not in $B$) to get a $(t+1) \times (t+1)$ submatrix of $M_g$ that has full rank. 
			
			Our next goal is to restrict only to columns corresponding to monomials in $B$, $w$ and perhaps another variable $x_i \not \in B$ while maintaining the linear independence of the $t+1$ rows.
			
			By \autoref{cl:good-basis-for-shift}, if we look at $g(\vx+\vz') = f(\vx+\vz+\vz')$, where $\vz'$ is yet again another symbolic shift, then
			the matrix $M_{f(\vx+\vz+\vz')}$, restricted to the rows of $B \cup \set{w}$, has such a basis. In particular there's some $(t+1) \times (t+1)$ minor of this matrix whose determinant is non-zero in $\F(\vz+\vz')$. As  in $M_{f(\vx+\vz+\vz')}$ every entry is a polynomial in $\vz+\vz'$, this implies that when doing the ``symbolic'' shift $f(\vx+\vz+\vz')$, and looking at the same minor, its determinant is a non-zero polynomial $\Lambda_B$ in $\vz+ \vz'$  
			of degree at most $(t+1)n$. By doing the change of variables $\vy=\vz+\vz'$ we can think of $\Lambda_B$ also as a polynomial in a new set of variables $\vy$.
						
			Thus, we have found a non-zero polynomial $\Lambda_B(\vy)$ such that whenever $\Lambda_B(\va) \neq 0$, it holds, for $g(\vx) = f(\vx+\va)$ and $B' = B \cup \set{x_i, w}$, that
			$\rank(M_{g|_{B',0}}) \ge t+1$.
		\end{proof}
		
		\begin{proof}[Proof of \autoref{cl:add-var-to-increase-rank}]
			Let $\va$ be such that $\Lambda_B(\va) \neq 0$, where $\Lambda_B$ is as in \autoref{cl:good-shift}. Let $g(\vx) = f(\vx+\va)$ and $B' = B \cup \set{w, x_i}$ where $x_i, w$ are as guaranteed by \autoref{cl:good-shift}.
			By \autoref{lem:rank-under-shifts} it follows that
			\[
			\rank(M_{g|_{B', 0}(\vx-\va)}) = \rank(M_{g|_{B', 0}}) \ge t+1.
			\]
			Since $g|_{B', 0}(\vx-\va) = f|_{B', \va} (\vx)$, \autoref{cl:add-var-to-increase-rank} follows.
		\end{proof}
		
		Having proved \autoref{cl:add-var-to-increase-rank}, we now move on to handle the linear factors.
		
		\begin{claim}
			\label{cl:no-linear-factors}
			Let $P$ be an irreducible multilinear polynomial of degree at least $2$ and  $B \subseteq [n]$. Then, there exists a polynomial $\Lambda'_{B,P}$ of degree at most $3n^3$ with the following property: for every $\va \in \F^n$, if $\Lambda'_{B,P}(\va) \neq 0$ then there are $y_1, y_2 \not \in B$ such that
			$P|_{B \cup \set{y_1, y_2}, \va}$ has no linear factors.
		\end{claim}
		
		\begin{proof}
			Without loss of generality, suppose $B=\set{x_1, \ldots, x_b}$.
			
			We first claim that there exists a polynomial $\Lambda''_{B,P}$, of degree at most $n^3$, and at most two variables $y_1, y_2 \notin B$ such that if $\Lambda''_{B,P}(\va) \neq 0$ then $P|_{B \cup \set{y_1, y_2} ,\va}$ is not a linear function.
			
			Indeed, for every two variables $x_i, x_j$, write $P = x_i x_j Q_{i,j} + R_{i,j}$, and let $\Lambda''_{B,P} = \prod_{\set{i,j} : Q_{i,j} \neq 0} Q_{i,j}$.
			
			By direct calculation, $\deg(\Lambda''_{B,P}) \leq \binom{n}{2} \cdot n \le n^3$.			
			Further, since  $\deg P \ge 2$, there exists at least one choice of $i,j$ such that $Q_{i,j} \neq 0$ and thus $\Lambda''_{B,P}$ is not the zero polynomial.
			Let $\va$ such that $\Lambda''_{B,P} (\va) \neq 0$ and let $x_i, x_j$ be two variables such that $Q_{i,j}(\va) \neq 0$. It follows that $P|_{B\cup \set{x_i, x_j},\va}$ contains monomials that are divisible by $x_i x_j$ with non-zero coefficients, and therefore it is not a linear function.
			
			In order to simplify notation, we rename $B \cup \set{x_i, x_j}$ to $B$ again for the rest of the proof.
			We now define the polynomial $\Lambda'_{B,P}$, which will be a product of $\Lambda''_{B,P}$ defined above by another non-zero polynomial.
			
			Suppose now that $P|_{B,\va}$ has a linear factor. Write $P|_{B,\va} = \ell \cdot P'$ where, without loss of generality, $\ell = \sum_{i=1}^t \beta_i x_i$ with $\beta_i \neq 0$ and $P'$ depends on $x_{t+1}, \ldots, x_b$.
			
			Suppose first that $\deg(P') \ge 1$.
			For every $i \in [t]$ we can write $P = x_i \cdot h_{i,1} + h_{i,0}$ with $h_{i,0} \neq 0$ (since $P$ is irreducible).
			
			It holds that $P|_{B,\va} =  x_i \cdot h_{i,1}|_{B,\va} + h_{i,0}|_{B,\va}$. Since $P|_{B,\va}$ is divisible by $\ell$, $\Res_{x_i} (P|_{B,\va}, \ell)  = 0$ (recall \autoref{sec:resultant}). Note that
			$\Res_{x_i} (P|_{B,\va}, \ell) = \beta_i \cdot  h_{i,0}|_{B,\va} - h_{i,1}|_{B,\va} \cdot (\ell- \beta_i x_i)$,  which implies that $h_{i,1}|_{B,\va}$ divides $h_{i,0}|_{B,\va}$.
			
			We claim that for all $i \in [t]$, $h_{i,1}$ depends on at least one variable in $B$. Indeed, write
			$P = \sum_{j=b+1}^n (x_j - \va_j) A_j + \ell P'$. The variable $x_i$ appears in $\ell$ and it is multiplied by $P'$, which is a nonconstant polynomial in the $B$ variables. All of its other occurrences are multiplied by $(x_j - \va_j)$ for some $j \ge b+1$, which implies that there's at least one coefficient with a variable $x_{j_i}\in B$ that appears in $h_{i,1}$.
			
			Since $P$ is irreducible, $\gcd(h_{i,1}, h_{i,0})=1$ and therefore $\Res_{x_{j_i}} (h_{i,1}, h_{i,0}) \neq 0$.
			
			Let $c_1(\vx)$ denote the coefficient of $x_{j_i}$ in $h_{i,1}$ and similarly $c_0(\vx)$ the coefficient of $x_{j_i}$ in $h_{i,0}$ (these are polynomials that may depend on all variables except $x_i$ and $x_{j_i}$ and  $c_1(\vx)\neq 0$).
			
			If $(c_1)|_{B,\va} \neq 0$ and $(c_0)|_{B,\va} \neq 0$, then it holds that
			$\Res_{x_{j_i}} (h_{i,1}|_{B,\va}, h_{i,0}|_{B,\va}) = \Res_{x_{j_i}} (h_{i,1}, h_{i,0})|_{B,\va}$
			(see, e.g., Proposition 6 in Chapter 3, Section 6 of \cite{CLO07}).

			Since $\Res_{x_{j_i}} (h_{i,1}, h_{i,0})$ is a polynomial of degree at most $2n$ and $c_0, c_1$ are polynomials of degree at most $n$, we obtain a polynomial $\Lambda'''_{B,i,j_i}=c_1\cdot c_0 \cdot \Res_{x_{j_i}} (h_{i,1}, h_{i,0})$ of degree at most $4n$
			such that if $\Lambda'''_{B,i,j_i}(\va)$ is non-zero then $\Res_{x_{j_i}} (h_{i,1}, h_{i,0})|_{B,\va}$ is non-zero, which implies that $\Res_{x_{j_i}} (h_{i,1}|_{B,\va}, h_{i,0}|_{B,\va})$ is non-zero, in contradiction to the assumption that $\ell$ divides $P|_{B,\va}$. Further, one should note that $\Lambda'''_{B,i,j_i}$ depends on $B$ and $i$ but does \emph{not} depend on the restriction $\va$.
			
			We define $\Lambda'''_{B,P}$ to be the product of all non-zero polynomials $\Lambda'''_{B,i,j}$.
			Thus, if $\Lambda'''_{B,P}(\va) \neq 0$, it can't be the case that $P|_{B,\va} = \ell \cdot P'$ for $\deg(P') \ge 1$.
			
			Finally, we let $\Lambda'_{B,P} = \Lambda'_{B,P} \cdot \Lambda''_{B,P}$. Let $\va$ be such that $\Lambda'_{B,P}(\va) \neq 0$. Then by the arguments above, $P_{B,\va}$ cannot be a linear
			function nor a product of a linear function by a polynomial with positive degree. The bound on the degree of $\Lambda'$ follows by direct inspection.
		\end{proof}

		\begin{corollary}
			\label{cor:restriction-no-linear-factor}
			Let $P$ be an irreducible multilinear polynomial with no linear factors and $B$  a subset of the variables. Suppose there exists $\va \in \F^n$ such that $P|_{B,\va}$ is non-constant and has no linear factors. Then, there exists a non-zero polynomial $\Lambda'''_B$ of degree at most $2n^3$ such that $\Lambda'''_B (\va) \neq 0$ and for all $\vb \in \F^n$, if $\Lambda'''_B (\vb) \neq 0$ then $P|_{B,\vb}$ has no linear factors.
		\end{corollary}

		\begin{proof}
			Observe that in the proof of \autoref{cl:no-linear-factors}, when  $\deg(P') \ge 1$, the polynomial $\Lambda'''_{B,P}$ has the claimed property.
		\end{proof}
		
		We are now ready to prove \autoref{cla:add-variables-to-B}.
		
		\begin{proof}[Proof of \autoref{cla:add-variables-to-B}]
			Write $f = \prod_{i=1}^{a} \ell_i \cdot \prod_{i=1}^b P_i$ where the $\ell_i$'s are linear functions and the $P_i$'s are irreducible non-linear polynomials. Denote $P = \prod_{i=1}^b P_i$ to be the non-linear part of $f$. By definition, it holds that $\ranksem(f) = \rank(M_P)$.
			
			Similarly, write $f|_{B,\va} = \prod_{i=1}^a (\ell_i)|_{B,\va} \prod_{i=1}^b (P_i)|_{B,\va}$.
			
			For a set $\tilde{B}$, set $\tilde{\Psi}_{\tilde{B}} = \Lambda_{\tilde{B}} \cdot \Lambda'_{\tilde{B}}$ where $\Lambda_{\tilde{B}}$ is as in \autoref{cl:add-var-to-increase-rank} and $\Lambda'_{\tilde{B}}=\prod_{i=1}^{b}\Lambda'_{\tilde{B},P_i}$, where $\Lambda'_{\tilde{B},P_i}$ is as defined in \autoref{cl:no-linear-factors}.
			$\tilde{\Psi}_{\tilde{B}}$ is a polynomial of degree at most $6n^5$.
			
			Now set $\Psi_B = \prod_{ \tilde{B} \supseteq B, |\tilde{B}| \le |B|+2} \tilde{\Psi}_{\tilde{B}}$. As there are at most $n^2$ ways to extend the set $B$ by at most $2$ elements, the degree of $\Psi_B$ is at most $6n^7$.
			
			
			Pick $\va$ such that $\Psi_B(\va) \neq 0$.
			
			We split the proof into two cases.
			\begin{enumerate}
				\item No $(P_i)|_{B,\va}$ has linear factors.
				In this case, $\ranksem(f|_{B,\va}) = \rank(M_{P|_{B,\va}})$. By \autoref{cl:add-var-to-increase-rank}, there exist two variables $x,w$ not in $B$ such that for $B'=B \cup \set{x,w}$,
				$\rank(M_{P|_{B',\va}}) >  \rank(M_{P|_{B,\va}})$.
				
				We split this case further into two subcases.
				
				\begin{enumerate}
					\item No $(P_i)|_{B',\va}$ has linear factors. In this case, $\ranksem(f|_{B',\va}) = \rank(M_{P|_{B',\va}})$ and the proof is completed.
					\item There exists $i$ such that $(P_i)|_{B',\va}$ has linear factors. This case is handled similarly to case 2 below by setting $B=B'$, observing that by definition of $\Psi_B$, it holds that $\tilde{\Psi}_{B'}(\va)$ is non-zero, and therefore $\Lambda_{B'}$ and $\Lambda'_{B'}$ are non-zero.
				\end{enumerate}

				\item There exists $i \in [b]$ such that $(P_i)|_{B,\va}$ has linear factors. Write
				\[
				f|_{B,\va} = \prod_{i=1}^a (\ell_i)_{B,\va} \cdot \prod_{i=1}^{b} \left( \prod_{j=1}^{b_i} (\tilde{\ell}_{i,j}) \right) \cdot {P_{i,B,\va}},
				\]
				where the $\tilde{\ell}_{i,j}$'s denote the linear factors added in the restriction, and ${P_{i,B,\va}}$ has no linear factors. Let ${P_{B,\va}} = \prod_{i=1}^b {P_{i,B,\va}}$ so that
				$\ranksem(f|_{B,\va}) = \rank(M_{{P_{B,\va}}})$.
				
				By the assumption there exists $i$ such that $b_i\geq 1$. Without loss of generality suppose $i=1$. I.e., $P_1$ is an irreducible polynomial of degree at least $2$ that under the restriction has linear factors. By \autoref{cl:no-linear-factors}, we can add two new variables to $B$ to obtain $B'$ such that $(P_1)|_{B',\va}$ has no linear factors. We now claim that $\ranksem(f|_{B',\va}) > \ranksem(f|_{B,\va})$.
				
				We start by similarly splitting $f|_{B',\va}$ into a product of linear functions times a polynomial ${P_{B',\va}}$ such that $\ranksem(f|_{B',\va})$ is the number of linearly independent linear functions that ${P_{B',\va}}$ depends on. First, we note that
				\[
				\ranksem(f|_{B',\va}) = \ranksem({P_{B',\va}}) \ge \ranksem({P_{B',\va}}|_{B,\va}),
				\]
				so it suffices to prove that $\ranksem({P_{B',\va}}|_{B,\va}) > \ranksem({P_{B,\va}})$. We first claim that ${P_{B,\va}}$ divides ${P_{B',\va}}|_{B,\va}$: consider the process of obtaining $f|_{B,\va}$ from $f|_{B',\va}$ by restricting the variables in $B' \setminus B$: some irreducible factors of $f|_{B',\va}$ become linear; Others remain non-linear and appear in the product defining ${P_{B,\va}}$, and those must come from restrictions to $B$ of non-linear irreducible factors that appear in ${P_{B',\va}}$. In particular, their restrictions to $B$ appear in ${P_{B',\va}}|_{B,\va}$, which implies that ${P_{B,\va}}$ divides ${P_{B',\va}}|_{B,\va}$.
				Further, we claim that $\tilde{\ell}_{1,1}$ divides ${P_{B',\va}}|_{B,\va}$. Indeed, $(P_1)|_{B',\va}$ has no linear factors. Thus by definition of ${P_{B',\va}}$,
				\[
				(P_1)|_{B',\va} \mid {P_{B',\va}}.
				\]
				Restricting both sides further to $B$, we get on the left hand side a product involving $\tilde{\ell}_{1,1}$ and on the right hand side ${P_{B',\va}}|_{B,\va}$.
				
				Thus, ${P_{B',\va}}|_{B,\va}$ depends on at least $\ranksem({P_{B,\va}})+1$ linear functions: the $\ranksem({P_{B,\va}})$ linear functions that ${P_{B,\va}}$ depend on, plus $\tilde{\ell}_{1,1},$ which is variable disjoint and hence linearly independent of all of them. \qedhere
			\end{enumerate}
		\end{proof}
		
		\subsection{Why Semantic Rank}\label{sec:whysemantic}
		In this section we add some more technical details to the explanation in \autoref{sec:errors}. Recall that \autoref{thm:unique-syn-part} proved that there are certain parameters that guarantee uniqueness of clusters in a syntactic partition. This result is enough to fix the relevant theorems in \cite{KarninShpilka09,BSV21}. As explained, one main contribution of this work is an algorithmic construction of a set $B$ such that restricting (a random shift of) the black-box polynomial $f$ to the variables in $B$, maintains the cluster structure of any $\Sps{k}$ circuit for $f$. I.e., the new clusters are (up to a permutation) equal to the restrictions of the clusters of $f$ to $B$. 
		
		We constructed $B$ gradually (\autoref{algo:finding-B}). At each step we checked whether adding a small set of variables increases either the number of clusters of $f|_B$ or the rank of at least one cluster. To prove that this process finds $B$ that preserves the cluster structure, we had to prove that if this was not the case then we could have found variables to add to $B$ (\autoref{cla:add-variables-to-B}). We then needed to prove that if we guessed the right set of variables to add then we can \emph{verify} that they indeed contributed to the rank (Step~\ref{line:rank-increase} of \autoref{algo:finding-B}). All of this works when using \emph{semantic} rank. However, when using \emph{syntactic} rank this is not quite the case. The problem is that when we learn $f_I$ (Step~\ref{step:cluster} in \autoref{algo:finding-B}) then it may be the case that $C'_j|_{x_I=\va_I} = C_{j}$, yet $\ranksyn(C'_j)<\ranksyn(C_j)$, something that cannot happen if we work with semantic rank. Indeed, the simple example (where $\omega$ is a primitive root of unity of order $d$)
		\[
		\sum_{i=1}^{k}\prod_{j=1}^{d} (x_i-\omega^j x_{i+1}) = x_1^d-x_{k+1}^{d} = \prod_{j=1}^{d} (x_1-\omega^j   x_{k+1})
		\]   
		shows that the left hand side has syntactic rank $k+1$ while the middle expression has syntactic rank  $2$. Both sides have semantic rank $0$ (although recall that by \autoref{rem:semantic-rank-0}, it's sometimes convenient to define this rank as 1). We can also cook up such examples where the number of multiplication gates is the same in both circuits, and where the circuits are multilinear and have semantic rank at least $1$. Thus, we were able to prove the correctness of \autoref{algo:finding-B} relying on semantic rank and not on syntactic rank.

		\section{Reconstruction Algorithm for Multilinear $\Sps{k}$ Circuits}
		\label{sec:highdeg}
		
		In this section we provide our algorithm for learning multilinear $\Sps{k}$ circuits.
		Our proof is similar in structure to the proof of Bhargava, Saraf and Volkovich \cite{BSV21}. However, since many of our definitions are different we are required to prove several basic claims.		
		We start by explaining how the results of the previous sections imply that we can get black box access to the clusters on arbitrary points. In the previous section, we picked a random $\va$ in \autoref{algo:finding-B} and argued about the clusters of $f|_{B,\va}$. We'd like to obtain similar claims about the clusters of $f|_{B,\vb}$, assuming $\vb$ doesn't satisfy certain degeneracy conditions.
		
		\begin{lemma}
			\label{lem:generic-projection}
			Let $P$ be a multilinear polynomial and $B$ a subset of the variables. Suppose that $\ranksem(P|_{B,\va}) \ge t$ and that $P(\va)\cdot \Lambda'''_B(\va) \neq 0$ (where $\Lambda'''_B$ is as defined in \autoref{cor:restriction-no-linear-factor}). Then, there exists a non-zero polynomial $\Theta_{B}$ of degree at most $2n^5$ such that $\Theta_B (\va) \neq 0$ and, if $\Theta_B (\vb) \neq 0$, then $\ranksem(P|_{B,\vb}) \ge t$.   
		\end{lemma}

		\begin{proof}
			Suppose without loss of generality that $P$ is irreducible (as otherwise argue separately on each irreducible factor). We may also assume that $P|_{B,\va}$ is non-constant as otherwise the statement is trivial.
			If $P|_{B,\va}$ has no linear factors, then
			\[
			\ranksem(P|_{B,\va}) = \rank(M_{P|_{B,\va}}).
			\]
			Since by assumption $\rank(M_{P|_{B,\va}}) \ge t$, there's a $t \times t$ minor whose determinant  is non-zero. This implies that, in the matrix $M_{P|_{B,\vz}}$, the same $t \times t$ minor has non-zero determinant, as a polynomial in $\vz$. Denote this determinant with $\text{Det}(\vz)$, and observe that $\text{Det}(\va) \neq 0$. Furthermore, as each coefficient in $P|_{B,\vz}$ has degree at most $n$ as  a polynomial in $\vz$, and $t\leq n$ we have that $\deg(\text{Det})\leq n^2$.
			
			Define  $\Theta_B := \Lambda'''_B \cdot \text{Det}$ where $\Lambda'''_B$ is as in \autoref{cor:restriction-no-linear-factor}. As $\deg(\Lambda'''_B)\leq 2n^3$ we get that $\deg(\Theta_B)\leq 2n^5$ and that $\Theta_B(\va) \neq 0$.
			
			Let $\vb$ be such that $\Theta_B(\vb)\neq 0$. \autoref{cor:restriction-no-linear-factor} guarantees that $P|_{B,\vb}$ has no linear factors, and since $\text{Det}(\vb) \neq 0$, $\rank(M_{P|_{B,\vb}}) \ge t$. This implies that $\ranksem(P|_{B,\vb}) \ge t$.
			
			Now suppose $P|_{B,\va}$ has linear factors and can be written as $L \cdot P'$ where $L$ is a product of linear factors and $P'$ has no linear factors, so that $\ranksem(P|_{B,\va})  = \rank(M_{P'})$. Let $B'$ be the set of variables that appears in $P'$ (and is disjoint from the set of variables that appear in $L$). Note that since $P(\va) \neq 0$ it must be the case that $P|_{B',\va}$ is non-zero and hence $L|_{B',\va}$ is a non-zero constant, i.e., $P' = c P|_{B',\va}$ for some $c \neq 0$.
			
			We now apply the same reasoning as before to $P'$, which has no linear factors, to deduce that for some non-zero polynomial $\Theta_B$, if $\Theta_B(\vb) \neq 0$ then $\ranksem(P|_{B',\vb}) \ge \ranksem(P|_{B',\va}) = \ranksem(P') = t$. Finally, note that as $B' \subseteq B$, $\ranksem(P|_{B,\vb}) \ge \ranksem(P|_{B',\vb}) \ge t$.
		\end{proof}

		We say that an output $(B,\va)$ of \autoref{algo:finding-B} is \emph{good} if it satisfies $\Gamma_C(\va) \neq 0$, where $\Gamma_C$ is the polynomial defined in \autoref{cl:tau-semantic-restriction}.

		\begin{claim}
			\label{cl:beta-clusters}
			Let $f \in \Sps{k}$ be a multilinear polynomial and let $C$ be a minimal $\Sps{k}$ computing $f$.	Let $(B,\va)$ be a good output of \autoref{algo:finding-B} on $f$.
			
			Consider a partition of $f$, $f = \sum_{i=1}^s f_i$, as given by \autoref{cl:partition-two-taus} with $\varphi(k) = k^2$ and $\tau_{\text{min}}$ as in \autoref{cl:tau-semantic-restriction}. Let $\tau_0, \tau_1, r$ be its parameters as promised by the claim. Denote $\tau=\tau_0^{k}$.
			
			Then, there exists a polynomial $\Theta_{B,C}$ of degree at most $2n^7$ such that $\Theta_{B,C} (\va) \neq 0$, and the following property holds:
			for every $\vb \in \F^n$ such that $\Theta_{B,C}(\vb) \neq 0$ and circuit $D$ computing $f|_{B,\vb}$, it holds that the output of \autoref{algo:semantic-clustering-algorithm}, when given $D$ and $\tau$ as input, which we denote $[D] = \sum_{i=1}^{s'} g_i$, satisfies:
			\begin{enumerate}
				\item $s'=s$,
				\item $g_i = (f_i)|_{B,\vb}$, up to reordering of the indices,
				\item $\ranksem(g_i) = \ranksem(f_i)$.
			\end{enumerate}
			In particular, the $g_i$'s also form a $(\tau,r)$ partition.
		\end{claim}
				
		\begin{proof}
			By \autoref{cl:tau-semantic-restriction}, we know that there exists a choice of $\vb$ (that is, $\vb=\va$) that satisfies the required properties. We first show that there is a non-zero polynomial $\Theta_{B,C}$ of degree at most $2n^7$ such that if $\Theta_{B,C}(\vb) \neq 0$ then
			\begin{enumerate}
				\item for every $i \in [s]$, $\ranksem(f_i|_{B,\vb}) = \ranksem(f_i)$,
				\item for every $i \neq j \in [s]$, $\ranksem(f_i|_{B,\vb}, f_j|_{B,\vb}) \ge \tau r$.
			\end{enumerate}
			For the first item, denote by $\Theta_i$ the polynomial promised by \autoref{lem:generic-projection} applied to $f_i$ and $B$. Since $(B,\va)$ are good, $\Gamma_C(\va) \neq 0$ (where $\Gamma_C$ is as defined in the proof of \autoref{cl:tau-semantic-restriction}) and one can verify that assumptions on $\va$ in the statement of \autoref{lem:generic-projection} indeed hold. In particular, $\Theta_i (\va) \neq 0$. Thus, if $\Theta_i (\vb) \neq 0$ then
			\[
			\ranksem((f_i)|_{B,\vb}) \ge	\ranksem((f_i)|_{B,\va}) = \ranksem (f_i),
			\]
			and the reverse inequality is clear.
			
			For the second item, we similarly let $\Theta_{i,j}$ be the polynomial promised by \autoref{lem:generic-projection} applied to $f_i + f_j$ and $B$. The assumptions on $\va$ in the statement of \autoref{lem:generic-projection} again hold because $(B,\va)$ is good.
			
			Finally, set $\Theta_{B,C} = (\prod_{i=1}^s \Theta_i) \cdot (\prod_{i \neq j \in [s]} \Theta_{i,j})$. Note that $\Theta_{B,C}(\va) \neq 0$ as each factor is non-zero. It is also clear that $\deg(\Theta_{B,C})\leq 2n^5(n+{n\choose 2})<2n^7$.
			
			Let $\vb$ be such that $\Theta_{B,C} (\vb) \neq 0$. Consider \autoref{algo:semantic-clustering-algorithm} run on $f|_{B,\vb}$ with parameter $\tau$. Denote its output by $f|_{B,\vb} = \sum_{i=1}^{s'} g_i$.
			By \autoref{cl:cluster-refinement-in-B}, we have that $s' \le s$. On the other hand, as  $\Theta_{B,C} (\vb) \neq 0$, \autoref{lem:generic-projection} implies that the set $\set{(f_i)|_{B,\vb}}$ for $i \in [s]$ is a $\tau$ partition with $s$ clusters. Hence, by \autoref{cor:max-clus-min-rank}, and the fact that \autoref{algo:semantic-clustering-algorithm} returns the partition of minimal rank (and in particular with maximal number of clusters), the output will be $\set{(f_i)|_{B,\vb}}$, as we wanted to show.			
		\end{proof}

		\begin{claim}
			\label{cl:output-B-with-random-beta}
			Let $f \in \Sps{k}$ be a multilinear polynomial and let $C$ be a $\Sps{k}$ circuit computing $f$.
			Let $(B,\va)$ be good outputs of \autoref{algo:finding-B} on $f$. Let $B' \supseteq B$.
			
			For every $\vb \in \F^n$ such that $\Theta_{B,C}(\vb) \neq 0$ the following holds:
			Denote by $f|_{B,\vb} = \sum_{i=1}^s (f_i)|_{B,\vb}$ the output of \autoref{algo:semantic-clustering-algorithm} on $f|_{B,\vb}$ with parameter $\tau$ (as promised by \autoref{cl:beta-clusters}). Let $f|_{B',\vb} = \sum_{i=1}^{s'} g_i$ be the output of \autoref{algo:semantic-clustering-algorithm} on $f|_{B',\vb}$ with parameter $\tau$.
			Then $s=s'$ and up to permutation of the indices, $g_i = (f_i)|_{B',\vb}$.
		\end{claim}
		
		\begin{proof}
			By \autoref{cl:beta-clusters}, it holds that:
			\begin{enumerate}
				\item for every $i \in [s]$, $\ranksem((f_i)|_{B,\vb}) = \ranksem(f_i)$.
				\item for every $i \neq j \in [s]$, $\ranksem((f_i)|_{B,\vb}, (f_j)|_{B,\vb}) \ge \tau r$.
			\end{enumerate}
			
			Thus, as $B' \supseteq B$,
			\begin{enumerate}
				\item for every $i \in [s]$, $\ranksem(f_i) =  \ranksem((f_i)|_{B,\vb}) \le \ranksem((f_i)|_{B',\vb}) \le \ranksem(f_i)$.
				\item for every $i \neq j \in [s]$, $\ranksem((f_i)|_{B',\vb}, (f_j)|_{B',\vb}) \ge \tau r$.
			\end{enumerate}
			Thus, $(f_i)|_{B',\vb}$ is a $\tau$ partition with $s$ clusters.  \autoref{cl:cluster-refinement-in-B} and \autoref{cor:max-clus-min-rank} show that $s'=s$ and up to permutation of the indices, $g_i = (f_i)|_{B',\vb}$.
		\end{proof}

		\subsection{Cluster Evaluation}
		\label{sec:cluster-evaluation}
		
		In this section we explain how to evaluate the clusters at arbitrary points. Recall that what we have is access to the clusters $f_i|_{B,\va}$ so we'd like to replace $\va$ by an arbitrary point $\vb \in \F^n$.
		We do it in several stages as in \cite{BSV21}. 
		Basically, we replace all uses of their Lemma 6.14 by our \autoref{cl:beta-clusters}. We briefly explain how this is done.
				
		The first step shows how to evaluate the clusters of $f|_{B,\vb'}$ if the Hamming distance of $\vb$ and $\vb'$ is $1$.
		
		\begin{lemma}[Similar to Lemma 6.17 in \cite{BSV21}]
			\label{lem:hamming-distance-one}
			Let $f \in \Sps{k}$ be a  multilinear polynomial and let $C$ be a minimal multilinear $\Sps{k}$ circuit computing $f$. 
			Let $(B,\va)$ be a good output  of \autoref{algo:finding-B} on $f$. Let $\vb$  such that $\Theta_{B,C}(\vb) \neq 0$, where $\Theta_{B,C}$ is as given in \autoref{cl:beta-clusters}, and  $f|_{B,\vb} = \sum_{i=1}^s f_i|_{B,\vb}$  the output of \autoref{algo:semantic-clustering-algorithm} with parameter $\tau$ on $f|_{B,\vb}$.
			
			Let $\vb' \in \F^n$ be of Hamming distance 1 from $\vb$. Then, there exists an algorithm that runs in time $2^{k^2} \cdot \poly(n)$ and outputs $(f_1|_{B,\vb'}, \ldots, f_s|_{B,\vb'})$.
		\end{lemma}
		
		\begin{proof}
			Let $j$ be the coordinate on which $\vb'$ differs from $\vb$. 
			We may assume $j \not \in B$ as otherwise the statement is immediate.
			Let $B' = B \cup \set{x_j}$ and run \autoref{algo:semantic-clustering-algorithm} on $f|_{B',\vb}$, and denote its output by $\sum_{i=1}^s g_i$.
			By \autoref{cl:output-B-with-random-beta}, $g_i= f_i|_{B',\vb}$ up to permutation of the indices. It is easy to find the permutation by running a randomized PIT algorithm  between $g_i|_{B,\vb}$ and $f_{i'}|_{B,\vb}$ for all $i,i' \in [s]$.
			Thus by fixing the $j$-th coordinate appropriately we get $f_i|_{B,\vb'}$.
		\end{proof}
		
		Note that in the above lemma, as we move from $\vb$ to $\vb'$ we can match the clusters of $f|_{B,\vb}$ to the corresponding clusters of $f|_{B,\vb'}$. In what follows, we'll implicitly use the fact that we can find this matching.
		
		The next step shows how to evaluate the clusters of $f|_{B,\vb}$ for $\vb$ that is arbitrarily far from $\va$ but satisfies a technical condition. For $0 \le i \le n$, let $\gamma_i(\va, \vb)$ denote the $i$-th ``hybrid'' between $\va$ and $\vb$, i.e., a vector whose first $(n-i)$ coordinates are the first $(n-i)$ coordinates of $\va$ and whose last $i$ coordinates are the last $i$ coordinates of $\vb$, so that $\gamma_0(\va,\vb)=\va$, $\gamma_n (\va, \vb) = \vb$ and the Hamming distance between $\gamma_i(\va,\vb)$ and $\gamma_{i+1}(\va,\vb)$ is 1.
		
		\begin{lemma}[Similar to Corollary 6.18 in \cite{BSV21}]
			\label{lem:clusters-line}
			Let $f \in \Sps{k}$ be a multilinear polynomial and let $C$ be a minimal multilinear $\Sps{k}$ circuit computing $f$. 
			Let $(B,\va)$ be a good output of \autoref{algo:finding-B} on $f$. 
			Let $\vb \in \F^n$  such that $\Theta_{B,C}(\gamma_i(\va,\vb)) \neq 0$ for all $0 \le i \le n-1$, and let $f|_{B,\va} = \sum_{i=1}^s f_i|_{B,\va}$ be the output of \autoref{algo:semantic-clustering-algorithm} on $f|_{B,\va}$.

			Then, there exists an algorithm that runs in time $2^{k^2} \cdot \poly(n)$ and outputs $(f_1|_{B,\vb}, \ldots, f_s|_{B,\vb})$.
		\end{lemma}
		
		\begin{proof}
			We apply \autoref{lem:hamming-distance-one} repeatedly on $\va=\gamma_0 (\va, \vb), \gamma_1(\va,\vb), \ldots, \gamma_n(\va,\vb) = \vb$.
		\end{proof}
		
		Finally, we show how to evaluate the clusters on \emph{arbitrary} $\vb \in \F^n$. To do this, we consider, as in \cite{BSV21}, the line $\ell_{\va,\vb} (t)$ through $\va$ and $\vb$ and show that most points on the line are non-zeros of the polynomial $\Theta_{B,C}$. For each such ``good'' point we can recover the clusters, and then  apply the Berlekamp-Welch algorithm in order to recover the clusters of $f|_{B,\ell_{\va,\vb}(t)}$, for every $t$.

		\begin{lemma}[Similar to Lemma 6.19 in \cite{BSV21}]
			\label{lem:clusters-arbitrary}
			Let $f \in \Sps{k}$ be a multilinear polynomial and let $C$ be a $\Sps{k}$ circuit computing $f$. 
			Let $(B,\va)$ be good outputs of \autoref{algo:finding-B} on $f$. Let $f|_{B,\va} = \sum_{i=1}^s f_i|_{B,\va}$ be the output of \autoref{algo:semantic-clustering-algorithm} on $f|_{B,\va}$.
			
			Then, there exists an algorithm that, given any $\vb \in \F^n$, runs in time $2^{k^2} \cdot \poly(n)$ and outputs $(f_1|_{B,\vb}, \ldots, f_s|_{B,\vb})$.
		\end{lemma}
		
		\begin{proof}
			Let $\ell_{\va,\vb} (t)$ be the line through $\va$ and $\vb$ so that $\ell_{\va,\vb}(0)=\va$ and $\ell_{\va,\vb}(1)=\vb$. Let $W \subseteq \F$ be a set of size $10n^9$.
			
			For each $u \in W$, let $\vb_u=\ell_{\va,\vb}(u)$. Use the algorithm in \autoref{lem:clusters-line} to learn  $(f_1|_{B,\vb_u}, \ldots, f_s|_{B,\vb_u})$.
			Note that $f_i(\vb_u) = f_i|_{B,\vb_u} (\vb_u) = f_i (\ell_{\va,\vb}(u))$.
			
			For each $i \in [s]$ we apply the Berlekamp-Welch algorithm on the points $f_i (\ell_{\va,\vb}(u))$, for $u \in W$, to recover the univariate polynomial $f_i (\ell_{\va,\vb}(t))$, and output the value $f_i (\ell_{\va,\vb}(1))$, which equals $f_i|_{B,\vb}(\vb) = f_i(\vb)$.
			
			To prove that this indeed works, we need to show that there are many values of $u$ for which the conditions of \autoref{lem:clusters-line} hold and thus the computation is correct, so in particular the Berlekamp-Welch algorithm returns the correct polynomial.
			
			Let $Q(t) = \prod_{i=1}^{n-1} \Theta_{B,C} (\gamma_i (\va, \ell_{\va,\vb} (t)))$.
			$Q$ is a non-zero polynomial since $Q(0) = \Theta_{B,C}(\va)^{n-1} \neq 0$. Further, $Q$ has degree at most $n \cdot \deg(\Theta_{B,C}) \le 2n^8$ and if $Q(u) \neq 0$, then $\vb_u=\ell_{\va,\vb}(u)$ satisfies the condition of \autoref{lem:clusters-line}.
			The number of roots of $Q$ is thus at most $2n^8$, which bounds the number of errors for the Berlekamp-Welch algorithm.			
			Thus, since the number of evaluations $|W| > 4n^8 + \deg(f_i) + 1$, we can indeed recover the polynomial $f_i (\ell_{\va,\vb}(t))$ correctly.
		\end{proof}

		\subsection{The Reconstruction Algorithm}
		
		We now give our reconstruction algorithm for multilinear $\Sps{k}$ circuits. 
		
		\begin{algorithm}[H]
			\caption{: Reconstruction of $\Sps{k}$ circuits}
			\label{alg:reconstruct-sumprodsumk}
			\begin{algorithmic}[1]
				\Require{Black box access to a degree $d$, $n$-variate multilinear polynomial $f$ computed by a minimal multilinear $\Sps{k}$ circuit, $C$.}	
				\Ensure{A $\Sps{k}$ circuit $\tilde{C}$ such that, with high probability, $\tilde{C}$ computes $f$.}	
				\For{every $R_M(2k) \le \tau_0 \le R_M(2k)^{k^{2k}}$} \label{line:find-tau}
				\State{Run \autoref{algo:finding-B} with parameter $\tau=\tau_0^{k}$ on $f$  to obtain a cluster preserving pair $(B,\va)$ where $\card{B}\leq k^{k^{O(k)}}$}
				\State{Denote by $(f_1|_{B,\va}, \ldots, f_s|_{B,\va})$ the output of \autoref{algo:semantic-clustering-algorithm} on $f|_{B,\va}$ with parameter $\tau$}
				\For{$i\in [s]$}
				\State{Use \autoref{lem:low semantic rank exact}, along with \autoref{lem:clusters-arbitrary} to simulate black-box access to $f_i|_{B,\va}$, to learn a circuit $\tilde{C}_i$.}
				\EndFor
				\State{Let $\tilde{C} =  \tilde{C}_1 +\ldots+ \tilde{C}_{s}$.}
				\If{$[\tilde{C}] \equiv [C]$ (which can be checked using a randomized PIT algorithm)} \label{line:PIT}
				\State{Output $\tilde{C}$.}
				\EndIf
				\EndFor
			\end{algorithmic}
		\end{algorithm}
		
		\begin{theorem}
			\label{thm:learning-multilinear-spsk-ckts}
			Suppose $|\F| > n^{k^{k^{O(k)}}}$.
			Let $f$ be a polynomial computed by a multilinear $\Sps{k}$ circuit. Then, with high probability, 
			\autoref{alg:reconstruct-sumprodsumk} returns a multilinear $\Sps{k}$ circuit $\tilde{C}$ computing $f$ in time $\poly(n) \cdot k^{k^{k^{k^{\poly(k)}}}}$.
		\end{theorem}
		
		\begin{proof}
			Let $f$ be as given in the theorem. Consider a partition of $f$, $f = \sum_{i=1}^s f_i$, as given by \autoref{cl:partition-two-taus} (for $\varphi(k)=k^2$) with parameters $\tau_0, \tau_1$. As the algorithm scans all possible values in Line \ref{line:find-tau}, it will find the ``correct'' value of $\tau_0$. 
			
			Focusing on this $\tau_0$, the assumption on the field size and \autoref{cl:tau-semantic-restriction} guarantee  that with high probability  $(B,\va)$ will be a good output  of \autoref{algo:finding-B}. \autoref{cla:size-of-B} promises that \autoref{algo:finding-B} indeed outputs $B$ of the required size such that the output of \autoref{algo:semantic-clustering-algorithm} on $f|_{B,\va}$ are the clusters  $(f_1|_{B,\va}, \ldots, f_s|_{B,\va})$.
			
			\autoref{lem:clusters-arbitrary} now guarantees that we can evaluate each $f_i$ correctly on each $\vb \in \F^n$ so that we can use the algorithm from  \autoref{lem:low semantic rank exact} to learn circuits computing each $f_i$. Thus, we are able to reconstruct a circuit computing $f$ and pass the check in Line \ref{line:PIT}.
			
			The bound on the running time follows from the bound on the running time of \autoref{algo:finding-B} given in \autoref{cla:size-of-B} and from the running time given in \autoref{lem:low semantic rank exact} (see also \autoref{rem:running-time-of-finding-B} for a remark on the running time of \autoref{algo:finding-B}).
		\end{proof}

		\subsection{Proper Learning of Depth-$3$ Set-Multilinear Circuits}		
		We now explain the changes required in \autoref{alg:reconstruct-sumprodsumk} in order to prove \autoref{thm:intro:set-multilinear}.
		The main change is to replace each application of \autoref{lem:low semantic rank exact} with \autoref{lem:low rank set-ml exact}. In particular, this makes sure that in Step~\ref{step:learn-C}  of \autoref{algo:finding-B} we always find subcircuits of $C$ that are \emph{set-multilinear}, and not merely multilinear (note that by fixing the variables in the original circuit, we know that the restricted polynomial has a small set-multilinear circuit, and thus the algorithm can find it). 
		
		\section*{Acknowledgement}
		We thank the anonymous reviewers for various comments that greatly improved the presentation of the paper.
		
		\bibliographystyle{customurlbst/alphaurlpp}
		\bibliography{references}

\begin{thebibliography}{KMSV13}

\bibitem[AGKS15]{AGKS15}
Manindra Agrawal, Rohit Gurjar, Arpita Korwar, and Nitin Saxena.
\newblock \href {http://dx.doi.org/10.1137/140975103} {Hitting-Sets for {ROABP}
  and Sum of Set-Multilinear Circuits}.
\newblock {\em SIAM Journal of Computing}, 44(3):669--697, 2015.
\newblock Pre-print available at \href {http://arxiv.org/abs/1406.7535}
  {\path{arXiv:1406.7535}}.

\bibitem[AV08]{AV08}
Manindra Agrawal and V.~Vinay.
\newblock \href {http://dx.doi.org/10.1109/FOCS.2008.32} {Arithmetic Circuits:
  A Chasm at Depth Four}.
\newblock In {\em \FOCS{2008}}, pages 67--75, 2008.
\newblock Pre-print available at \parseECCC{TR08/062}.

\bibitem[BIJL18]{BIJL18}
Markus Bl{\"{a}}ser, Christian Ikenmeyer, Gorav Jindal, and Vladimir Lysikov.
\newblock \href {http://dx.doi.org/10.1145/3188745.3188832} {Generalized matrix
  completion and algebraic natural proofs}.
\newblock In {\em \STOC{2018}}, pages 1193--1206. {ACM}, 2018.

\bibitem[BS24]{BS24}
Vishwas Bhargava and Devansh Shringi.
\newblock \href {https://eccc.weizmann.ac.il/report/2024/123} {Faster {\&}
  Deterministic {FPT} Algorithm for Worst-Case Tensor Decomposition}.
\newblock {\em Electron. Colloquium Comput. Complex.}, {TR24-123}, 2024.
\newblock Pre-print available at \href {http://arxiv.org/abs/TR24-123}
  {\path{arXiv:TR24-123}}.

\bibitem[Bsh13]{Bshouty13}
Nader~H. Bshouty.
\newblock \href {http://dx.doi.org/10.1007/978-3-642-40935-6\_4} {Exact
  Learning from Membership Queries: Some Techniques, Results and New
  Directions}.
\newblock In {\em Algorithmic Learning Theory - 24th International Conference,
  {ALT} 2013}, volume 8139 of {\em Lecture Notes in Computer Science}, pages
  33--52. Springer, 2013.

\bibitem[BSV20]{BSV20}
Vishwas Bhargava, Shubhangi Saraf, and Ilya Volkovich.
\newblock \href {http://dx.doi.org/10.1137/1.9781611975994.132} {Reconstruction
  of Depth-4 Multilinear Circuits}.
\newblock In {\em \SODA{2020}}, pages 2144--2160. {SIAM}, 2020.

\bibitem[BSV21]{BSV21}
Vishwas Bhargava, Shubhangi Saraf, and Ilya Volkovich.
\newblock \href {http://dx.doi.org/10.1145/3406325.3451096} {Reconstruction
  algorithms for low-rank tensors and depth-3 multilinear circuits}.
\newblock In {\em \STOC{2021}}, pages 809--822. {ACM}, 2021.

\bibitem[BT88]{BOT88}
Michael Ben{-}Or and Prasoon Tiwari.
\newblock \href {http://dx.doi.org/10.1145/62212.62241} {A Deterministic
  Algorithm for Sparse Multivariate Polynomial Interpolation (Extended
  Abstract)}.
\newblock In {\em \STOC{1988}}, pages 301--309. {ACM}, 1988.

\bibitem[Car06]{Carlini06}
Enrico Carlini.
\newblock \href {http://dx.doi.org/10.1007/978-3-540-33275-6_15} {Reducing the
  number of variables of a polynomial}.
\newblock In {\em Algebraic Geometry and Geometric Modeling}, pages 237--247,
  2006.

\bibitem[CLO07]{CLO07}
David~A.\ Cox, John~B.\ Little, and Donal O'Shea.
\newblock \href {http://dx.doi.org/10.1007/978-0-387-35651-8} {{\em Ideals,
  Varieties and Algorithms}}.
\newblock Undergraduate texts in mathematics. Springer, 2007.

\bibitem[DS07]{DS07}
Zeev Dvir and Amir Shpilka.
\newblock \href {http://dx.doi.org/10.1137/05063605X} {Locally Decodable Codes
  with Two Queries and Polynomial Identity Testing for Depth 3 Circuits}.
\newblock {\em {SIAM} J. Comput.}, 36(5):1404--1434, 2007.
\newblock \pSTOC{2005}.

\bibitem[FGS18]{FGS18}
Michael~A. Forbes, Sumanta Ghosh, and Nitin Saxena.
\newblock \href {http://dx.doi.org/10.4230/LIPIcs.ICALP.2018.54} {Towards
  Blackbox Identity Testing of Log-Variate Circuits}.
\newblock In {\em \pICALP{2018}}, volume 107 of {\em LIPIcs}, pages
  54:1--54:16. Schloss Dagstuhl - Leibniz-Zentrum f{\"{u}}r Informatik, 2018.

\bibitem[FS13]{FS13}
\mfbiberr{toupdate(FS13): journal}Michael~A. Forbes and Amir Shpilka.
\newblock \href {http://dx.doi.org/10.1109/FOCS.2013.34} {Quasipolynomial-Time
  Identity Testing of Non-commutative and Read-Once Oblivious Algebraic
  Branching Programs}.
\newblock In {\em \FOCS{2013}}, pages 243--252, 2013.
\newblock \farXiv{1209.2408}.

\bibitem[FSS14]{FSS14}
Michael~A. Forbes, Ramprasad Saptharishi, and Amir Shpilka.
\newblock \href {http://dx.doi.org/10.1145/2591796.2591816} {Hitting sets for
  multilinear read-once algebraic branching programs, in any order}.
\newblock In {\em \STOC{2014}}, pages 867--875, 2014.

\bibitem[GG20]{GG20}
Zeyu Guo and Rohit Gurjar.
\newblock \href {http://dx.doi.org/10.4230/LIPIcs.APPROX/RANDOM.2020.4}
  {Improved Explicit Hitting-Sets for ROABPs}.
\newblock In {\em \RANDOM{2020}}, volume 176 of {\em LIPIcs}, pages 4:1--4:16,
  2020.

\bibitem[GKKS16]{GKKS16}
Ankit Gupta, Pritish Kamath, Neeraj Kayal, and Ramprasad Saptharishi.
\newblock \href {http://dx.doi.org/10.1137/140957123} {Arithmetic Circuits: {A}
  Chasm at Depth 3}.
\newblock {\em {SIAM} J. Comput.}, 45(3):1064--1079, 2016.

\bibitem[GKL11]{GKL11}
Ankit Gupta, Neeraj Kayal, and Satyanarayana~V. Lokam.
\newblock \href {http://dx.doi.org/10.1109/FOCS.2011.70} {Efficient
  Reconstruction of Random Multilinear Formulas}.
\newblock In {\em \FOCS{2011}}, pages 778--787. {IEEE} Computer Society, 2011.

\bibitem[GKL12]{GKL12}
Ankit Gupta, Neeraj Kayal, and Satyanarayana~V. Lokam.
\newblock \href {http://dx.doi.org/10.1145/2213977.2214035} {Reconstruction of
  depth-4 multilinear circuits with top fan-in 2}.
\newblock In {\em \STOC{2012}}, pages 625--642. {ACM}, 2012.

\bibitem[GKS20]{GKS20}
Ankit Garg, Neeraj Kayal, and Chandan Saha.
\newblock \href {http://dx.doi.org/10.1109/FOCS46700.2020.00087} {Learning sums
  of powers of low-degree polynomials in the non-degenerate case}.
\newblock In {\em \FOCS{2020}}, pages 889--899. {IEEE}, 2020.

\bibitem[H{\aa}s90]{Hastad90}
Johan H{\aa}stad.
\newblock \href {http://dx.doi.org/10.1016/0196-6774(90)90014-6} {Tensor Rank
  is NP-Complete}.
\newblock {\em J. Algorithms}, 11(4):644--654, 1990.

\bibitem[Kay11]{Kayal11}
Neeraj Kayal.
\newblock \href {http://dx.doi.org/10.1137/1.9781611973082.108} {Efficient
  algorithms for some special cases of the polynomial equivalence problem}.
\newblock In {\em \SODA{2011}}, pages 1409--1421. {SIAM}, 2011.

\bibitem[Kay12]{Kayal12}
Neeraj Kayal.
\newblock \href {http://dx.doi.org/10.1145/2213977.2214036} {Affine projections
  of polynomials: extended abstract}.
\newblock In {\em \STOC{2012}}, pages 643--662. {ACM}, 2012.

\bibitem[KMSV13]{KarninMSV13}
Zohar~Shay Karnin, Partha Mukhopadhyay, Amir Shpilka, and Ilya Volkovich.
\newblock \href {http://dx.doi.org/10.1137/110824516} {Deterministic Identity
  Testing of Depth-4 Multilinear Circuits with Bounded Top Fan-in}.
\newblock {\em {SIAM} J. Comput.}, 42(6):2114--2131, 2013.

\bibitem[KNS19]{KNS19}
Neeraj Kayal, Vineet Nair, and Chandan Saha.
\newblock \href {http://dx.doi.org/10.1007/s00037-019-00189-0} {Average-case
  linear matrix factorization and reconstruction of low width algebraic
  branching programs}.
\newblock {\em Comput. Complex.}, 28(4):749--828, 2019.

\bibitem[Koi12]{K12b}
Pascal Koiran.
\newblock \href {http://dx.doi.org/10.1016/j.tcs.2012.03.041} {Arithmetic
  Circuits: The Chasm at Depth Four Gets Wider}.
\newblock {\em Theoretical Computer Science}, 448:56--65, 2012.
\newblock Pre-print available at \href {http://arxiv.org/abs/1006.4700}
  {\path{arXiv:1006.4700}}.

\bibitem[KS01]{KS01}
Adam Klivans and Daniel~A. Spielman.
\newblock \href {http://dx.doi.org/10.1145/380752.380801} {{Randomness
  efficient identity testing of multivariate polynomials}}.
\newblock In {\em \STOC{2001}}, pages 216--223, 2001.

\bibitem[KS08]{KarninShpilka08}
Zohar~Shay Karnin and Amir Shpilka.
\newblock \href {http://dx.doi.org/10.1109/CCC.2008.15} {Black Box Polynomial
  Identity Testing of Generalized Depth-3 Arithmetic Circuits with Bounded Top
  Fan-In}.
\newblock In {\em Proceedings of the 23rd Annual {IEEE} Conference on
  Computational Complexity, {CCC} 2008, 23-26 June 2008, College Park,
  Maryland, {USA}}, pages 280--291. {IEEE} Computer Society, 2008.

\bibitem[KS09a]{KarninShpilka09}
Zohar~Shay Karnin and Amir Shpilka.
\newblock \href {http://dx.doi.org/10.1109/CCC.2009.18} {Reconstruction of
  Generalized Depth-3 Arithmetic Circuits with Bounded Top Fan-in}.
\newblock In {\em \CCC{2009}}, pages 274--285. {IEEE} Computer Society, 2009.

\bibitem[KS09b]{KS09a}
Neeraj Kayal and Shubhangi Saraf.
\newblock \href {http://dx.doi.org/10.1109/FOCS.2009.67} {{Blackbox polynomial
  identity testing for depth-$3$ circuits}}.
\newblock In {\em \FOCS{2009}}, 2009.

\bibitem[KT90]{KT90}
Erich Kaltofen and Barry~M. Trager.
\newblock \href {http://dx.doi.org/10.1016/S0747-7171(08)80015-6} {Computing
  with Polynomials Given By Black Boxes for Their Evaluations: Greatest Common
  Divisors, Factorization, Separation of Numerators and Denominators}.
\newblock {\em J. Symb. Comput.}, 9(3):301--320, 1990.

\bibitem[PSV24]{PelegSV24}
Shir Peleg, Amir Shpilka, and Ben~Lee Volk.
\newblock \href {http://dx.doi.org/10.4230/LIPICS.ITCS.2024.87} {Tensor
  Reconstruction Beyond Constant Rank}.
\newblock In {\em 15th Innovations in Theoretical Computer Science Conference,
  {ITCS} 2024, Berkeley, CA, USA}, volume 287 of {\em LIPIcs}, pages
  87:1--87:20. Schloss Dagstuhl - Leibniz-Zentrum f{\"{u}}r Informatik, 2024.

\bibitem[Shi16]{Shitov16}
Yaroslav Shitov.
\newblock \href {https://arxiv.org/abs/1611.01559} {How hard is the tensor
  rank?}
\newblock {\em arXiv preprint arXiv:1611.01559}, 2016.

\bibitem[Sin16]{Sinha16}
Gaurav Sinha.
\newblock \href {http://dx.doi.org/10.4230/LIPIcs.CCC.2016.31} {Reconstruction
  of Real Depth-3 Circuits with Top Fan-In 2}.
\newblock In {\em \CCC{2016}}, volume~50 of {\em LIPIcs}, pages 31:1--31:53,
  2016.

\bibitem[Sin22]{Sinha22}
Gaurav Sinha.
\newblock \href {http://dx.doi.org/10.4230/LIPIcs.ITCS.2022.118} {Efficient
  Reconstruction of Depth Three Arithmetic Circuits with Top Fan-In Two}.
\newblock In {\em \ITCS{2022}}, volume 215 of {\em LIPIcs}, pages
  118:1--118:33. Schloss Dagstuhl - Leibniz-Zentrum f{\"{u}}r Informatik, 2022.

\bibitem[SS11]{SaxenaS11}
Nitin Saxena and C.~Seshadhri.
\newblock \href {http://dx.doi.org/10.1137/090770679} {An Almost Optimal Rank
  Bound for Depth-3 Identities}.
\newblock {\em {SIAM} J. Comput.}, 40(1):200--224, 2011.

\bibitem[SS12]{SS12}
Nitin Saxena and C.~Seshadhri.
\newblock \href {http://dx.doi.org/10.1137/10848232} {Blackbox Identity Testing
  for Bounded Top-Fanin Depth-3 Circuits: The Field Doesn't Matter}.
\newblock {\em {SIAM} J. Comput.}, 41(5):1285--1298, 2012.
\newblock \pSTOC{2011}.

\bibitem[SS13]{SaxenaS13}
Nitin Saxena and C.~Seshadhri.
\newblock \href {http://dx.doi.org/10.1145/2528403} {From sylvester-gallai
  configurations to rank bounds: Improved blackbox identity test for depth-3
  circuits}.
\newblock {\em J. {ACM}}, 60(5):33:1--33:33, 2013.

\bibitem[Swe18]{Swernofsky18}
Joseph Swernofsky.
\newblock \href {http://dx.doi.org/10.4230/LIPIcs.APPROX-RANDOM.2018.26}
  {Tensor Rank is Hard to Approximate}.
\newblock In {\em Approximation, Randomization, and Combinatorial Optimization.
  Algorithms and Techniques, {APPROX/RANDOM} 2018}, volume 116 of {\em LIPIcs},
  pages 26:1--26:9, 2018.

\bibitem[SY10]{SY10}
Amir Shpilka and Amir Yehudayoff.
\newblock \href {http://dx.doi.org/10.1561/0400000039} {Arithmetic Circuits: A
  survey of recent results and open questions}.
\newblock {\em Foundations and Trends in Theoretical Computer Science},
  5:207--388, March 2010.

\bibitem[Tav15]{T15}
S{\'{e}}bastien Tavenas.
\newblock \href {http://dx.doi.org/10.1016/j.ic.2014.09.004} {Improved bounds
  for reduction to depth 4 and depth 3}.
\newblock {\em Inf. Comput.}, 240:2--11, 2015.
\newblock \pMFCS{2013}.

\end{thebibliography}
		
		\appendix

	\end{document}